%% file: main.tex
\documentclass[acmsmall]{acmart}

\usepackage{amsmath}
\usepackage{amsthm}
\usepackage{amsfonts}
\usepackage[linesnumbered,ruled,vlined]{algorithm2e}
\usepackage{xspace}
\usepackage{xcolor}
\usepackage{enumitem}
\usepackage{array}
\usepackage{stmaryrd}
\usepackage{graphicx}
\usepackage{float}
\usepackage{pifont}
\usepackage{listings}
\usepackage{subcaption}
\usepackage{multicol}
\usepackage{multirow}
\usepackage{enumitem}

\setlist[itemize]{leftmargin=0.5cm}
\setlist[enumerate]{leftmargin=0.5cm}

\definecolor{eclipseBlue}{RGB}{42,0.0,255}
\definecolor{eclipseGreen}{RGB}{63,127,95}
\definecolor{eclipsePurple}{RGB}{127,0,85}

\newcommand{\revise}[1]{\noindent\textcolor{black}{#1}}
\newcommand{\crc}[1]{\noindent\textcolor{black}{#1}}

\lstdefinelanguage{pgq}
{
  morekeywords={
    NODE,
    EDGE,
    TYPE,
    IMPORTS,
    OPTIONAL,
    STRICT,
    LOOSE,
    OPEN,
    CLOSED,
    ABSTRACT,
    CREATE,
    GRAPH,
    FOR,
    WITHIN,
    MATCH,
    WHERE,
    NOT,
    OR,
    AND,
    EXISTS,
    RETURN,
    IN,
    IS,
    NULL,
    ORDER,
    BY,
    INSERT,
    SET,
    REMOVE,
    DROP,
    DELETE,
    FINISH,
    COLUMNS,
    GRAPH_TABLE,
    SELECT,
    DETACH,
    STRING,
    VARCHAR,
    BOOLEAN,
    BOOL,
    SIGNED,
    INTEGER,
    INT,
    FLOAT,
    GRAPH,
    TYPE,
    ALTER,
    ROLLBACK,
    COMMIT,
    TRANSACTION,
    START,
    SESSION,
    USE,
    GROUP,
    VALUE,
    COUNT,
    CALL,
    ANY,
    IMPLIES,
    SCHEMA,
    AT,
    AS,
    FROM,
    GRAPH\_TABLE,
    INTERSECT,
    \$,
    KEY,
    TABLES,
    TABLE,
    SOURCE,
    DESTINATION,
    REFERENCE,
    PROPERTY,
    PROPERTIES,
    VERTEX,
    JOIN,
    ON
  },
  sensitive=false, 
  morecomment=[l]{//}, 
  morecomment=[s]{/*}{*/}, 
  morestring=[b]" 
}

\newtheorem{definition}{Definition}
\newtheorem{lemma}{Lemma}
\newtheorem{theorem}{Theorem}
\newtheorem{example}{Example}
\newtheorem{remark}{Remark}

\AtBeginDocument{%
  }

\setcopyright{acmlicensed}
\acmJournal{PACMMOD}
\acmYear{2024} \acmVolume{2} \acmNumber{6 (SIGMOD)} \acmArticle{252} \acmMonth{12}\acmDOI{10.1145/3698828}

\input{command}

\begin{document}

\title{Towards a Converged Relational-Graph Optimization Framework}

\author{Yunkai Lou}
\affiliation{%
  \institution{Alibaba Group}
  \city{Hangzhou}
  \country{China}
}
\email{louyunkai.lyk@alibaba-inc.com}

\author{Longbin Lai}
\affiliation{%
  \institution{Alibaba Group}
  \city{Hangzhou}
  \country{China}
}
\email{longbin.lailb@alibaba-inc.com}

\author{Bingqing Lyu}
\affiliation{%
  \institution{Alibaba Group}
  \city{Hangzhou}
  \country{China}
}
\email{bingqing.lbq@alibaba-inc.com}

\author{Yufan Yang}
\affiliation{%
  \institution{Alibaba Group}
  \city{Hangzhou}
  \country{China}
}
\email{xiaofan.yyf@alibaba-inc.com}

\author{Xiaoli Zhou}
\affiliation{%
  \institution{Alibaba Group}
  \city{Hangzhou}
  \country{China}
}
\email{yihe.zxl@alibaba-inc.com}

\author{Wenyuan Yu}
\affiliation{%
  \institution{Alibaba Group}
  \city{Hangzhou}
  \country{China}
}
\email{wenyuan.ywy@alibaba-inc.com}

\author{Ying Zhang}
\affiliation{%
  \institution{Zhejiang Gongshang University}
  \city{Hangzhou}
  \country{China}
}
\email{ying.zhang@zjgsu.edu.cn}

\author{Jingren Zhou}
\affiliation{%
  \institution{Alibaba Group}
  \city{Hangzhou}
  \country{China}
}
\email{jingren.zhou@alibaba-inc.com}

\renewcommand{\shortauthors}{Yunkai Lou et al.} 

\input{abstract}



\received{April 2024}
\received[revised]{July 2024}
\received[accepted]{August 2024}

\maketitle

\input{sec-intro}
\input{sec-preliminaries}
\input{sec-match-opt}
\input{sec-framwork}
\input{sec-exp-2}
\input{sec-relatedwork}
\input{sec-conclusion}

\bibliographystyle{ACM-Reference-Format}
\bibliography{sample}

\end{document}
\endinput

%% file: command.tex
\newcommand{\kw}[1]{{\ensuremath {\mathsf{#1}}}\xspace}

\newcommand{\name}{\kw{RelGo}}

\newcommand{\relgohash}{\kw{RelGoHash}}
\newcommand{\relgomj}{\kw{RelGoNoEI}}
\newcommand{\relgonofi}{\kw{RelGoNoRule}}

\newcommand{\lab}{\mathcal{\ell}}

\newcommand{\id}{\kw{id}}

\newcommand{\glogue}{\kw{GLogue}}
\newcommand{\glogs}{\kw{GLogS}}

\newcommand{\todo}[1]{\textcolor{red}{$\Rightarrow$#1}}
\newcommand{\modify}[1]{\textcolor{black}{#1}}

\newcommand{\kk}[1]{\texttt{#1}}
\newcommand{\scan}{{\kk{SCAN}}}
\newcommand{\expandedge}{{\kk{EXPAND}\_\kk{EDGE}}}
\newcommand{\getvertex}{{\kk{GET}\_\kk{VERTEX}}}
\newcommand{\expandvertex}{{\kk{EXPAND}}}
\newcommand{\expand}{{\kk{EXPAND}}}

\newcommand{\expandintersect}{\kk{EXPAND}\_\kk{INTERSECT}}

\newcommand{\scangraphtable}{{\kk{SCAN}\_\kk{GRAPH}\_\kk{TABLE}}}
\newcommand{\hashjoin}{{\kk{HASH}\_\kk{JOIN}}}

\newcommand{\filterrule}{\kw{FilterIntoMatchRule}}

\newcommand{\joinfuserule}{\kw{TrimAndFuseRule}}

\newcommand{\rgmapping}{\kw{RGMapping}}
\newcommand{\pattern}{\mathcal{P}}
\newcommand{\matching}{\mathcal{M}}
\newcommand{\cost}{\kw{Cost}}
\newcommand{\searchgraph}{\mathbb{G}}
\newcommand{\decompnum}{\mathcal{H}}
\newcommand{\decompnumsq}[1]{\mathcal{H}^{#1}}
\newcommand{\spj}{\kw{SPJ}}
\newcommand{\spjm}{\kw{SPJM}}
\newcommand{\mmc}{\kw{MMC}}
\newcommand{\mmcs}{\kw{MMCs}}
\newcommand{\EVjoin}{\kw{EVJoin}}
\newcommand{\adj}{{\ensuremath \mathcal{N}}}
\newcommand{\relation}[1]{{\ensuremath R_{#1}}}
\newcommand{\relationx}[2]{{\ensuremath R^{#1}_{#2}}}

\newcommand{\gproject}{\ensuremath \widehat{\pi}_{A*}}
\newcommand{\gjoin}{\ensuremath \widehat{\Join}}
\newcommand{\evjoin}{\ensuremath \Join_{\epsilon\nu}}
\newcommand{\rowid}{\kw{rid}}
\newcommand{\constraints}{\ensuremath \Psi}

\newcommand{\ot}{\kw{OT}}

\long\def\comment#1{}

\newcommand{\stitle}[1]{\noindent{\underline{ #1}}}

\newcommand{\reffig}[1]{Fig.~\ref{fig:#1}}
\newcommand{\refsec}[1]{Sec.~\ref{sec:#1}}
\newcommand{\reftable}[1]{Table~\ref{tab:#1}}

\newcommand{\refeq}[1]{Eq.~\ref{eq:#1}}
\newcommand{\refdef}[1]{Def.~\ref{def:#1}}
\newcommand{\refthm}[1]{Theorem~\ref{thm:#1}}
\newcommand{\reflem}[1]{Lemma~\ref{lem:#1}}
\newcommand{\refrem}[1]{Remark~\ref{rem:#1}}

\newcommand{\refex}[1]{Example~\ref{ex:#1}}

\newcommand{\revisewy}[1]{\textcolor{black}{#1}}

%% file: abstract.tex
\begin{abstract}
\revisewy{The recent ISO SQL:2023 standard adopts SQL/PGQ (Property Graph Queries), facilitating graph-like querying within relational databases. This advancement, however, underscores a significant gap in how to effectively optimize SQL/PGQ queries within relational database systems. To address this gap, we extend the foundational} \spj (Select-Project-Join) queries to \spjm queries, which include an additional matching operator for representing graph pattern matching in SQL/PGQ. Although \spjm queries can be converted to \spj queries and optimized using existing relational query optimizers, our analysis shows that such a graph-agnostic method fails to benefit from graph-specific optimization techniques found in the literature. To address this issue, we develop a converged relational-graph optimization framework called \name for optimizing \spjm queries, leveraging joint efforts from both relational and graph query optimizations. Using DuckDB as the underlying relational execution engine, our experiments show that \name can generate efficient execution plans for \spjm queries. On well-established benchmarks, these plans exhibit an average speedup of $21.90\times$ compared to those produced by the graph-agnostic optimizer.

\end{abstract}

%% file: sec-intro.tex
\section{Introduction}
\label{sec:introduction}

In the realms of data management and analytics, relational databases have long been the bedrock of structured data storage and retrieval, empowering a plethora of applications. 
The ubiquity of these databases has been supported by the advent of Structured Query Language (SQL)~\cite{chamberlin1974sequel}, a standardized language that has been adopted widely by various relational database management systems for managing data through schema-based operations. 

Despite its considerable success and broad adoption, SQL has its limitations, particularly when it comes to representing and querying intricately linked data. Consider, for instance, the relational tables of \kk{Person} and \kk{Knows}, the latter symbolizing a many-to-many relationship between instances of the former. Constructing a SQL query to retrieve a group of four persons who are all mutually acquainted is not a straightforward endeavor, potentially leading to a cumbersome and complex SQL expression.

In comparison, such a scenario could be succinctly addressed using graph query languages such as Cypher~\cite{opencypher}, where queries are expressed as graph pattern matching. This discrepancy between the relational and graph querying paradigms has given rise to the innovative SQL/Property Graph Queries (SQL/PGQ), an extension \revisewy{formally adopted in the ISO SQL:2023 standard}~\cite{sql-pgq}. SQL/PGQ is designed to amalgamate the extensive capabilities of SQL with the inherent benefits of graph pattern matching. With SQL/PGQ, it is now possible to define and query graphs within SQL expressions, transforming otherwise complex relational queries -- characterized by multiple joins -- into simpler and more intuitive graph queries.


\begin{figure}[h]
    \centering
    \includegraphics[width=.5\linewidth]{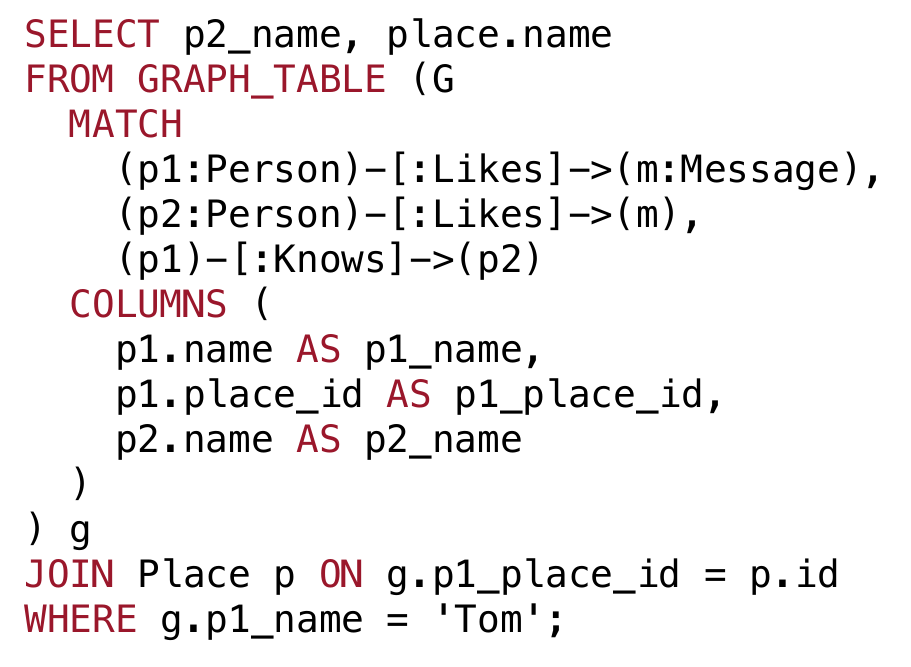}
    \vspace*{-1mm}
    \caption{An example of SQL/PGQ query.}
    \label{fig:intro-query}
    \vspace*{-1mm}
\end{figure}

\begin{example}
    \label{ex:introduction:sqlpgq}
    Consider the four relational tables in the database: \kk{Person(id, name, place\_id)}, \kk{Message(id, content, date)}, \kk{Like(p\_id, m\_id, date)}, and \kk{Place(id, name)}. Using SQL/PGQ, a property graph $G$ is articulated as a \lstinline{GRAPH_TABLE}, established on the basis of the first three tables. In this mapping, rows from \kw{Person} and \kw{Message} are interpreted as vertices with labels ``Person'' and ``Message'' respectively, while rows from Like represent edges with the label ``Likes''. \revise{This mapping process will be elaborated as \rgmapping in \refsec{data-model}}. An SQL/PGQ query to discover the friends of a person named ``Tom'' and the place they live in, where ``Tom'' and friends share an affinity for the same message, can be formulated as shown in \reffig{intro-query}.
    \revisewy{In graph $G$, a {\sc Graph Pattern Matching} is employed to decode the intricate relationships between persons and messages. Upon executing the pattern matching, a \lstinline{COLUMNS} clause projects the results into a tabular format, enumerating essential attributes. Then the {\sc Relational}~\lstinline{JOIN} is performed on resultant table \lstinline{g} and \kk{Place} table to obtain the place's name.}
\end{example}

The SQL/PGQ standardization, while a significant leap forward in the realm of relational databases, primarily addresses language constructs. A discernible gap exists in the theoretical landscape, particularly in analyzing, transforming, and optimizing SQL/PGQ queries with hybrid relational and graph semantics. 

Relational query optimization has historically leaned on the \spj (selection-projection-join) skeleton~\cite{spj,Chaudhuri98}, which provides a systematic approach for analyzing query complexity~\cite{IbarakiK84,ChatterjiEGY02} , devising heuristic optimization rules~\cite{Chaudhuri99heuristics,goldsteinheuristics}, and
computing optimal join order~\cite{Haffnerjoinorder,chenjoinorder}. Recently, graph techniques have been introduced to optimize relational queries~\cite{wanderjoin,Haffnerjoinorder,gqfast,graindb}. 
In particular, GRainDB~\cite{graindb} introduced a predefined join operator that materializes the adjacency list (rows) of vertices, enabling more efficient join execution. While these techniques can be empowered by graph techniques, they target purely relational query rather than the relational-graph hybrid query of SQL/PGQ.

In parallel to relational query optimization, significant strides have been made in optimizing graph pattern matching.
A common practice is to
leverage join-based techniques to optimize the query~\cite{lai2019distributed,lai2015scalable,ammar2018distributed,huge}. Scalable join algorithms, such as binary-join~\cite{lai2015scalable}, worst-case optimal (\revise{abbr.~wco}) join~\cite{ammar2018distributed}, and their hybrid variants~\cite{mhedhbi2019optimizing,huge,GLogS}, have been proposed for solving the problem over large-scale graphs. However, despite the effectiveness of these techniques for pattern matching on graphs, they cannot be directly applied to relational databases due to the inherent differences in data models.

In this paper, we propose the first converged optimization framework, \name, that optimizes relational-graph hybrid queries in a relational database, in response to the advent of SQL/PGQ. A straightforward implementation~\cite{DuckPGQ,DuckPGQ-VLDB,apache-age} can involve directly transforming the graph component in SQL/PGQ queries into relational operations, allowing the entire query to be optimized and executed in any existing relational engine. While we contribute to building the theory to make such a transformation workable, this \emph{graph-agnostic} optimization approach suffers from several issues, including graph-unaware join orders, suboptimal join plans, and increased search space, as will be discussed in \refsec{graph-aware}.

\begin{table}[t]
    \small
    \centering
   \caption{Frequently used notations.}
   \label{tab:notations}
    \scalebox{0.9}{
        \begin{tabular}{|m{0.27\linewidth}|m{0.74\linewidth}|}
            \hline
            Notation                 & Definition                                                                                       \\
            \hline
            $R$    & a relation or relational table   \\
            \hline
            $\tau$ and $\tau.attr$    & a tuple in a relation, and the value of an attribute of $\tau$ \\
            \hline
            $G(V,E)$         & a property graph with $V$ and $E$                                                          \\
            \hline
            $\pattern(V, E)$   & a pattern graph with $V$ and $E$  \\
            \hline
            $\id(\epsilon)$, $\lab(\epsilon)$, $\epsilon$.attr  & the identifier, label, and the value of given attribute of a graph element $\epsilon$ \\
            \hline
            $\adj(u)$ and $\adj^E(u)$           & neighbors and adjacent edges of $u$           \\
            \hline
            $GR$   & a graph relation     \\
            \hline
            $\matching(GR, \pattern)$, $\matching(\pattern)$  & matching $\pattern$ on a graph relation $GR$ or a graph $G$ \\
            \hline
            $\pi_{A}$, $\sigma_\constraints$, $\Join$  & projection, selection, and join operators over relations \\
            \hline
            $\gproject$, $\gjoin$  & projection and join operators over graph relations \\
            \hline
            $\lambda^s_\lab(e)$, $\lambda^t_\lab(e)$  & the total functions for mapping tuples in an edge relation to source and target vertex relations \\
            \hline
        \end{tabular}
    }
\end{table}

To address these challenges, \name is proposed to leverage the strengths of both relational and graph query optimization techniques. Building upon the foundation of \spj queries, we introduce the \spjm query skeleton, which extends \spj with a matching operator to represent graph queries. We adapt state-of-the-art graph optimization techniques, such as the decomposition method~\cite{huge} and the cost-based optimizer~\cite{GLogS}, to the relational context, effectively producing worst-case optimal graph subplans for the matching operator. To facilitate efficient execution of the matching operator, we introduce graph index inspired by GRainDB's predefined join~\cite{graindb}, based on which graph-based physical operations are implemented. The relational part of the query, together with the optimized graph subplans encapsulated within a special operator called \scangraphtable, is then optimized using standard relational optimizers. Finally, we incorporate heuristic rules, such as \filterrule, to handle cases unique to \spjm queries that involve the interplay between relational and graph components.


We have made the following contributions in this paper:
\begin{enumerate}
\item \revise{We map relational data models to property graph models as specified by SQL/PGQ using \rgmapping. Based on \rgmapping, we introduce a new query skeleton called \spjm, which is designed to better analyze relational-graph hybrid queries}. \hfill(\refsec{preliminaries})

\item We construct the theory for transforming any \spjm query into an \spj query. Such a graph-agnostic approach enables existing relational databases to handle \spjm queries without low-level modifications. We also formally prove that the search space
of the graph-agnostic approach can be exponentially larger than our solution. \hfill(\refsec{optimizing-matching-operator})

\item We introduce \name, a converged optimization framework that leverages the strengths of both relational and graph query optimization techniques to optimize \spjm queries. \name adapts state-of-the-art graph optimization techniques to the relational context, and implements graph-based physical operations based on graph index for efficient query execution. \hfill(\refsec{optimizations})

\item We develop \name~ by integrating it with the industrial relational optimization framework, Calcite~\cite{calcite}, and employing DuckDB~\cite{duckdb} for execution runtime. We conducted extensive experiments to evaluate its performance. The results on the LDBC Social Network Benchmark~\cite{ldbc_snb} indicate that \name significantly surpasses the performance of the graph-agnostic baseline, with an average speedup of \revise{$21.9\times$}, and \revise{$5.4\times$} even after graph index is enabled for the baseline. \hfill(\refsec{evaluation})

\end{enumerate}

This paper is organized in the order of the contributions. We survey related work in \refsec{related-work}
and conclude the paper in \refsec{conclusions}.

\comment{
Relational databases have been a research hotspot for a long time, and the related achievements are applied in various scenarios such as finance, healthcare, and online services.
In order to manage data in relational databases conveniently and efficiently, Structured Query Language (abbr.~SQL) is presented.
Various relational database management systems implement their own versions of SQL.

Although SQL has achieved tremendous success, there are scenarios wherein utilizing SQL proves to be rather cumbersome.
For instance, given two tables \textbf{Person} and \textbf{Knows}, we are going to find all the quadruples of persons where every two persons know each other.
Then, a possible SQL expression for this query is as follows:
\begin{lstlisting}
    SELECT p1.id, p2.id, p3.id, p4.id
    FROM Person p1, Person p2, Person p3, Person p4, Knows k1, Knows k2, Knows k3, Knows k4, Knows k5, Knows k6
    WHERE k1.person1id = p1.id AND k1.person2id = p2.id
    AND k2.person1id = p1.id AND k2.person2id = p3.id
    AND k3.person1id = p1.id AND k3.person2id = p4.id
    AND k4.person1id = p2.id AND k4.person2id = p3.id
    AND k5.person1id = p2.id AND k5.person2id = p4.id
    AND k6.person1id = p3.id AND k6.person2id = p4.id;
\end{lstlisting}
Such an SQL expression is intricate and inconvenient to write.

\begin{figure}
    \centering
    \includegraphics[width=.3\linewidth]{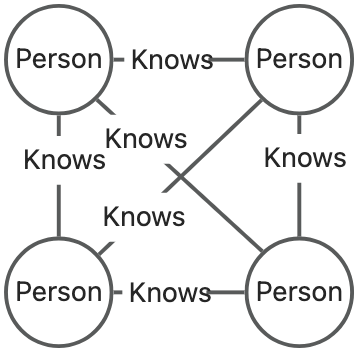}
    \caption{Pattern graph corresponding to the conditions in the SQL query.}
    \label{fig:intro-pattern}
\end{figure}

Indeed, the conditions specified within this SQL statement collectively establish the structure of a pattern graph.
As shown in Fig.~\ref{fig:intro-pattern}, there are four vertices representing four persons, and these four persons form a complete graph.
Therefore, this SQL query can be expressed as a graph query.
The corresponding expression following the grammar of Cypher \cite{opencypher} is as follows:
\begin{lstlisting}
    MATCH (p1:Person)-[:Knows]-(p2:Person)-[:Knows]-(p3:Person)-[:Knows]-(p4:Person),
        (p1)-[:Knows]-(p3), (p2)-[:Knows]-(p4)
    RETURN p1.id, p2.id, p3.id, p4.id;
\end{lstlisting}
This graph query expression is much more concise and understandable than that in SQL.
It suggests that SQL is not always the optimal choice, and sometimes employing graph queries is more advantageous.
Then, in order to combine the extensive developments made in SQL queries over the years with the benefits of graph queries, it would be helpful for SQL to support both relational and graph queries.

Following this idea, the striking SQL/Property Graph Queries (abbr.~SQL/PGQ) is proposed.
In detail, SQL/PGQ is a part of SQL 2023 and its grammar allows to define and query graphs in SQL/PGQ expressions.
Consequently, some complex relational queries (such as those containing multiple joins) can be represented as relatively simple graph queries.
Specifically, graphs in SQL/PGQ are presented as views, and vertices and edges in the graphs are represented as tables.
It implies that with SQL/PGQ, graph queries and relational queries can be expressed in one statement and optimized together for a better execution plan.
An example of an SQL/PGQ query is provided in Example \ref{example:introduction:sqlpgq}.

\begin{example}
    \label{example:introduction:sqlpgq}
    Suppose three tables, i.e., \textbf{Person, Knows, Department}, are stored in the relational database.
    With SQL/PGQ, a graph view named \textbf{friendship\_graph} is created based on tables \textbf{Person} and \textbf{Knows}.
    Specifically, rows in table \textbf{Person} represent the vertices in the graph while rows in table \textbf{Knows} represent the edges.
    Besides, the department a person belonging to is stored in table \textbf{Person} as a foreign key (i.e., \textit{dept\_id}).

    Suppose we are going to find three persons satisfying: (1) They know each other; (2) Two of them belong to the Department of Computer Science.
    Then, the corresponding SQL/PGQ query is as follows:
    \begin{lstlisting}
        SELECT pn1, pn2, pn3
        FROM Department p, GRAPH_TABLE (friendship_graph
            MATCH (p1:Person)-[:Knows]-(p2:Person)-[:Knows]-(p3:Person),
            (p1)-[:Knows]-(p3)
            COLUMNS (
                p1.name as pn1,
                p1.dept_id as dept1,
                p2.name as pn2,
                p2.dept_id as dept2,
                p3.name as pn3,
                p3.dept_id as dept3)
        ) f
        WHERE dept1 = p.dept_id
        AND dept2 = p.dept_id AND
        AND p.dept_name = 'Computer Science';
    \end{lstlisting}
    According to the first condition, the obtained three persons should form a triangle.
    It is a problem of pattern matching, and such triangles are searched for on \textbf{friendship\_graph}.
    The output of the graph query is a table (named \textbf{f}) with six columns, i.e., pn1, dept1, pn2, dept2, pn3, and dept3, representing the names of the three persons and the identifiers of the departments they belonging to.

    For the second condition, due to the existence of the foreign key, it is efficient to perform natural join between table \textbf{f} and table \textbf{Department} to obtain the ideal results.
    Please note that the outputs of graph queries are still tables, and the outputs can be treated as tables within relational queries.
\end{example}

To support SQL/PGQ queries, it is more reasonable to enhance relational databases with the capability to process graph queries.
The reason is that relational databases have been widely utilized in both academia and industry, and it is costly and impractical to migrate data from relational databases to graph databases.
After the SQL/PGQ query is parsed, relational databases need to optimize the obtained abstract syntax tree (abbr.~AST).
However, since graph queries are allowed in SQL/PGQ queries, the existing relational databases can hardly optimize ASTs with graph operators.
Intuitively, there are some possible solutions, which can be mainly categorized into four types.
The stages of query optimization that these four types of solution are involved in respectively are shown in Fig.~\ref{fig:catagory}.

\begin{figure*}
    \centering
    \begin{subfigure}[b]{0.4\linewidth}
        \centering
        \includegraphics[width=\linewidth]{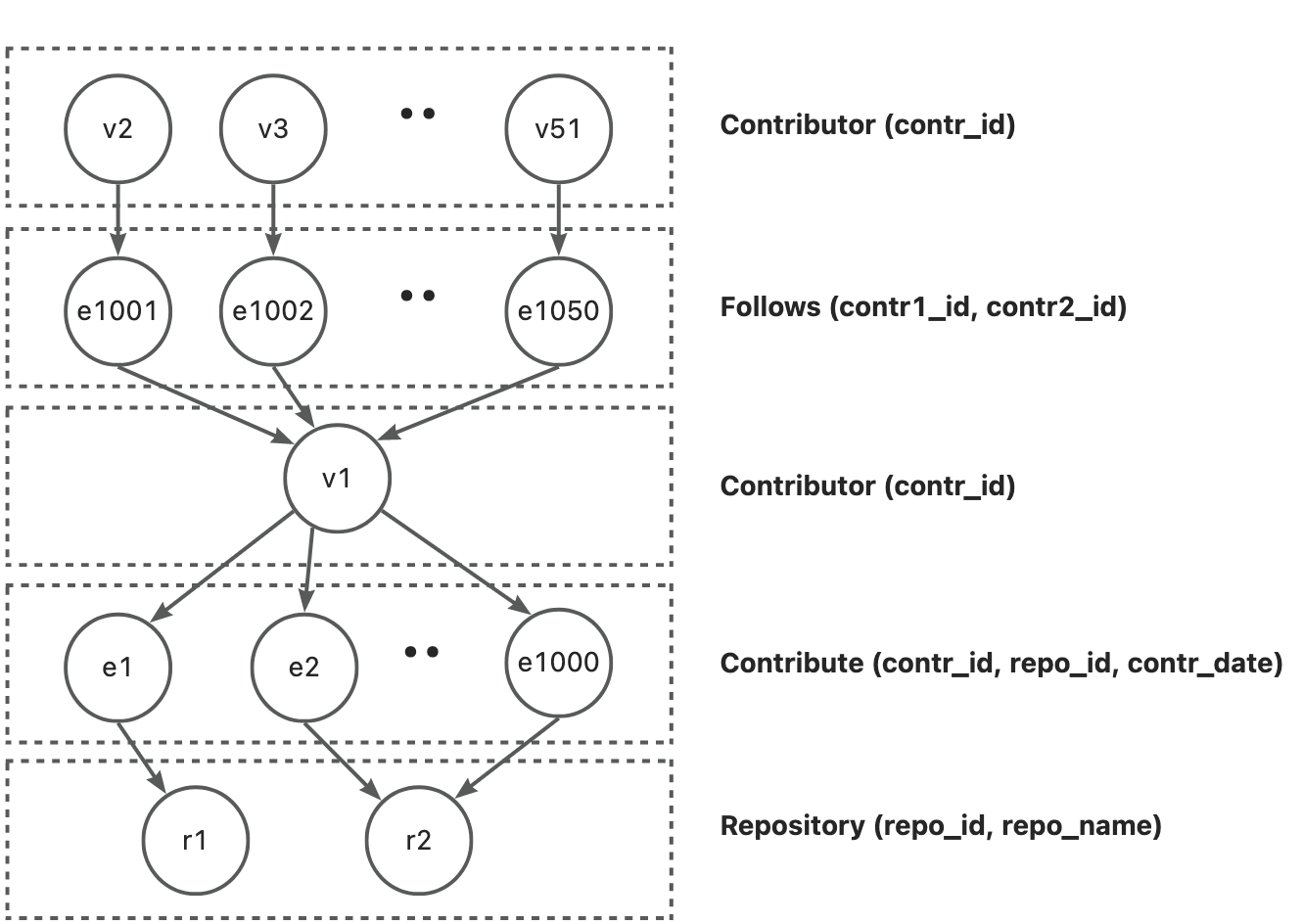}
        \caption{Relationship 1.}
        \label{fig:intro-order-case}
    \end{subfigure}
    \begin{subfigure}[b]{0.4\linewidth}
        \centering
        \includegraphics[width=\linewidth]{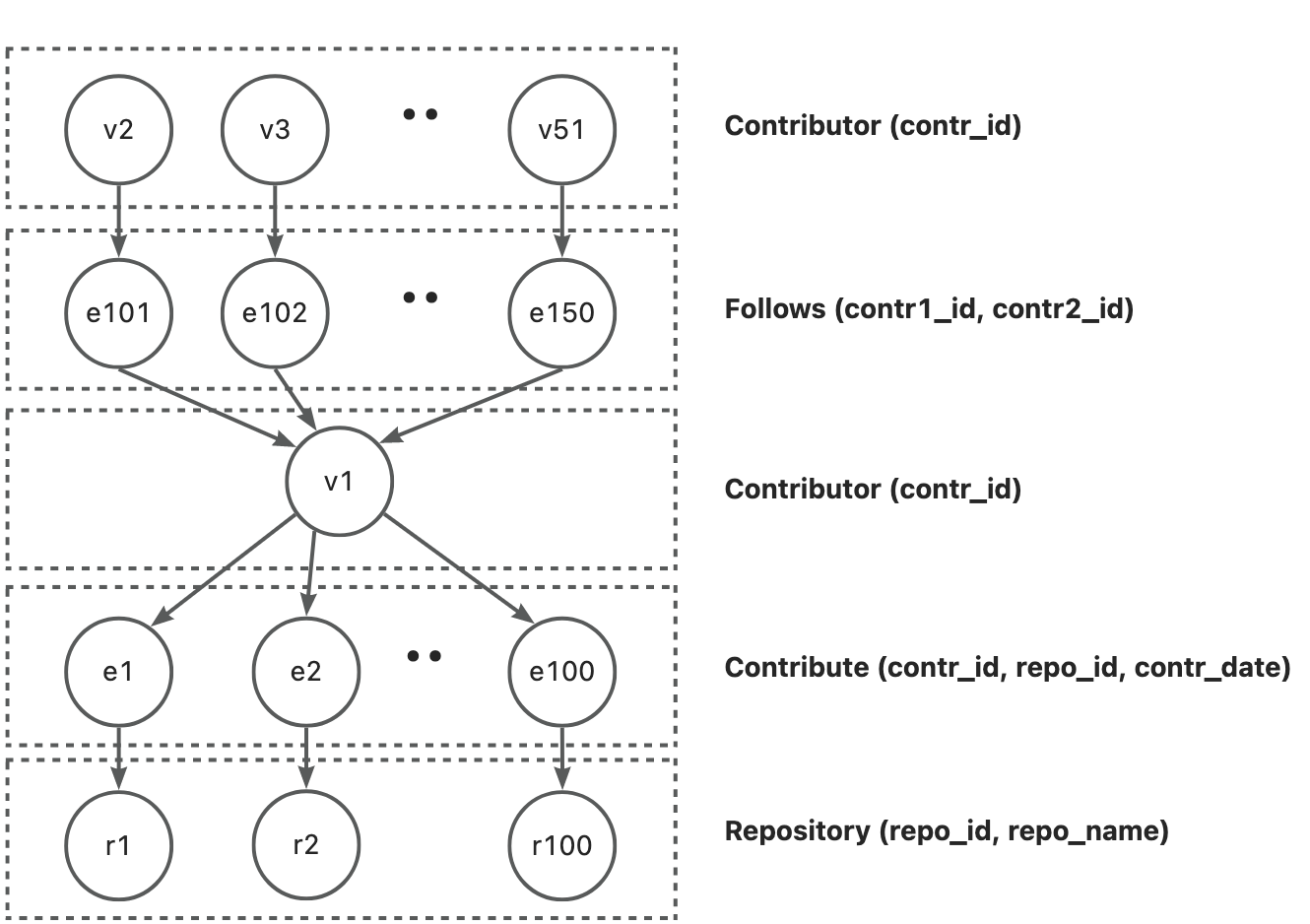}
        \caption{Relationship 2.}
        \label{fig:intro-order-case2}
    \end{subfigure}
    \caption{Graphs representing the relationships among tuples in different tables. In detail, tuples in Tables \textbf{Contributor} and \textbf{Repository} represent vertices in the graph, while those in Tables \textbf{Follows} and \textbf{Contribute} represent edges.}
    \label{fig:intro-replace-example}
\end{figure*}

\begin{figure*}
    \centering
    \includegraphics[width=.6\linewidth]{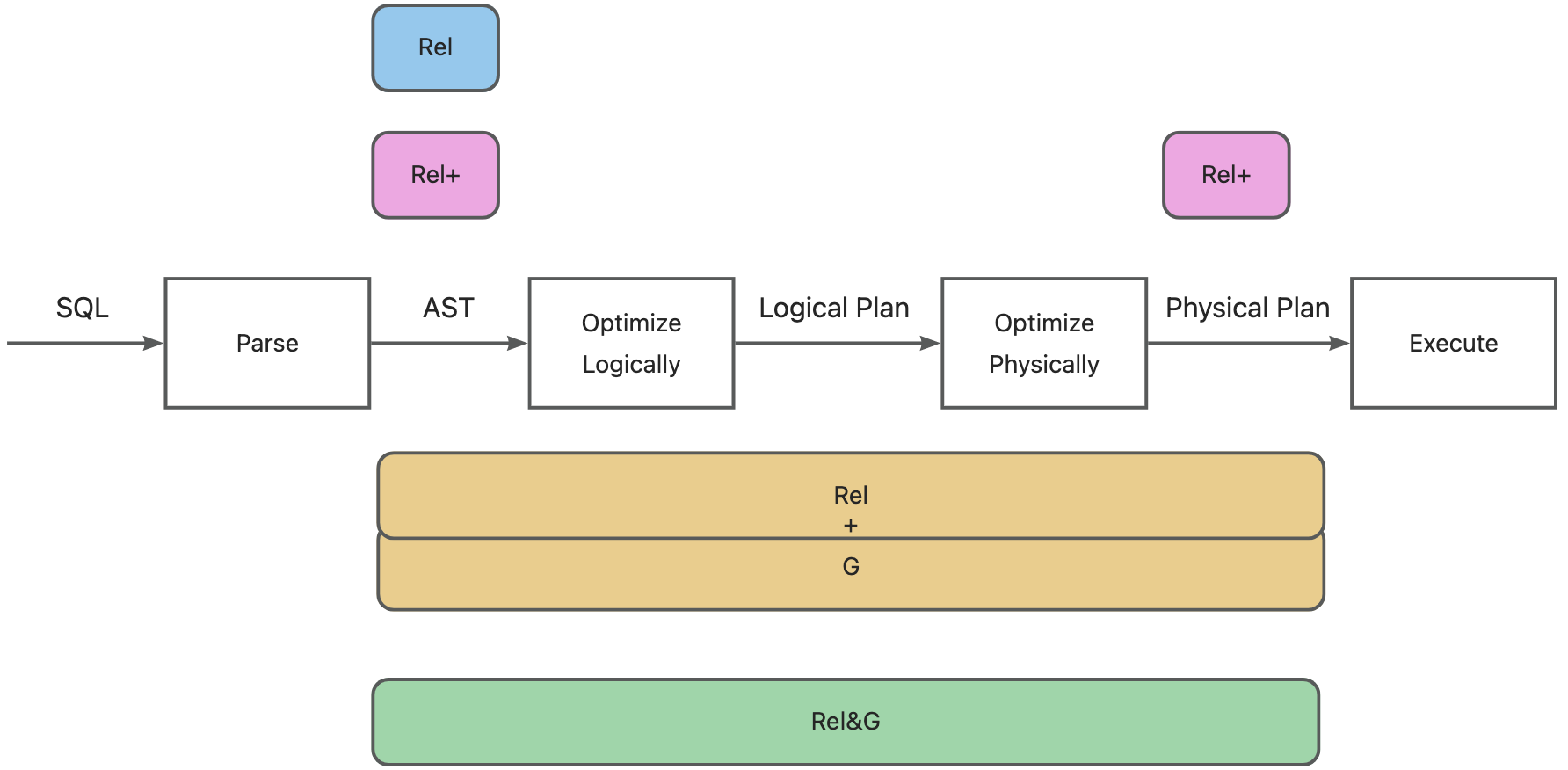}
    \caption{The stages of query optimization that the four kinds of solutions are involved in.}
    \label{fig:catagory}
\end{figure*}

\textbf{Solution 1 ($Rel$)}.
The most direct solution is to transform graph queries to relational ones, and then optimize the new queries with relational optimizers.
Apache Age \cite{apache-age} and DuckPGQ \cite{DuckPGQ,DuckPGQ-VLDB} are typical examples.
However, methods of this type only take effect in the stage of logical optimization and degrade into relational optimizers.
Therefore, they lose the chance of query optimization from the graph perspective.

\begin{example}
    Given four tables \textbf{Contributor}, \textbf{Follows}, \textbf{Contribute}, and \textbf{Repository} as shown in Fig.~\ref{fig:intro-order-case}, the relationships among their tuples are presented.
    Specifically, edge (v1, e1) means that e1.contr\_id = v1.contr\_id and edge (e1, r1) means that e1.repo\_id = r1.repo\_id.
    Moreover, suppose graph indices have been built on tables \textbf{Follows} and \textbf{Contribute}.
    Then, given a contributor, the followers of the contributor and the repositories the contributor contributes to can be directly retrieved with the indices, respectively.

    Suppose we are going to find the followers of $v1$, and the query is as follows:
    \begin{lstlisting}
        SELECT pid
        FROM GRAPH_TABLE (graph_view
            MATCH (p1:Contributor)-[:Follows]->(p2:Contributor {id: 1})
            COLUMNS (p1.id as pid)
        );
    \end{lstlisting}
    If the graph query is transformed to the corresponding relational query, then the followers of $v1$ can only be found with a join between \textbf{Contributor} and \textbf{Follows}, and another join between the resultant table and \textbf{Contributor}.
    In this process, the graph indices cannot be utilized, and the process is far from efficient.
\end{example}

\textbf{Solution 2 ($Rel^+$)}.
Methods of this type build graph indices on relational databases and introduce new operators to perform graph queries based on the indices.
However, these new operators are applied after the optimal physical plan is obtained with the relational optimizer.
As shown in Fig.~\ref{fig:catagory}, such methods are separately involved in the stages of logical optimization and physical optimization.
It means that the optimizations w.r.t.~graph operators at the logical layer and the physical layer are disjointed.
To be more specific, the cost of new operators introduced in physical optimization are unaware in logical optimization.

The optimizer in GRainDB \cite{graindb} is a representative of type $Rel^+$.
GRainDB builds RID indices on DuckDB \cite{duckdb}, and proposes two new join methods, i.e., sip-join and merge-sip-join.
In detail, sip-join gets adjacent edges of vertices or gets adjacent vertices of edges based on the RID indices, while merge-sip-join obtains the neighbors of vertices.
Since GRainDB follows the grammar of SQL, given a SQL/PGQ query, the query is transformed to the equal relational query first, and then GRainDB optimizes the query with the relational optimizer of DuckDB to obtain the optimal execution plan.
Next, GRainDB replaces some hash-joins in the plan with sip-joins and merge-sip-joins to leverage the graph indices.
It indicates that the cost-based optimization in GRainDB only finds the optimal execution plan before the graph indices are aware.
Therefore, the plan can be suboptimal after replacement.
Moreover, some efficient replacement cannot be applied w.r.t.~the obtained execution plan due to the order of joining tables representing vertices and edges.
An example is shown as follows.

\begin{example}
    Given four tables as shown in Fig.~\ref{fig:intro-order-case}, a SQL/PGQ query is as follows:
    \begin{lstlisting}
        SELECT contr_id, repo_name
        FROM GRAPH_TABLE (graph_view
            MATCH (p2:Contributor)-[f:Follows]->(p1:Contributor {contr_id: 1})-[c:Contribute]->(r:Repository)
            COLUMNS (p2.contr_id as contr_id,
                    r.repo_name as repo_name)
        );
    \end{lstlisting}
    From the perspective of a relational database (e.g., DuckDB), the best join order can be \textbf{p1$\rightarrow$f$\rightarrow$p2$\rightarrow$c$\rightarrow$r}, since tables \textbf{Follows} and \textbf{Contributor} have much smaller cardinalities than table \textbf{Contribute}.
    Then, by replacing the join operators with getNeighbor, the finally obtained execution plan is \textbf{p1$\xrightarrow{\textit{getNeighbor}}$p2$\xrightarrow{\textit{getNeighbor}}$r}.

    However, as $v_1$ has much fewer neighbors in table \textbf{Repository} than in table \textbf{Contributor}, join order \textbf{p1$\xrightarrow{\textit{getNeighbor}}$r$\xrightarrow{\textit{getNeighbor}}$p2} would be more efficient from the perspective of graph databases.
    Therefore, it suggests that
    replacing relational operators with graph operators after optimization with relational optimizers can miss optimal execution plans.

    Besides, given the relationships among the tuples as shown in Fig.~\ref{fig:intro-order-case2}, the best join order from the perspective of a relational database like DuckDB can be \textbf{p1$\rightarrow$f$\rightarrow$c$\rightarrow$p2$\rightarrow$r}.
    Then, \textbf{p1$\rightarrow$f} and \textbf{c$\rightarrow$p2} cannot be replaced with \textbf{p1$\xrightarrow{\textit{getNeighbor}}$p2} and some efficient execution plans are missing.

\end{example}

\textbf{Solution 3 ($Rel+G$)}.
According to the grammar of SQL/PGQ, the graph queries usually starts with keyword GRAPH\_TABLE.
Therefore, graph queries can be easily distinguished in SQL/PGQ queries.
Hence, it is possible to optimize graph queries first with graph optimizers, and then optimize the relational query with relational optimizers.
As shown in Fig.~\ref{fig:catagory}, the logical and physical optimization are related, and the cost of graph operators are estimated in the process of optimization.
However, this solution still has limitations, e.g., graph queries and relational queries are optimized separately, and cross-queries optimizations are missing.
An example about this limitation is presented.

\begin{example}
    \label{example:push_down}
    Suppose we are going to find persons that know John and the query expression is as follows:
    \begin{lstlisting}
        SELECT p FROM GRAPH_TABLE (friendship_graph
            MATCH (p1:Person)-[:Knows]-(p2:Person)
            COLUMNS (p1.name as p, p2.name as p2)
        )
        WHERE p2 = 'John';
    \end{lstlisting}
    For $Rel+G$ methods, the graph query is first optimized with a graph optimizer, and the optimized plan finds all pairs of persons that know each other.
    Then, the relational optimizer optimizes the relational query, which finds the persons that know John.
    Please note that the condition ``p2 = 'John''' in the relational query can be pushed down into the graph query, so that the graph query only returns persons that know John.
    The optimized query is as follows.
    \begin{lstlisting}
        SELECT p FROM GRAPH_TABLE (friendship_graph
            MATCH (p1:Person)-[:Knows]-(p2:Person {name: 'John'})
            COLUMNS (p1.name as p)
        );
    \end{lstlisting}
    However, since $Rel+G$ methods optimize relational queries and graph queries separately and do not apply cross-queries optimizations, the condition cannot be pushed down and the optimal execution plan is missed.
\end{example}

In this paper, we propose \textbf{Solution 4 ($Rel\&G$)}, which optimizes the graph queries and relational queries simultaneously with cross-queries optimizations.
Such a solution can fully leverage the advantages of both relational optimizers and graph optimizers.
In detail, we propose a new converged optimization framework of this type for SQL/PGQ.
The framework first generates the converged logical plan consisting of a relational subplan and several graph subplans.
Then, optimization strategies including CBOs and RBOs are applied to optimize inside and crossing subplans.
The contributions of this paper are mainly as follows:

(1) To the best of our knowledge, this is the first optimization framework for SQL/PGQ query optimization.
Property graphs are represented as views in SQL/PGQ, and vertices and edges are associated with tables in the relational databases.
Then, it is crucial to offer a converged query optimizer efficient for SQL/PGQ queries to optimize both relational and graph queries in SQL/PGQ statements.

(2) The framework proposes a new Scan operator named ScanGraphTable to retrieve data from graph tables obtained with graph queries.
The output of ScanGraphTable is a relational table and it bridges the gap between graph subplans and relational subplans.

(3) We prove that graph pattern matching expressed in graph queries of SQL/PGQ can be expressed with graph relational algebra, which confirms that the converged graph relational optimization framework can optimize SQL/PGQ queries and obtain correct results.


(4) Theoretical analysis on the complexity of the optimization framework is conducted.
The obtained theorems prove that for graph queries, the join order optimization with a graph optimizer can be exponentially faster than that with a relational optimizer.
It theoretically confirms that relational optimizer is usually not suitable for graph queries, and it is indispensable for the existence of a converged optimization framework.

(5) Extensive experiments are conducted to show the efficiency of the proposed converged query optimization framework.
The experimental results show that the framework can be ?$\times$ faster than the baselines.

The rest of this paper is organized as follows.

The existing methods for optimizing the SPJM queries can be divided into four categories.
The main difference between these methods lies in the approach to handling the matching operator.
Before the details of these methods are introduced, concepts about graph structure and graph matching decomposition are proposed as follows.
}

%% file: sec-preliminaries.tex
\section{Preliminaries}
\label{sec:preliminaries}

In this section, we propose the utilized data model and define the SPJM query processed in this paper. Frequently used notations in this paper are summarized in \reftable{notations}.

\vspace*{-3mm}
\subsection{Data Model}
\label{sec:data-model}

A schema, denoted as \(S = (a_1, a_2, \ldots, a_n)\), is a collection of attributes. Each attribute \(a_i\) is associated with a specific data domain \(D_i\), which defines the set of permissible values that \(a_i\) can take.
A relation \(R\) is defined as a set of tuples. We consider \(R\) to be a relation over schema \(S\), if and only if, every tuple \(\tau = (d_1, d_2, \ldots, d_n)\) in \(R\) adheres to the schema's constraints, such that the value \(d_i\) for each position in the tuple corresponds to the data domain \(D_i\) of the attribute \(a_i\) in \(S\). In other words, each value \(d_i\) in a tuple \(\tau\) is drawn from the appropriate data domain \(D_i\) for its corresponding attribute \(a_i\).
Moreover, for any tuple \(\tau\) in the relation \(R\), the notation \(\tau.a_i = d_i\) signifies that the attribute \(a_i\) in tuple \(\tau\) has value \(d_i\).
A table is a representation of a relation with rows corresponds to tuples in the relation, and columns represent attributes in the schema. In this paper, we use the terms of relation and table interchangeably.

We define a \emph{Property Graph} as $G = (V_G, E_G)$,
where $V$ stands for the set of vertices. Let $E \subseteq V \times V$ denote the set of edges in the graph. An edge $e \in E$ is represented as an ordered pair $e = (v_s, v_t)$, where $v_s \in V$ is the source vertex and $v_t \in V$ is the target vertex, indicating that the edge $e$ connects from $v_s$ to $v_t$.
For any graph element $\epsilon$ that is either a vertex or an edge, we denote $\id(\epsilon)$ and $\lab(\epsilon)$ as the globally unique ID and the label of $\epsilon$, respectively. Given an attribute $a_i$, $\epsilon.a_i$ denotes the value of the attribute $a$ of $\epsilon$.

Given a vertex $v$, we denote its adjacent edges as $\adj_G^E(v) = \{e = (v, v_t) | e \in E\}$ and its adjacent vertices (i.e., neighbors) as $\adj_G(v) = \{v_t | (v, v_t) \in E\}$. It is important to note that the adjacent edges and vertices can be defined for both directions of an edge $e = (v_s, v_t)$, i.e., when $v = v_s$ or $v = v_t$. However, for simplicity, we only define one direction in this notation. In the actual semantics of the paper, both directions may be considered. The degree of $v$ is defined as $d_G(v) = |\adj_G(v)|$, and the average degree of all vertices in the graph is $\overline{d}_G = \frac{1}{|V_G|} \sum_{v \in V_G} d_G(v)$.
In the rest of the paper, when the context is clear, we may remove $G$ from the subscript for simplicity, for example $G=(V, E)$.

Considering two graphs \(G_1\) and \(G_2\), we assert that \(G_2\) is a subgraph of \(G_1\), symbolized as \(G_2 \subseteq G_1\), if and only if \(V_{G_2} \subseteq V_{G_1}\), and \(E_{G_2} \subseteq E_{G_1}\). Furthermore, \(G_2\) qualifies as an induced subgraph of \(G_1\) under the condition that \(G_2\) is already a subgraph of \(G_1\), and for every pair of vertices in \(G_2\), any edge \(e\) that exists between them in \(G_1\) must also present in \(G_2\).

To illustrate the integration of graph syntax within the realm of relational data, we introduce the concept of a \textit{Relations-to-Graph Mapping} (i.e. \rgmapping), to facilitate the transformation of relational data structures into a property graph.

An \revise{\rgmapping consists of an vertex mapping and an edge mapping that map tuples in relations to unique vertices or edges. To better describe these vertex and edge mappings, we can leverage the Entity-Relationship (ER) diagram~\cite{song1995comparative,chen1983english}. In relational data modeling, an ER diagram includes entities and relationships. Consequently, vertices can be mapped from relations corresponding to entities, and edges can be mapped from relations corresponding to relationships. Relations mapped to vertices and edges are referred to as vertex relations and edge relations, respectively.}

In \revise{detail, if a tuple $\tau$ in relation $R$ is mapped to a vertex $v \in V$ (or an edge $e = (v_s, v_t) \in E$), it is assigned an ID $\id(v)$ (or $\id(e)$), a label $\lab(v)$ (or $\lab(e)$)  that corresponds to the name of $R$, and attributes $v.attr*$ (or $e.attr*$) that reflect the attributes $attr*$ of $\tau$. For an edge relation $R_{e}$, there must exist two vertex relations, $R_{p_s}$ and $R_{p_t}$. Two total functions are defined: $\lambda_{e}^s: R_{e} \to R_{p_s}$ and $\lambda_{e}^t: R_{e} \to R_{p_t}$. Consider a tuple $\tau \in R_{e}$ mapped to an edge $e$, and tuples $\tau_s \in R_{p_s}$ and $\tau_t \in R_{p_t}$, where $\lambda_{e}^s(e) = \tau_s$ and $\lambda_{e}^t(e) = \tau_t$. Through the vertex mapping, $\tau_s$ is mapped to the source vertex $v_s$ and $\tau_t$ to the target vertex $v_t$ of the edge $e$. The two total functions are often established through primary-foreign key relationships, as illustrated in an ER diagram.}

\comment{
In detail, \revise{an \rgmapping comprises a vertex mapping and an edge mapping that maps tuples in relations to unique vertices or edges. Suppose a tuple $\tau$ in relation $R$ is mapped to a vertex $v \in V$ (or edge $e = (v_s, v_t) \in E$), it is assigned an ID $\id(v)$ (or $\id(e)$) that matches the name of $R$, a label $\lab(v)$ (or~$\lab(e)$), and attributes $v.attr*$ (or $e.attr*$) that mirror the attributes $attr*$ of $\tau$.
Relations mapped to vertices and edges are referred to as vertex relations and edge relations, respectively.
}

Given \revise{an edge relation $R_{e}$ and two vertex relations $R_{p_s}$ and $R_{p_t}$, we further define two total functions: $\lambda_{e}^s: R_{e} \to R_{p_s}$ and $\lambda_{e}^t: R_{e} \to R_{p_t}$, which maps a tuple $\tau \in R_{e}$ to a tuple $\tau_s \in R_{p_s}$ and a tuple $\tau_t \in R_{p_t}$, respectively.
Applying the vertex mapping, $\tau_s$ and $\tau_t$ are mapped to the source vertex $v_s$ and target $v_t$ of the edge, respectively.
These total functions are typically established through primary-foreign key relationships, which are best illustrated in an Entity-Relationship (ER) diagram~\cite{song1995comparative,chen1983english}.}
}

\begin{figure*}
    \centering
    \includegraphics[width=0.95\linewidth]{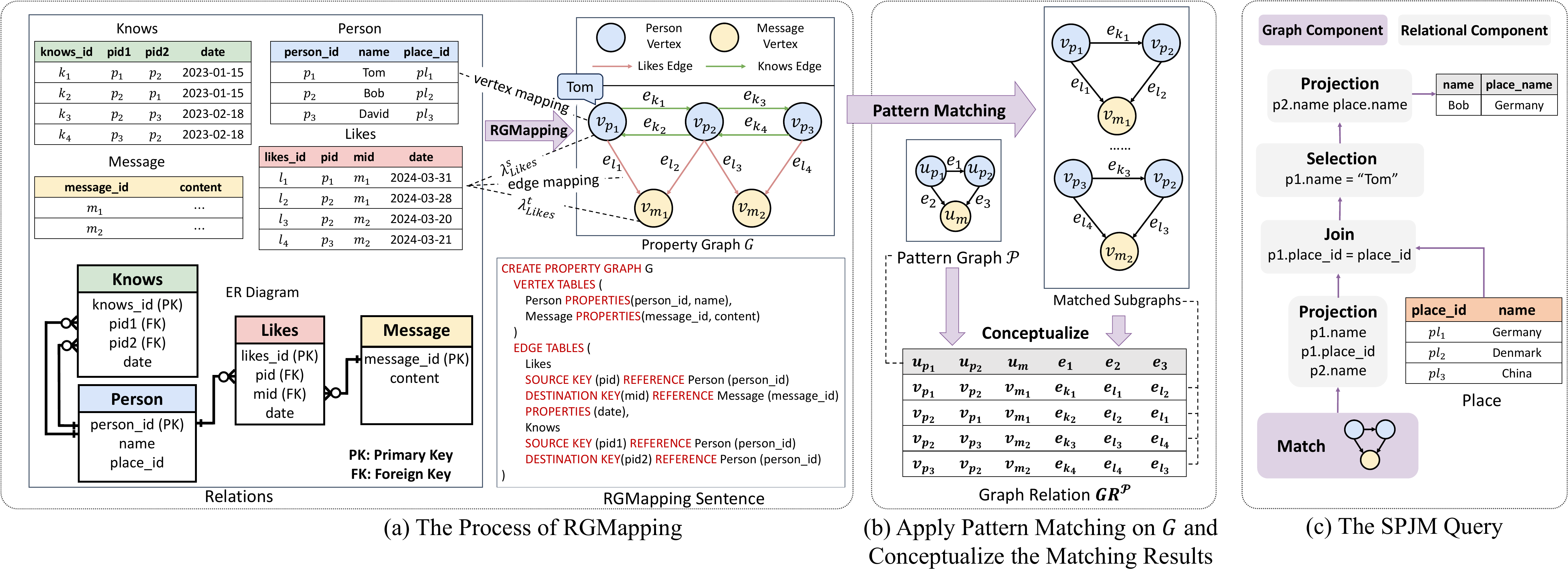}
    \vspace*{-2mm}
    \caption{An example of \rgmapping.}
    \vspace*{-1mm}
    \label{fig:intro-rgmapping-example}
\end{figure*}

\begin{example}
    \label{ex:rgmapping}
    \revise{In Fig.~\ref{fig:intro-rgmapping-example}(a), we have illustrated some relational tables and their corresponding ER diagram.
    An \rgmapping can be defined following the grammar of SQL/PGQ with \lstinline{CREATE PROPERTY GRAPH} statements.}
    The described \rgmapping involves assigning tuples from vertex relations (\revise{i.e. entities}), such as $\relation{Person}$ and $\relation{Message}$, to graph vertices. For instance, the vertex $v_{p_1}$ is associated with the tuple $\tau_{p_1}$ in $\relation{Person}$, and thus assigned the label ``Person'' and the name attribute ``Tom''. Similarly, edge relations (\revise{i.e. relationships}) $\relation{Likes}$ and $\relation{Knows}$ correspond to graph edges.
    \revise{Regarding $\relation{Likes}$
    that is mapped to graph edges, two total functions can be identified, namely $\lambda_{\text{Likes}}^s: \relation{\text{Likes}} \to \relation{\text{Person}}$ and $\lambda_{\text{Likes}}^t: \relation{\text{Likes}} \to \relation{\text{Message}}$.}
    Let's consider the edge $e_{l_1}$. It originates from the tuple $\tau_{l_1}$ in the $\relation{Likes}$ relation. Its source vertex $v_{p_1}$ is linked to the tuple $\tau_{p_1}$ in $\relation{Person}$ via the function $\lambda_{Likes}^s$, following the \revise{primary-foreign key relationship} ``$\tau_{l_1}.pid = \tau_{p_1}.\text{person\_id}$''. Similarly, its target vertex $v_{m_1}$ is associated with the tuple $\tau_{m_1}$ in $\relation{Message}$ via the function $\lambda_{Likes}^t$, following ``$\tau_{l_1}.mid = \tau_{m_1}.\text{message\_id}$''. As a result of this mapping, the edge $e_{l_1}$ is assigned the label ``Likes'' and the attribute ``date'' with the value ``2024-03-31''.
\end{example}



\vspace*{-3mm}
\subsection{Matching Operator}
\label{sec:matching-operator}
Consider a property graph \(G(V_G, E_G)\), alongside a \emph{connected} pattern graph, represented as \(\pattern(V_\pattern, E_\pattern)\). Here, \(\pattern\) is a property graph that does not possess attributes, and
we denote $n$ and $m$ as the number of vertices and edges in the $\pattern$, respectively.
Graph pattern matching seeks to determine all subgraphs in \(G\) that are \emph{homomorphic} to \(\pattern\).
Formally, given a subgraph $g \subseteq G$, a homomorphism from \(\pattern\) to \(g\) is a \emph{surjective}, total mapping \(f: V_\pattern \cup E_\pattern \to V_{g} \cup E_{g}\) that satisfies the following conditions: (1) For every vertex \(u \in V_\pattern\), there is a corresponding vertex \(v = f(u) \in V_{g}\) with \(\lab(v) = \lab(u)\); (2) For each edge \(e = (u_s, u_t) \in E_\pattern\), there is a corresponding edge \(f(e) = (v_s, v_t) \in E_{g}\), ensuring that the mapping preserves the edge's the label, as well as its source and target vertices, that is \(\lab(e) = \lab(f(e))\), and \(f(u_s) = v_s\), \(f(u_t) = v_t\). It's important to highlight the homomorphism semantics, as one of the widely used semantics for graph pattern matching~\cite{angles2017foundations}, do not require each pattern vertex and edge being uniquely mapped to distinct vertices and edges in the data graph. This facilitates a seamless integration between graph pattern matching and relational operations, but alternative semantics for graph pattern matching such as isomorphism can also be adopted, as will be further discussed in~\refsec{handling-match-operator}.

The outcomes of graph pattern matching can be succinctly modeled as a relation \(GR_G^\pattern\), or more compactly \(GR^\pattern\) in clear contexts, defined over the schema \(S = V_\pattern \cup E_\pattern\). Here, the sets \(V_G\) and \(E_G\) serve as the respective domains for the vertices and edges identified through the matching process. Within this framework, we refer to such a relation as a \emph{Graph Relation}, a construct where all attributes are derived from the domain of a property graph.
It is essential to recognize that any property graph \(G\) can be conceptualized as a graph relation \(GR^G\), represented by a singular tuple that collectively encompasses all of its vertices and edges. Throughout this paper, we treat the notions of a property graph and a tuple of graph relation as essentially interchangeable terms. In alignment with this perspective, we elaborate on the \emph{Matching} operator as follows.



\begin{definition}[Matching Operator, \(\matching\)]
    \label{def:matching}
    The Matching Operator, denoted as \(\matching\), is designed to perform graph pattern matching on a given graph relation \(GR\) against a specified pattern graph \(\pattern\). For each graph instance \(g\) in \(GR\), \(\matching\) identifies all subgraphs of \(g\) that are homomorphic to \(\pattern\), and subsequently, aggregates these mappings to construct a comprehensive graph relation. The operation of the matching Operator can be formally articulated as \(\matching(GR, \pattern) = \bigcup_{g \in GR} GR_g^\pattern\).
\end{definition}

\begin{example}
    \label{ex:matching}
    Let \(G\) denote the property graph derived from the relations via \rgmapping in \refex{rgmapping}.
    Given a pattern graph \(\pattern\) in \reffig{intro-rgmapping-example}(b), the results of graph pattern matching are subgraphs of \(G\) that are homomorphic to \(\pattern\), represented as a graph relation \(GR^\pattern = \matching(GR^{G}, \pattern)\), each tuple corresponds to one matched subgraph.
\end{example}

This definition ensures that the matching operator is inherently closed regarding graph relations,
which adheres to the language opportunities of ``nested matching'' (specified as PGQ-079) in SQL/PGQ~\cite{sql-pgq}.
In this paper, we only handle cases where $G$ represents the entire property graph, and thereafter simplify the matching operator notation to $\matching(\pattern)$ when the context is clear.


\vspace*{-1mm}
\subsection{Problem Definition}
\label{sec:problem-definition}

To study relational query optimization, it is common to focus on \spj queries,
which consists of three most frequently used operations: select, project, and (natural) join.
These operations form the backbone of many relational queries. Given a set of relations \(R_1, R_2, \ldots, R_m\),
an \spj query is formally represented as:
\[
Q = \pi_A(\sigma_\constraints(R_1 \Join \cdots \Join R_m)).
\]

Inspired from the \spj paradigm, we introduce a novel category of queries, termed \spjm queries, to logically formulate SQL/PGQ~\cite{sql-pgq} queries that
blend relational and graph operations.
The \spjm framework augments \spj queries by incorporating a matching operator to  enrich the query's expressive power, to seamlessly navigate both relational and graph data domains.
Given the set of relations and a property graph \(G\) constructed from these relations via an \rgmapping, 
an \spjm query is articulated as:
\begin{equation}
    \label{eq:spjm}
Q = \pi_A(\sigma_\constraints(R_1 \Join \cdots \Join R_m \Join (\widehat{\pi}_{A*}\matching_G(\pattern))))
\end{equation}
In this formulation, \(\widehat{\pi}_{A*}\matching_G(\pattern)\) is the \emph{graph component} of
the query, while the remaining part of the query is an \spj expression referred to as the \emph{relational component}.
Here, \(\matching_G(\pattern)\) represents the process of matching the pattern \(\pattern\) on the graph \(G\) and
returns a graph relation as defined in \refdef{matching}. The operator \(\widehat{\pi}_{A*}\) is a
graph-calibrated projection operator that extracts the ID, label, and other attributes from the vertices and edges in the matched results.
This process helps ``flatten'' graph elements into relational tuples.
For example, given a graph relation $GR$ that contains a vertex of \{ID:0, label:Person, name:``Tom''\}, the
 projection
$\widehat{\pi}_{\id(v) \rightarrow \text{v\_id}, \lab(v) \rightarrow \text{v\_label}, v.name \rightarrow \text{v\_name}}(GR)$
turns the vertex into a relational tuple of (0, Person, ``Tom'').
The projection is designed to reflect the \lstinline{COLUMNS} clause in SQL/PGQ
to retrieve specific attributes from vertices and edges as required. For simplicity, we assume that all
attributes are extracted unless otherwise specified.

In this paper, we study the problem of optimizing \spjm queries in \refeq{spjm}. \reffig{intro-rgmapping-example}(c) illustrates the \spjm query skeleton corresponding to the SQL/PGQ query in \refex{introduction:sqlpgq}.

\comment{
\begin{example}
    \reffig{intro-rgmapping-example}(c) illustrates an example of an \spjm query, along with its corresponding SQL/PGQ expression. The query's semantic is to identify Tom's friends (and their living place) who like the same messages as Tom. As per the query, once the results of the matching operator are obtained (as described in \refex{matching}), a $\gproject$ operator is applied to extract attributes from all vertices, forming a relational table. Subsequently, the join, selection, and projection operators are employed to compute the final results.
\end{example}
}

%% file: sec-match-opt.tex
\vspace*{-1mm}
\section{Optimizing Matching Operator}
\label{sec:optimizing-matching-operator}
In this section, we focus on handling the matching operator, 
which plays a distinct role within the \spjm queries. 
We discuss two main perspectives of optimizing
the matching operator: logical transformation and physical implementation. Logical transformation is
responsible for transforming a matching operator into a logically equivalent representation,
while physical implementation focuses on how the matching operator can be efficiently executed.

\vspace*{-2mm}
\subsection{Logical Transformation}
\label{sec:handling-match-operator}
We commence with an intuitive, graph-agnostic transformation before
introducing a graph-aware technique grounded on the concept of decomposition tree, which
is the key to the optimization of graph pattern matching in the literature~\cite{huge,GLogS}.

Before proceeding, we introduce the concept of pattern decomposition that decomposes $\pattern$ into two overlapping patterns, $\pattern_1$ and $\pattern_2$, with shared vertices $V_{o} = V_{\pattern_1} \cap V_{\pattern_2}$ and shared edges $E_{o} = E_{\pattern_1} \cap E_{\pattern_2}$.
Denote $\pattern = \pattern_1 \cup \pattern_2$. Under the homomorphism semantics, the matching of $\pattern$ can be represented as:
\begin{equation}
    \label{eq:join-pattern}
    \matching(\pattern) = \matching(\pattern_1) \widehat{\Join}_{V_{o}, E_{o}} \matching(\pattern_2),
\end{equation}
where $\widehat{\Join}$ is a natural join operator for joining two graph relations based on the common vertices and edges.
Note that \refeq{join-pattern} is also applicable to alternative semantics, including isomorphism and non-repeated-edge~\cite{angles2017foundations}. To support these semantics, a special all-distinct operator can be applied as a filter to remove results that contain duplicate vertices and/or edges. The adoption of the all-distinct operator is compatible with all techniques in this paper.



\vspace*{-1mm}
\subsubsection{Graph-agnostic Transformation}
\label{sec:intuitive-method}
If the matching operator can be transformed into purely relational operations, the \spjm query becomes a
standard \spj query, which can then be optimized using existing relational optimizers (\refsec{relational-only}). This graph-agnostic
approach is intuitive and easy to implement on top of existing relational databases, making it a straightforward
choice in prototyped systems~\cite{apache-age,DuckPGQ,DuckPGQ-VLDB}. However, there is no theoretical guarantee that
such a transformation is lossless in the context of \rgmapping. In this subsection,
we bridge this gap by demonstrating the lossless transformation of the matching
operator under \rgmapping.

Consider a pattern graph $\pattern$ and one of its edges $e = (u_s, u_t)$. According to the definition of the matching operator (\refsec{matching-operator}), the graph edges and vertices that can be matched with $e$ must have the labels $\lab(e)$, $\lab(u_s)$, and $\lab(u_t)$. We further denote the relations corresponding to these edges and vertices via \rgmapping as $R_{\lab(e)}$, $R_{\lab(u_s)}$, and $R_{\lab(u_t)}$, respectively. Moreover, there must be total functions $\lambda_{\lab(e)}^s$ and $\lambda_{\lab(e)}^t$ for mapping tuples from $R_{\lab(e)}$ to $R_{\lab(u_s)}$ and $R_{\lab(u_t)}$, respectively. We define the following \EVjoin relational operation regarding $\lambda_{\lab(e)}^s$ as:
\begin{equation} \label{eq:ev-join}
\begin{split}
R_{\lab(e)} & \evjoin R_{\lab(u_s)} = \{(\tau_e, \tau_s) \;|\; \\
  &  \tau_e \in R_{\lab(e)} \land \tau_s \in R_{\lab(u_s)} \land \lambda_{\lab(e)}^s(\tau_e) = \tau_s\}.
\end{split}
\end{equation}
The \EVjoin regarding $\lambda_{\lab(e)}^t$ is defined analogously. Although called \EVjoin, the operation is associative like any relation join, meaning that the order in which the edge and vertex relations are joined does not affect the final result.
%

We have the following lemma.

\begin{lemma}
    \vspace*{-3mm}
    \label{lem:spjm-to-spj}
    Under \rgmapping, the matching operation in an \spjm query can be losslessly transformed into a sequence of relational joins involving $n$ vertex relations and $m$ edge relations.
\end{lemma}
\vspace*{-4mm}
\begin{proof}
    Consider a pattern $\pattern_m$ of $m$ edges, where the $i$-th vertex is denoted as $u_i$, and the $i$-th edge is $e_i = (u_{s_i}, u_{t_i})$. 

The proof proceeds by induction, starting with a pattern graph $\pattern_0$ with a single vertex only. It is clear that $\matching(\pattern_0)$ yields a subset of vertices with label $\lab(u_0)$, which is mapped from the relation $\relation{\lab(u_0)}$ via $\rgmapping$. As a result, we have $R_0 = \gproject(\matching(\pattern_0)) = \relation{\lab(u_0)}$.

Next, consider $\pattern_1$ with one edge, $e_1 = (u_{s_1}, u_{t_1})$. Matching $\pattern_1$ is equivalent to retrieving the edge relation, together with the corresponding source and target vertices. Therefore, we have:
\[ R_1 = \gproject(M(\pattern_1)) = \relation{\lab(u_{s_1})} \evjoin \relation{\lab(e_1)} \evjoin \relation{\lab(u_{t_1})} \]
Assume that when $m = k-1$, $\gproject(\matching(\pattern_{k-1}))$ can be converted to a sequence of relational operators, resulting in $R_{k-1}$. When $m = k$, we consider $\pattern_{k}$ of $k$ edges constructed from $\pattern_{k-1}$ by adding edge $e_k = (u_{s_k}, u_{t_k})$. For $\pattern_{k}$ to be connected, it must share at least one common vertex $V_o$ with $\pattern_{k-1}$. According to \refeq{join-pattern}, we have:
\[ \matching(\pattern_{k}) =  \matching(\pattern_{e_k}) \gjoin_{V_o} \matching(\pattern_{k-1}), \]
where $\pattern_{e_k}$ denotes a pattern that contains only the edge $e_k$, and $V_o$ is the common vertex shared by $\pattern_{k-1}$ and $\pattern_{e_k}$. Applying $\gproject$ to the above equation, we get:
\vspace*{-1mm}
\begin{equation*}
\begin{split}
R_k &= \gproject(\matching(\pattern_{k})) \\
    &= \widehat{\pi}_{A_1*}(\matching(\pattern_{e_k})) \Join_{V_o.attr}  \widehat{\pi}_{A_2*}(\matching(\pattern_{k-1})) \\
    &= \relation{\lab(u_{s_k})} \evjoin \relation{\lab(e_k)} \evjoin \relation{\lab(u_{t_k})} \Join_{V_o.attr} R_{k-1}
\end{split}
\end{equation*}
By induction, denoting $R'_i = \relation{\lab(u_{s_i})} \evjoin \relation{\lab(e_i)} \evjoin \relation{\lab(u_{t_i})}$, we have the matching operator losslessly converted to a sequence of relational join operations:
\begin{equation}
    \label{eq:graph-agnostic}
    \gproject(\matching(\pattern_{k})) = R'_k \Join R'_{k-1} \Join \cdots \Join R'_1 \Join R_0.
\end{equation}

We thus conclude the proof.
\end{proof}


\begin{example}
  \label{ex:spjm-to-spj}
  Given pattern graph $\pattern$ in \reffig{intro-rgmapping-example}(b), the matching operation $\matching(\pattern)$ can be converted to a sequence of join operations as follows. Without loss of generality, we start from $\pattern_0$ containing only the vertex $u_{p_1}$, and we have $\relation{0} = \relationx{1}{\text{Person}}$ (note that the superscript 1 is used to differentiate relations of the same name).
  Next, we sequentially add the edges $e_1 = (u_{p_1}, u_{p_2})$, $e_2 = (u_{p_1}, u_{m})$, and $e_3 = (u_{p_2}, u_m)$ to $\pattern_0$, resulting in the following relations:
  \vspace*{-1mm}
  \begin{equation*}
    \begin{split}
    R'_1 &= \relationx{1}{\text{Person}} \Join_{\text{person\_id}=\text{pid1}} \relation{\text{Knows}} \Join_{\text{pid2}=\text{person\_id}} \relationx{2}{\text{Person}}, \\
    R'_2 &= \relationx{1}{\text{Person}} \Join_{\text{person\_id}=\text{pid}} \relationx{1}{\text{Likes}} \Join_{\text{mid}=\text{message\_id}} \relation{\text{Message}}, \\
    R'_3 &= \relationx{2}{\text{Person}} \Join_{\text{person\_id}=\text{pid}} \relationx{2}{\text{Likes}} \Join_{\text{mid}=\text{message\_id}} \relation{\text{Message}}.
    \end{split}
    \end{equation*}
    Finally, we have $\gproject(\matching(\pattern)) = R'_3 \Join R'_2 \Join R'_1 \Join \relation{0}$.
    Note that $\relationx{1}{\text{Person}}$ in $R'_2$, as well as $\relationx{2}{\text{Person}}$ and $\relation{\text{Message}}$ in $R'_3$, are redundant and can be removed from the final join. By eliminating them, we obtain a sequence of joins with 3 vertex relations and 3 edge relations.
\end{example}


\subsubsection{Graph-aware Transformation}
\label{sec:graph-aware}

We introduce a graph-aware transformation that incorporates key ideas from the literature on graph optimization. Following \refeq{join-pattern}, we can recursively decompose $\pattern$, forming a tree structure called the \emph{decomposition tree}. The tree has a root node that represents $\pattern$, and each non-leaf \emph{intermediate} node is a sub-pattern (a subgraph of the pattern) $\pattern' \subset \pattern$, which has a left and right child node, denoted as $\pattern'_l$ and $\pattern'_r$, respectively. 
The leaf nodes of the tree are called \emph{Minimum Matching Components} (\mmc), correspond to indivisible patterns directly solvable with specific physical operations
as will be introduced in~\refsec{physical-operators}. The decomposition tree naturally forms a logical plan for solving $\matching(\pattern)$, as demonstrated in \reffig{match-decomposition}. For any non-leaf node $\pattern'$, there exists a relationship $\matching(\pattern') = \matching(\pattern'_l) \gjoin \matching(\pattern'_r)$ according to \refeq{join-pattern}. The plan allows for the recursive computation of the entire pattern.

Following state-of-the-art graph optimizers~\cite{huge,GLogS}, to guarantee a \emph{worst-case optimal} execution plan~\cite{ngo2018worst}, all intermediate sub-patterns in the decomposition tree must be induced subgraphs of $\pattern$. Furthermore, \mmc is restricted to be a single-vertex pattern and a \emph{complete star}. A star-shaped pattern is denoted as $\pattern(u;V_s)$, where $u$ is the root vertex and $V_s$ is the set of leaf vertices\footnote{Edge directions between $u$ and $V_s$ are not important, and we assume they all point from $u$ to $V_s$.}. In the decomposition tree, given $\pattern' = \pattern'' \cup \pattern(u;V_s)$, $\pattern(u;V_s)$ is a complete star if and only if it is a right child and $V_s \subseteq V_{\pattern''}$, meaning that the leaf vertices of the complete star must all be common vertices for the decomposition. A single-edge pattern is a special case of a complete star. The complete star logically represents the physical operations of \expandintersect, which will be discussed in \refsec{physical-operators}. As shown in \reffig{match-decomposition}, a single-edge pattern, such as $\pattern_3$, is further decomposed into a single-vertex pattern and the pattern itself, allowing the optimizer to select from which vertex the edge can be expanded.
\revise{The intermediate sub-patterns pruned from the decomposition tree are also presented in \reffig{match-decomposition}}.
\revise{Some previous studies, such as EmptyHeaded \cite{Aberger2016Sigmod} and CLFTJ \cite{Kalinsky2017EDBT}, have also explored decomposition trees. However, our method significantly differs from theirs. Specifically, in these previous methods, the tree nodes represent sets of relations, and the edges in the decomposition trees connect nodes with common join keys. In contrast, the nodes in our decomposition trees represent sub-patterns (relations that can form a graph after \rgmapping) of queries. Each edge in our tree connects two nodes such that the child sub-pattern can be computed from the parent sub-pattern in some execution plan.
}


\comment{
Therefore, decomposing the matching operators recursively can finally result in a tree, whose leaf nodes are MMCs.
To ensure worst-case optimality, in the process of decomposition, the pattern graph of each decomposed matching operator should be an induced subgraph of $\mathcal{P}_0$.
The generated tree is called a decomposition tree and it is actually a logical plan of the matching operator.
\modify{Without loss of generality, the left-deep join order is employed on the tree.}

Then, it is crucial to identify which matching operators to treat as MMCs.
We adopt the definition of \emph{complete star} from \cite{huge}.
Specifically, suppose $\matching(GR, \mathcal{P})$ is decomposed into $\matching(GR, \mathcal{P}_1)$ $\widehat{\Join}$ $\matching(GR, \mathcal{P}_2)$, $\mathcal{P}_2$ is called a complete star iff $\mathcal{P}_2$ is a star $(v_r; \mathcal{H})$, $\mathcal{H} \subseteq E_{\mathcal{P}_1}$, and $|\mathcal{H}| \geq 2$, where $v_r$ is the root and $\mathcal{H}$ is the set of its leaf vertices.
Inspired by HUGE \cite{huge} and GLogS \cite{GLogS}, matching operators located on the right subtree of the join operator with complete stars as the pattern graphs are the MMCs in this paper.

Besides, a matching operator located on the left subtree of the join operator is an MMC iff its pattern graph is an edge (i.e., one edge and its adjacent two vertices).
}

\begin{figure}
    \centering
    \includegraphics[width=.8\linewidth]{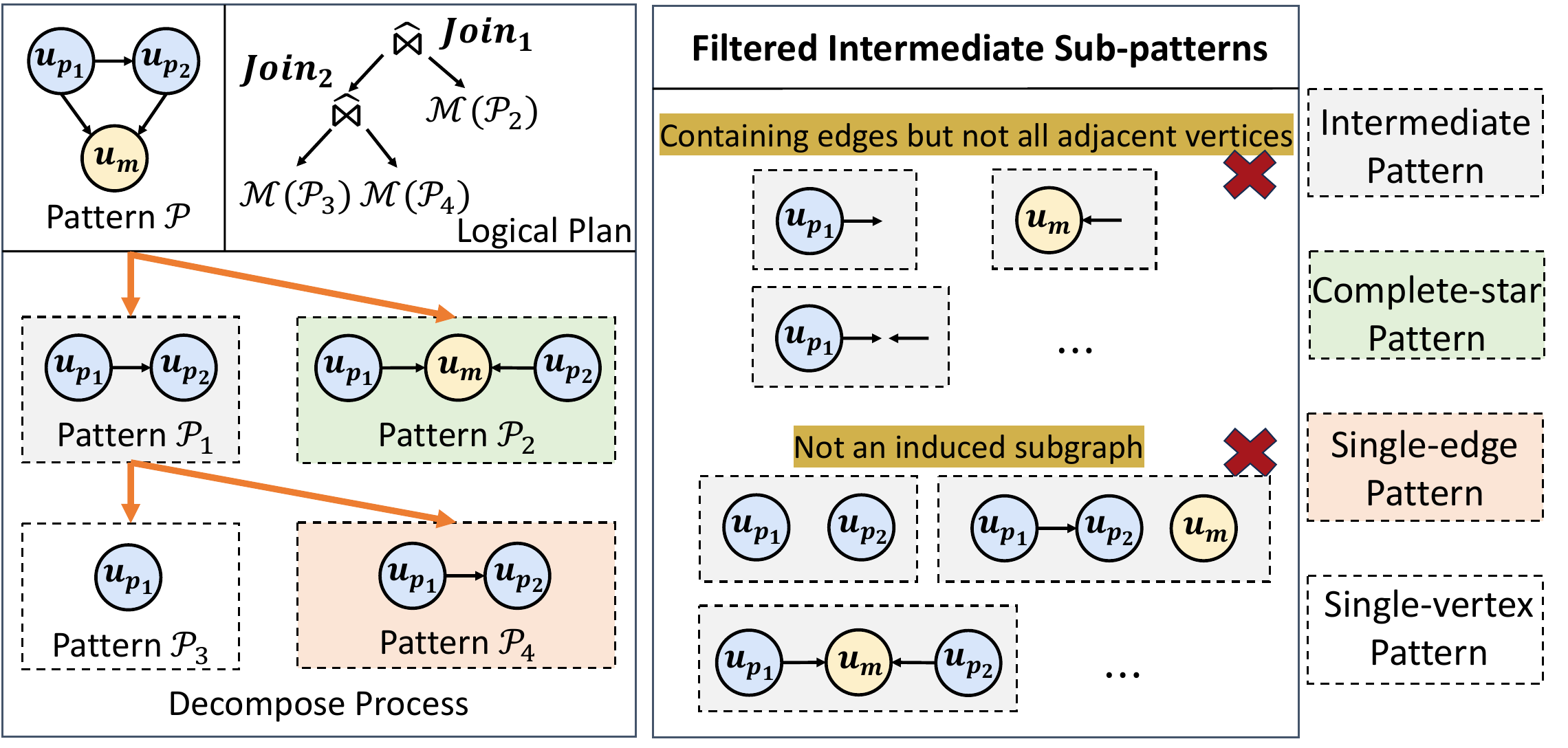}
    \caption{Example of decomposition trees and the corresponding logical plans. \revise{Note that sub-pattern $\pattern_2$ can be a leaf node, but it cannot be an intermediate node}.}
    \label{fig:match-decomposition}
\end{figure}

\comment{
\begin{example}
    \todo{the sample and figure must all be refined.}
    \modify{
    As shown in \reffig{match-decomposition}, $T_1$ and $T_2$ are two possible decomposition trees formed by recursively decomposing $\pattern_{0}$.
    Specifically, $\pattern_{1,1}$, $\pattern_{1,3}$, $\pattern_{1,4}$, $\pattern_{2,2}$, $\pattern_{2,3}$, and $\pattern_{2,4}$ are \mmcs. Among these \mmcs, $\pattern_{1,1}$, $\pattern_{1,3}$, $\pattern_{2,3}$, $\pattern_{2,4} are $single-edge patterns, while $\pattern_{1,4}$ and $\pattern_{2,2}$ are complete stars.
    Please note that $\pattern_{2,1}$ is not an \mmc, because it is not the right child of $\pattern_{0}$.
    }
\end{example}
}


\begin{remark}
    \label{rem:graph-agnostic-vs-graph-aware}
    The graph-aware transformation is fundamentally different from its graph-agnostic counterpart. While the graph-agnostic approach consistently converts pattern matching operations into relational joins between vertex and edge relations,
    the graph-aware transformation does not, due to the constraints imposed by pattern decomposition. While the graph-agnostic approach is straightforward, it has the following drawbacks:

    \noindent\textbf{Graph-unaware Join Order}: It may lead the relational optimizer to reorder the join of vertex and edge relations, potentially missing chances to use graph indexes for efficiently computing adjacent edges and vertices, as discussed in \refsec{graph-index}.

    \noindent\textbf{Suboptimal Join Plans}: It generates plans that consistently reflect edge-based join plans that have been shown to be suboptimal in terms of worst-case performance~\cite{lai2015scalable}.

    \noindent\textbf{Increased Search Space}: Compared to the graph-aware transformation, it can lead to an exponentially larger search space when computing optimal plans, which will be discussed in the following.

\end{remark}

\subsubsection{The Search Space: Graph-agnostic vs Graph-aware}
\label{sec:compare-search-space}

\begin{figure}[t]
    \centering
    \begin{subfigure}[b]{\linewidth}
        \centering
        \begin{subfigure}[b]{.4\linewidth}
            \centering
            \includegraphics[width=.85\linewidth]{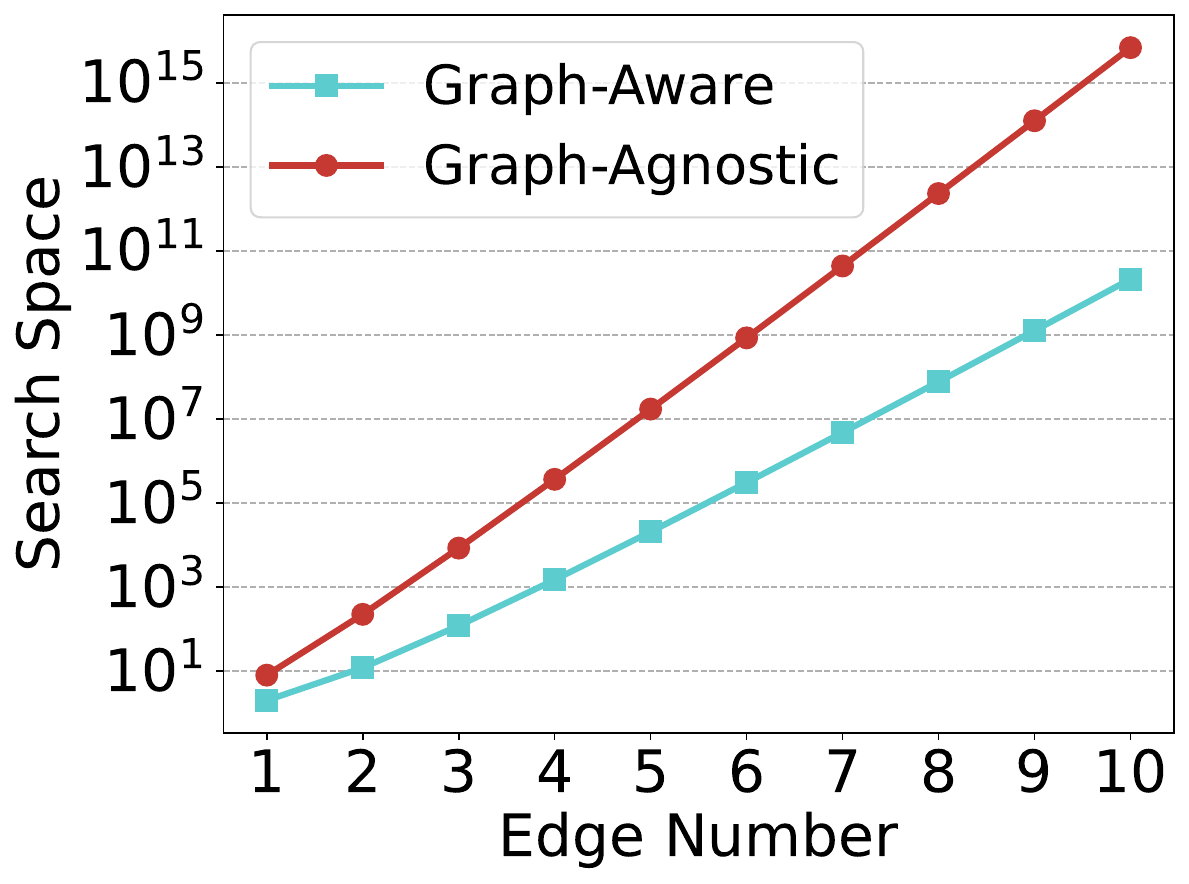}
        \end{subfigure}
        \begin{subfigure}[b]{0.4\linewidth}
            \centering
            \includegraphics[width=.85\linewidth]{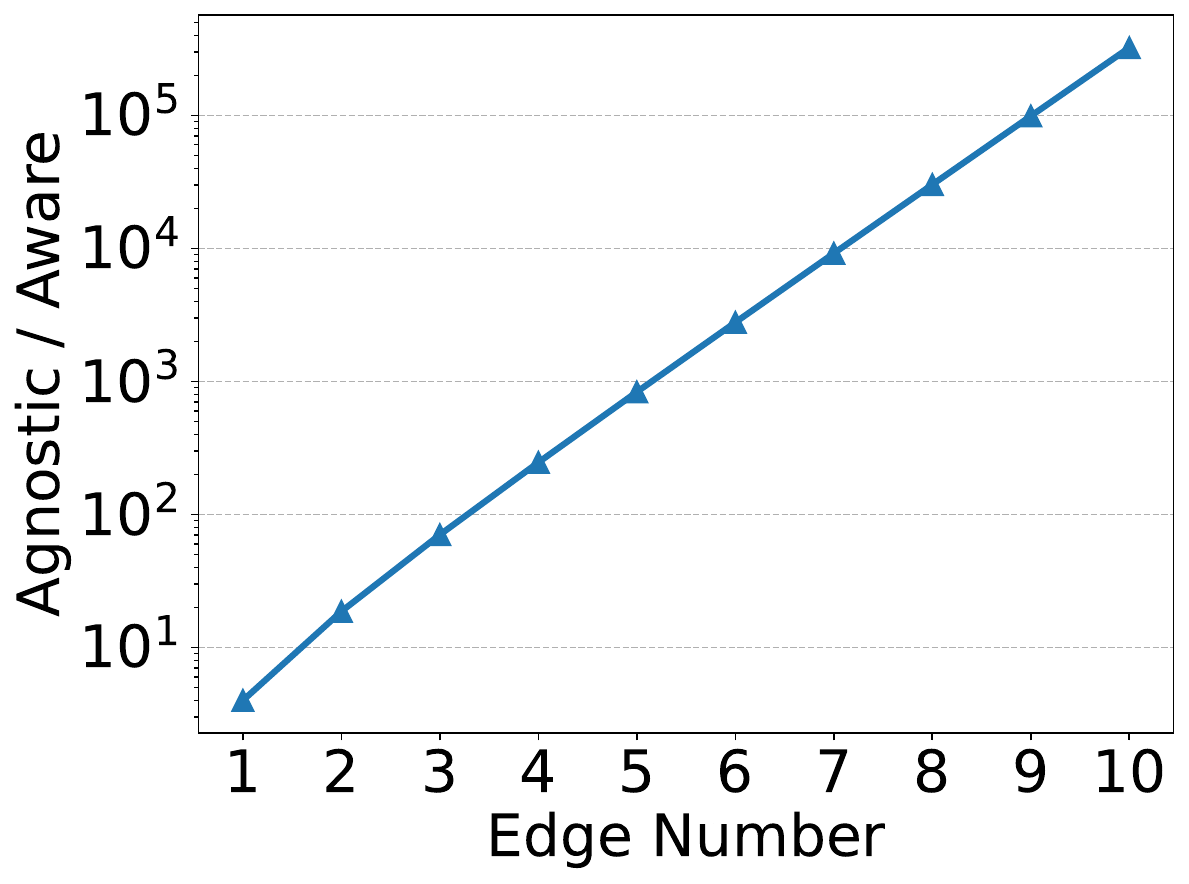}
        \end{subfigure}
        \vspace{-.5em}
        \caption{Search Space Comparison.}
        \label{fig:exp-search-space}
    \end{subfigure}
    \begin{subfigure}[b]{\linewidth}
        \centering
        \includegraphics[width=.8\linewidth]{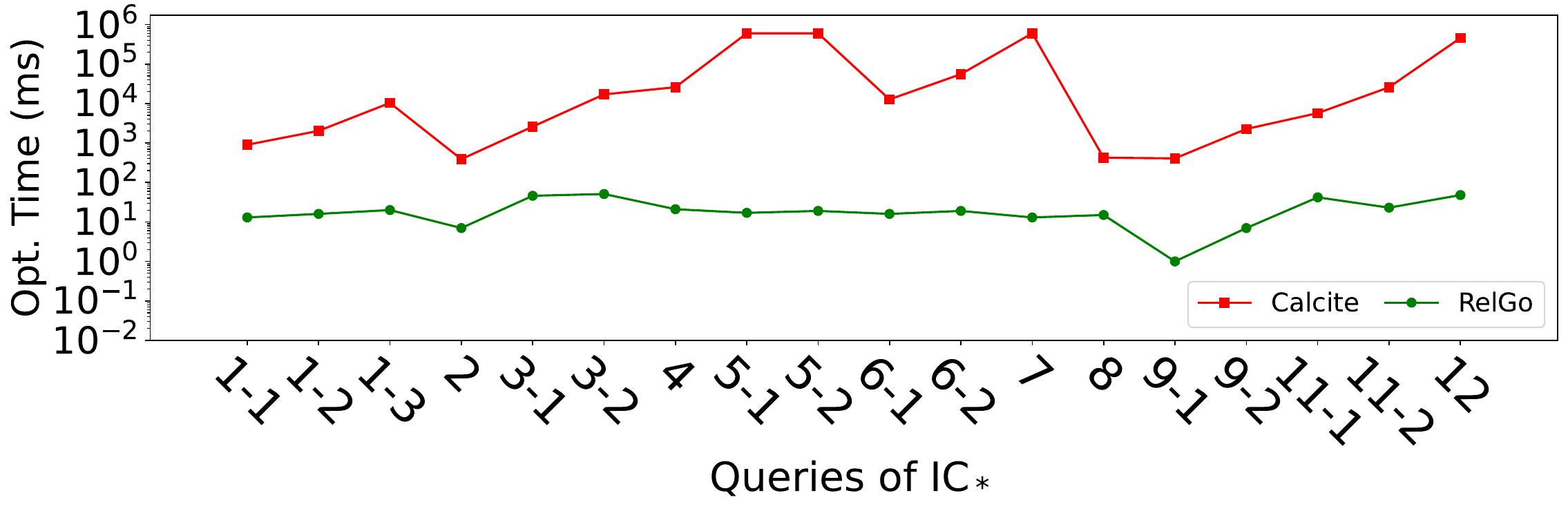}
        \vspace{-.9em}
        \caption{Optimization Time Cost on LDBC Queries.}
        \label{fig:exp-optimization-sf30}
    \end{subfigure}
    \caption{\revise{Compare the search space and optimization time}.}
    \label{fig:exp-optimization}
\end{figure}

After applying graph-agnostic transformations to the matching operator, the optimizer searches for the optimal join order. In contrast, applying graph-aware transformations leads to a search for the optimal decomposition tree. The search space for the graph-agnostic approach is clearly larger than that of the graph-aware approach, given the constraints imposed on the decomposition tree in the latter approach. However, the precise difference in search space complexity between the two approaches has not been rigorously analyzed.
In this subsection, we analyze the gap between the two search spaces and conclude that the graph-aware approach can be exponentially more efficient in this regard.

\comment{
According to \reflem{spjm-to-spj}, the graph-agnostic approach must produce a sequence of joins among $n$ vertex relations and $m$ edge relations,
which is equivalent to searching for the optimal join order among $n + m$ relations. The following lemma
establishes a lower bound on the search space complexity for this problem:

\begin{lemma}
\label{lem:complexity-of-volcano}
Given a join among $n + m$ relations, the search space for determining the optimal join order is at least $\Omega(4^{m+n-1})$.
\end{lemma}

\begin{proof}
    We first estimate the number of possible join orders, where each join order corresponds to a logical plan of join operators. The search space refers to the number of physical plans corresponding to these logical plans. To avoid cross products, for each explored logical plan, whenever two relations are joined together, there should be join conditions between them.

    Given $n + m$ relations, we construct a graph $\searchgraph = (V, E)$, where each relation corresponds to a vertex in $V$, and if there is a join condition between two relations, there is an edge between their corresponding vertices in $E$. Each possible logical plan forms a spanning tree in $\searchgraph$, and different logical plans may form the same tree. Therefore, the number of possible logical plans can be computed by obtaining all the possible logical plans corresponding to the spanning trees in $\searchgraph$.

    We consider the case where there is only one spanning tree $ST$ in $\searchgraph$ with $k$ edges. When there are multiple spanning trees in $\searchgraph$, we compute the number of logical plans corresponding to one of these trees, which provides a lower bound on the total number of possible logical plans.

    First, we examine the case where $ST$ is a path. Let $c_p(k)$ denote the number of logical plans corresponding to a spanning tree that is a path of length $k$. We have:
    \begin{equation*}
        c_p(k) = 2\sum_{i=0}^{i=k-1}c_p(i)c_p(k-1-i),
    \end{equation*}
    where $c_p(0) = 1$. Using the generating function, we obtain:
    \begin{equation*}
        c_p(k) = \frac{2^k}{k+1}\binom{2k}{k} \geq \frac{2^k}{k+1}2^{k-1}(k+1) = \frac{4^k}{2}
    \end{equation*}

    Next, we consider a more general scenario where $ST$ is not necessarily a path. Let $c(ST)$ denote the number of logical plans corresponding to $ST$. Suppose there are $k$ edges in $ST$. We denote the longest path in the tree by $p_1$, with length $P_1 = |p_1|$. By removing edges in $p_1$ from $ST$, we obtain a new subgraph $ST_1$. We then find the longest path $p_2$ in $ST_1$ that intersects with the already removed path $p_1$, with length $P_2 = |p_2|$. We remove edges in $p_2$ from $ST_1$ to obtain subgraph $ST_2$. Since $p_1$ and $p_2$ are both paths, the number of logical plans corresponding to them are $c_p(P_1)$ and $c_p(P_2)$, respectively. If $p_1$ and $p_2$ intersect at vertex $v_i$, the operator that scans $v_i$ appears in each logical plan corresponding to $p_1$. By replacing these scanning operators with the plans corresponding to $p_2$, we obtain $c(p_1 \cup p_2)$ plans, satisfying $c(p_1 \cup p_2) \geq c_p(p_1)c_p(p_2)$ because the relations corresponding to vertices in $p_1$ and those in $p_2$ can be joined in an interleaved fashion, which is overlooked by multiplying $c_p(p_1)$ and $c_p(p_2)$.

    As $ST$ is a tree, by repeatedly finding and removing paths as described above, all edges in $ST$ are eventually removed. Let $s$ be the number of paths removed. We have:
    \begin{equation*}
    \begin{split}
        c(ST) & = c(p_1 \cup \cdots \cup p_s) \geq c_p(P_1) \cdots c_p(P_s) \\
        & \geq \frac{4^{P_1 + \cdots + P_s}}{2^s} = \frac{4^{k}}{2^s} \geq 2^{k}.
    \end{split}
    \end{equation*}

    Since there are $m + n$ vertices in $\mathbb{G}$, the number of edges in the spanning tree is $k = m + n - 1$. Thus, the number of physical plans is at least $2^{k}t^{m+n-1} \geq 2^{m+n-1}t^{m+n-1} \geq 4^{m+n-1}$, which is also the search space of the problem. This concludes the proof.
\end{proof}

In contrast, when the graph-aware transformation is applied, the search space is equivalent to the number of possible decomposition trees. Despite the numerous works proposed to optimize graph pattern matching~\cite{huge,GLogS,mhedhbi2019optimizing}, to the best of our knowledge, the search space of this optimization problem has not been thoroughly analyzed. The following lemma provides an upper bound on the search space complexity for the graph-aware approach:

\begin{lemma}
\label{lem:complexity-of-graph-aware}
The search space for determining the optimal decomposition tree for pattern $\pattern$ is at most $O(4^{n-1})$.
\end{lemma}

\begin{proof}
    To prove the lemma, we construct a graph $\searchgraph_\pattern(V, E)$ to facilitate the analysis, where the vertex set contains all induced sub-patterns of $\pattern$ as well as an empty graph $\pattern_{\emptyset}$.
    We denote $V_i \subseteq V$ as the set of sub-patterns that contain exactly $i \leq n$ vertices. It is evident that $V_n = \{\pattern\}$.
    An edge exists from a larger (a pattern is considered larger if it contains more vertices) sub-pattern $\pattern_1$ to a smaller sub-pattern $\pattern_2$ if it is possible for $\pattern_2$ to be the child of $\pattern_1$ in any decomposition tree. Each edge has a weight representing the cost of extending graphs matching $\pattern_2$ to those matching $\pattern_1$.
    Moreover, there are edges weighted 0 from sub-patterns in $V_1$ to $\pattern_{\emptyset}$.
    Given the graph $\searchgraph_\pattern(V, E)$, 
    the search space equals the number of paths from $\pattern$ to $\pattern_{\emptyset}$.

    Determining the exact number of edges in $\searchgraph_\pattern$ for an arbitrary pattern is non-trivial. However, since we are studying the upper bound, we can consider the worst-case scenario. Given a sub-pattern $\pattern'$ consisting of $1 < i \leq n$ vertices, the number of induced subgraphs of $\pattern'$ with $j$ vertices is at most $\binom{i}{j}$, which constrains the maximum number of edges to $\pattern_j \in V_j$ that $\pattern'$ can connect.

    Denote the maximum number of paths from a sub-pattern with $s$ vertices (i.e., $\pattern_s \in V_s$) to $\pattern_{\emptyset}$ by $\decompnum_\pattern(s)$.
    Please note that $\decompnum_\pattern(1) = 1$ and $\decompnum_\pattern(n)$ is an upper bound of the search space for determining the optimal decomposition tree for $\pattern$.
    Then, we have the following inequality:
    \begin{equation*}
        \decompnum_{\pattern}(n) \leq \binom{n}{1}\decompnum_{\pattern}(n-1)\decompnum_{\pattern}(1) + \cdots + \binom{n}{n-1}\decompnum_{\pattern}(1)\decompnum_{\pattern}(n-1).
    \end{equation*}

    It is clear that $\decompnum_\pattern(i) \geq \decompnum_\pattern(j)$ if $i > j$.
    Therefore, we have
    \begin{equation*}
        \decompnum_{\pattern}(n) \leq \decompnumsq{2}_\pattern(n-1)(\binom{n}{1} + \cdots + \binom{n}{n-1}) < \decompnumsq{2}_\pattern(n-1)2^{n}
    \end{equation*}
    By recursively replacing $\decompnum_\pattern(n-1), \decompnum_\pattern(n-2), \cdots, \decompnum_\pattern(2)$, we can obtain the following inequalities:
    \begin{equation*}
        \begin{split}
            \decompnum_{\pattern}(n) & < (\decompnumsq{2}_\pattern(n-2)2^{n-1})^22^n = \decompnumsq{4}_\pattern(n-2)2^{n+2n-2} \\
            & < (\decompnumsq{2}_\pattern(n-3)2^{n-2})^42^{n+2n-2} = \decompnumsq{8}_\pattern(n-3)2^{n+2n-2 + 4n-8} \\
            & < \cdots \\
            & < \decompnumsq{2^{n-1}}_\pattern(1)2^{(1+\cdots+2^{n-2})n - (0+ 2 + \cdots + (n-2)2^{n-2})} < 4^{n-1}
        \end{split}
    \end{equation*}

    Therefore, the search space is at most $O(4^{n-1})$.
    We thus conclude the proof.
\end{proof}

    \begin{proof}
        \todo{refine the proof.}
        Given graph $\searchgraph_\pattern(V, E)$, the problem of determining the optimal decomposition tree for the graph-aware method can be transformed to finding the shortest path from $\pattern$ to $\pattern_{\empty}$ in $\searchgraph_\pattern(V, E)$.
        According to Dijkstra's algorithm, the time complexity of the shortest path query problem is $O(|E|)$.

        As the number of edges from patterns in $V_i$ is at most
        \begin{equation*}
            \sum\limits_{\pattern_i}(2^i - 2) = \binom{n}{i}(2^i - 2),
        \end{equation*}
        the upper bound of the number of edges in $\searchgraph_\pattern$ is
        \begin{equation*}
            \begin{split}
                \sum\limits_{i=2}^{i=n}\binom{n}{i}(2^i - 2) + n
                < \sum\limits_{i=0}^{i=n}\binom{n}{i}2^i = (2+1)^n
                 = 3^n.
            \end{split}
        \end{equation*}

        Therefore, the time complexity is at most $O(3^n)$.
        We thus conclude the proof.
    \end{proof}
}


\begin{theorem}
    \label{thm:compare-search-space}
    The search space in graph-aware transformation can be exponentially smaller than that of the graph-agnostic transformation,
    for optimizing the matching operator in an \spjm query.
\end{theorem}

\vspace*{-2mm}
\subsubsection{Comparison of Search Space and Optimization Time}
\revise{
To further illustrate \refthm{compare-search-space}, we used a special case of a path graph to compare the search spaces directly. We conducted a micro-benchmark experiment using a path graph with \( m \) edges, programming an enumerator to explore the search space of both graph-agnostic and graph-aware approaches while varying \( m \). The results, shown in \reffig{exp-search-space}, confirm the significant difference in search space size between the two approaches}.

Additionally, \revise{we compared the optimizer's query optimization time. In our comparison, Apache Calcite, a generic relational optimization framework, served as the optimizer for the graph-agnostic method. In contrast, our \name, implemented based on Calcite, acts as the optimizer for the graph-aware method. Both \name and Calcite are \crc{implemented} in Java, utilizing the VolcanoPlanner of Calcite with default rules.
Notably, we did not consider aggressive pruning rules as used in commercialized database like DuckDB~\cite{duckdb} for either Calcite or \name, providing a fair comparison and a clear demonstration of the reduced search space.
The optimization time was evaluated using the queries in our experiment (details in \refsec{evaluation}).
Optimizations that do not complete within 10 minutes are recorded as taking 10 minutes.
Since Calcite often exceeds the 10-minute limit on JOB queries\cite{full-version}, we only report the results on LDBC queries.
The results in \reffig{exp-optimization-sf30} indicate that \name can complete optimizing almost all queries within 10-100 milliseconds.
Besides, the results demonstrate \name's significant superiority over Calcite in query optimization speed.
For instance, \crc{on} $\text{IC}_{5-1}$, the optimization time using \name is more than \( 10^4 \) times faster compared to Calcite.}

\vspace*{-2mm}
\subsection{Physical Implementation}
\label{sec:physical-operators}

In the graph view, given a vertex $v$, it is efficient to obtain its adjacent edges and vertices (i.e., neighbors). However, in the relational view, such adjacency relationships between vertices and edges are not directly stored in relations but must be computed via the \EVjoin operations (\refeq{ev-join}). While there are multiple ways to construct the graph view in the literature~\cite{gart,GRFusion}, we refer to the method introduced in GRainDB~\cite{graindb}, which is free from materializing the graph. This approach avoids the extra storage cost associated with graph materialization and ensures compatibility with the relational context,
Specifically, GRainDB introduces an indexing technique called pre-defined join to improve the performance of join operations. As the pre-defined join essentially materializes the adjacency relationships, we treat it as a \emph{graph index} in this work.

\begin{figure}[t]
    \centering
    \includegraphics[width=.9\linewidth]{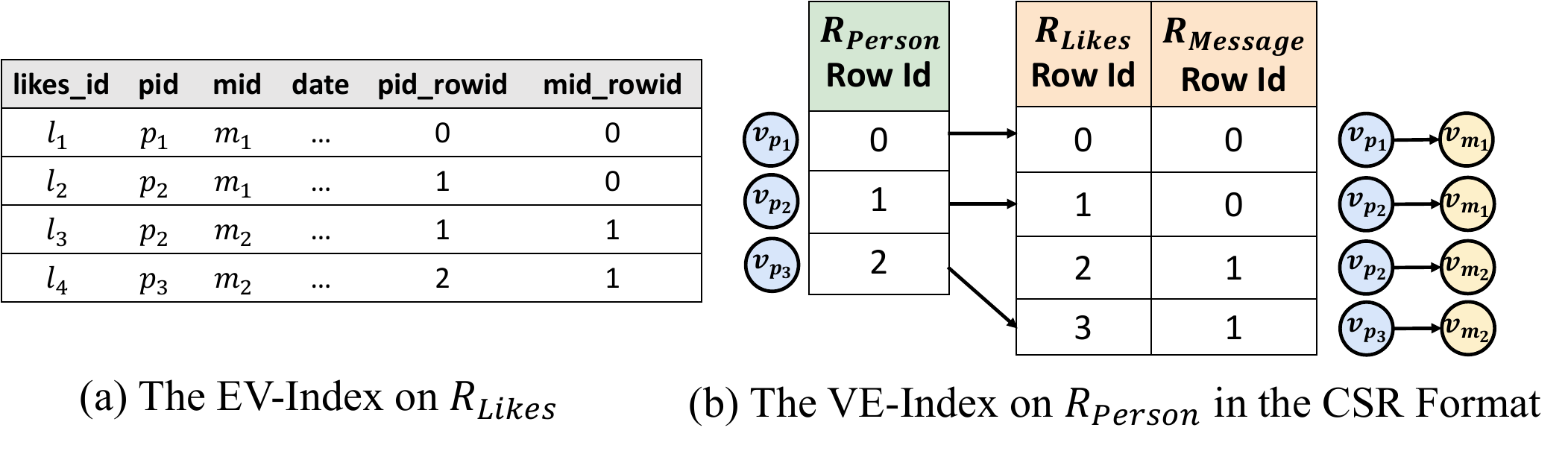}
    \caption{The graph index constructed among relations $\relation{\text{Person}}$, $\relation{\text{Likes}}$
    and $\relation{\text{Message}}$ in \reffig{intro-rgmapping-example}(a).}
    \label{fig:graph-index}
\end{figure}

\subsubsection{Graph Index}
\label{sec:graph-index}

As shown in \reffig{graph-index}, given the three relations $\relation{\text{Person}}$, $\relation{\text{Likes}}$, and $\relation{\text{Message}}$, the complete information of ``Person likes messages'' can be obtained by conducting the join:
\[ \relation{\text{Person}} \Join_\text{person\_id = pid} \relation{\text{Likes}} \Join_\text{mid = message\_id} \relation{\text{Message}}. \]


GRainDB introduces two kinds of indexes to the relational tables to efficiently process the join: the EV-index and the VE-index. The EV-index, shown in \reffig{graph-index}(a), is constructed by appending extra columns to the table $\relation{\text{Likes}}$. The column ``\text{pid\_rowid}'' stores the row ID of the corresponding tuple in the table $\relation{\text{Person}}$, denoted as $\rowid(\tau_{p})$, where $\tau_{p} \in \relation{\text{Person}}$. Similarly, the column ``\text{mid\_rowid}'' stores the row ID of the corresponding tuple in the table $\relation{\text{Message}}$, denoted as $\rowid(\tau_{m})$, where $\tau_{m} \in \relation{\text{Message}}$. These row ids help quickly route a tuple $\tau_{l} \in \relation{\text{Likes}}$ to the joinable tuples $\tau_{p}$ and $\tau_{m}$ without additional operations like hash-table lookup or sorting.

The VE-index in \reffig{graph-index}(b) is created on $\relation{\text{Person}}$ for efficiently computing its ``liked messages''. For each tuple $\tau_{p} \in \relation{\text{Person}}$, the VE-index records the row ids of $\tau_{l} \in \relation{\text{Likes}}$ and the corresponding $\tau_{m} \in \relation{\text{Message}}$ that are joinable with $\tau_{p}$. In the graph view, treating ``Person-[Likes]->Messages'' as an edge of a property graph, the VE-index maintains the adjacent edges and vertices of each person.

We can adopt GRainDB's approach to construct the graph indexes during the \rgmapping process. Given an edge relation $R_e$ and its associated vertex relations $R_{v_s}$ and $R_{v_t}$, the EV-index can be constructed on $R_e$ for each tuple $\tau_e \in R_e$ by including $\rowid(\lambda_e^s(\tau_e))$ and $\rowid(\lambda_e^t(\tau_e))$, which are the row ids of the corresponding tuples in $R_{v_s}$ and $R_{v_t}$, respectively. Meanwhile, the VE-index can be constructed on $R_{v_s}$ for each tuple $\tau_{v_s} \in R_{v_s}$ by including the row ids of all tuples $\tau_e \in R_e$ such that $\lambda_e^s(\tau_e) = \tau_{v_s}$, along with the row ids of the corresponding tuples $\tau_{v_t} \in R_{v_t}$ such that $\lambda_e^t(\tau_e) = \tau_{v_t}$.
The construction of VE-index on $R_{v_t}$ is analogous.

\comment{
\begin{remark}
  Several alternative approaches exist for handling the \rgmapping process. One such approach, proposed in~\cite{gart}, involves directly materializing the property graph using an additional graph store, rather than simply building graph indexes on the relational data. While this method may enable more efficient graph processing, it comes with the trade-off of requiring extra storage space and incurring higher maintenance costs. Moreover, \spjm queries often involve a combination of graph and relational operations, which may not be fully supported by the graph store alone.
\end{remark}
}


\vspace*{-3mm}
\subsubsection{The Graph-Aware Execution Plan}
\label{sec:join-matching-operator}
We delve into the physical implementation of the execution plan provided by the graph-aware method for solving $\matching(\pattern)$. The entry point of the plan is always matching a single-vertex pattern $\pattern_u$, which is one of the leaf nodes in the decomposition tree.


The implementation of $\matching(\pattern_u)$ is straightforward: scanning the corresponding vertex relation $\relation{\lab(u)}$ and encoding each tuple as a graph vertex object that contains its ID, label (mandatory) and necessary attributes. The row ID of the tuple in the relation can be directly used as the ID. To ensure globally uniqueness, the name of the relation can be incorporated as a prefix of the ID. Advanced encoding techniques are necessary for production use, but they are beyond the scope of this paper.

The plan is then constructed in a bottom-up manner. As shown in \reffig{match-decomposition}, there are three fundamental cases to consider when implementing the plan.

\stitle{Case I: Solving $\matching(\pattern') = \matching(\pattern'_l) \gjoin_{V_o, E_o} \matching(\pattern'_r)$}, where $\pattern_l'$ and $\pattern'_r$ are both intermediate patterns in the decomposition tree. The implementation of such a join is similar to a conventional relational join. The join is constrained to a natural join, where the join condition is simply the equality of the common vertices $V_o$ and edges $E_o$ between $\pattern_l'$ and $\pattern_r'$. During the implementation of the join, the identifiers of the vertices and edges can serve as the keys for comparison. Note that the input and output of the join are both graph relations, which will not
be projected into relational tuples until the last stage that obtains the results $\matching(\pattern)$.

\stitle{Case II: Solving $\matching(\pattern') = \matching(\pattern'_l) \gjoin_{u_s} \matching(\pattern_e)$}, where $\pattern_e$ is a single-edge pattern, and $u_s$ is the source vertex in $\pattern'_l$ from which the edge $e = (u_s, u_t)$ is expanded. Note that it's not possible for both $u_s$ and $u_t$ to be in $\pattern'_l$, as it would violate the fact that $\pattern'_l$ is either a single vertex or an induced sub-pattern.

\modify{When there is no graph index, $\matching(\pattern_e)$ is computed via $R_{\lab(u_s)} \evjoin R_{\lab(e)} \evjoin R_{\lab(u_t)}$}.
This case is then reduced to Case I. 

When graph indexes exist, the implementation is handled by the physical operators of \expandedge~ and \getvertex. For each tuple $\tau \in \matching(\pattern'_l)$, $\tau.u_s$ must record a graph vertex $v_s$ that matches $u_s$ in the pattern $\pattern'_l$. The \expandedge~ operator looks up the VE-index of $v_s$, which allows it to efficiently computes $v_s$'s adjacent edges (more precisely, it's the corresponding edge tuples). Furthermore, the \getvertex~ operator is used to obtain the matched vertex $v_t$ that is connected to $v_s$ via the previous matched edges, which can be achieved by looking up the EV-index of the matched edges.
By combining the results of \expandedge~ and \getvertex, the tuple of $(\tau, \adj^E(v_s), \adj(v_s))$ is rendered. For example, in \reffig{graph-index}(b), if we apply \expandedge~ and \getvertex~ to a tuple $\tau$ from $v_{p_2}$, the result $(\tau, [e_{l_2}, e_{l_3}], [v_{m_1}, v_{m_2}])$ is returned.
Furthermore, to obtain $\matching(\pattern')$, we flatten the adjacent edges and vertices and pair them up. In the case of $(\tau, [e_{l_2}, e_{l_3}], [v_{m_1}, v_{m_2}])$, two tuples $(\tau, e_{l_2}, v_{m_1})$ and $(\tau, e_{l_3}, v_{m_2})$ are generated.

In practice, a vertex may be adjacent to multiple types of edges. For example, in \reffig{intro-rgmapping-example}, a \kk{Person} vertex can be connected to both \kk{Likes} and \kk{Knows} edges. To handle such cases, we can record edge's ID instead of just the row ID of the tuple. Given that the edge's ID is a combination of its label and the tuple's row ID, the adjacent edges of a specific label can be easily obtained from the VE-Index.

\stitle{Case III: Solving $\matching(\pattern') = \matching(\pattern'_l) \gjoin_{V_s, E_s} \matching(\pattern(u;V_s))$}, where pattern $\pattern(u;V_s)$ is a complete $k$-star with $V_s = \{u_1, \ldots, u_k\}$. 

When there is no graph index, solving Case III involves continuously joining $|V_s|$ single-edge patterns.
When graph indexes are available, the \expandintersect~ operator can be used to efficiently compute the join.
\revise{Unlike HUGE~\cite{huge}, which has a graph storage that naturally supports \expandintersect, we have implemented this operator directly on a relational database}.
Given a tuple $\tau \in \matching(\pattern'_l)$, let $\{v_1, \ldots, v_k\}$ be the vertices in $\tau$ that match the leaf vertices $\{u_1, \ldots, u_k\}$ in the complete star $\pattern(u;V_s)$.
Vertices matching the root vertex $u$ of the star must be common neighbors of all the leaf vertices.

Consequently, for the tuple $\tau$, the physical \expandintersect~ operator performs the following steps:

\comment{When graph indexes are available, the \expandintersect~ operator, introduced in the literature~\cite{huge,GLogS,mhedhbi2019optimizing}, can be used to efficiently compute the join. Given a tuple $\tau \in \matching(\pattern'_l)$, let $\{v_1, \ldots, v_k\}$ be the vertices in $\tau$ that match the leaf vertices $\{u_1, \ldots, u_k\}$ in the complete star $\pattern(u;V_s)$. The vertices that can match the root vertex $u$ of the star must be the common neighbors of all the leaf vertices.}


\begin{enumerate}
\item For each leaf vertex $u_i \in V_s$ ($1 \leq i \leq k$), apply the \expandedge~ \\ and \getvertex~ operators to obtain the adjacent edges and neighbors of the corresponding vertices $v_i$ respectively.
\item Compute the intersections of all adjacent edges and neighbors returned by the \expandedge~ and \getvertex~ operators.
\item Return a new tuple as follows; for the sake of simplicity, the details of the edges are omitted: $(\tau, \bigcap\limits_{1 \leq i \leq k}\adj(v_i))$.


\end{enumerate}

Note that the above step (1) and (2) can be computed in a pipeline manner, following a certain order of among the leaf vertices.
Similar to Case II, we flatten the common edges and vertices and pair them up to obtain the final result.

\begin{example}
    \label{ex:physical-implementation}
    Given $\pattern$ in \reffig{match-decomposition}, a decomposition tree and its corresponding logical plan are presented. We illustrate the physical implementation of $\matching(\pattern_1) \gjoin \matching(\pattern_2)$ using \expandintersect when a graph index is available.
    Consider the tuple $(v_{p_1}, e_{k_1}, v_{p_2})$ from $\matching(\pattern_1)$ as an example. First, the \expandedge~ and \getvertex~ operators are applied to obtain the adjacent edges and neighbors of $v_{p_1}$ and $v_{p_2}$, resulting in
    \begin{equation*}
        \begin{split}
    &(v_{p_1}, e_{k_1}, v_{p_2}, [e_{l_1}], [v_{m_1}]) \text{ and }
    (v_{p_1}, e_{k_1}, v_{p_2}, [e_{l_2}, e_{l_3}], [v_{m_1}, v_{m_2}]).
        \end{split}
    \end{equation*}
    Next, the intersection process is conducted. Since $\adj(v_{p_1}) \cap \adj(v_{p_2}) = [v_{m_1}]$, the edges in both sets that have $v_{m_1}$ as the target vertex are retained, resulting in $(v_{p_1}, e_{k_1}, v_{p_2}, [(e_{l_1}, e_{l_2}, v_{m_1})])$. Finally, the tuple is flattened to $(v_{p_1}, e_{k_1}, v_{p_2}, e_{l_1}, e_{l_2}, v_{m_1})$.

    \comment{
    Case II: $\matching(\pattern_3) \gjoin \matching(\pattern_4)$ is implemented using \expandedge~and \getvertex. Take $v_{p_1} \in \matching(\pattern_3)$ as an example, we first do \expandedge by looking up the
    VE-index corresponding to $v_{p_1}$ in \reffig{graph-index}(b), which
    locates the edge $e_{}$ the $0^{th}$ row of $\relation{\text{Likes}}$. Then, \getvertex~is computed by looking up the EV-index \reffig{graph-index}(a), which locates the vertex tuples corresponding to $v_{m_1}$

    is implemented and it pertains to Case II, since the $\pattern_6$ is a single-edge pattern.
    Thus, when there is no graph index on $\relation{Likes}$, the results of $\matching(\pattern_4)$ is computed via $\relation{Person} \Join_{person\_id = pid1} \relation{Knows} \Join_{pid2 = person\_id} \relation{Person}$.
    Specifically, the results include tuples $(v_{p_1}, e_{k_1}, v_{p_2})$, $(v_{p_2}, e_{k_2}$, $v_{p_1})$, $(v_{p_2}, e_{k_3}, v_{p_3})$, and $(v_{p_3}, e_{k_4}, v_{p_2})$.
    Then, since $\pattern_3$ is a single-vertex pattern whose vertex exists in $\pattern_4$, the results of the join equals $\matching(\pattern_4)$.
    Otherwise, if a graph index exists, $\matching(\pattern_4)$ can be computed with \expandedge~and \getvertex.
    }
\end{example}

%% file: sec-framwork.tex
\section{The Converged Optimization Framework}
\label{sec:optimizations}
This section presents \name, a converged relational/graph optimization framework designed to optimize the query
processing of \spjm queries. We begin by introducing a naive solution built upon the graph-agnostic
method for solving the matching operator. We then delve into the converged workflow of \name, which leverages the graph-aware method for solving the matching operator and introduces a complete workflow that aims to integrate techniques from both relational and graph optimization modules.


\begin{figure*}
    \centering
    \includegraphics[width=\linewidth]{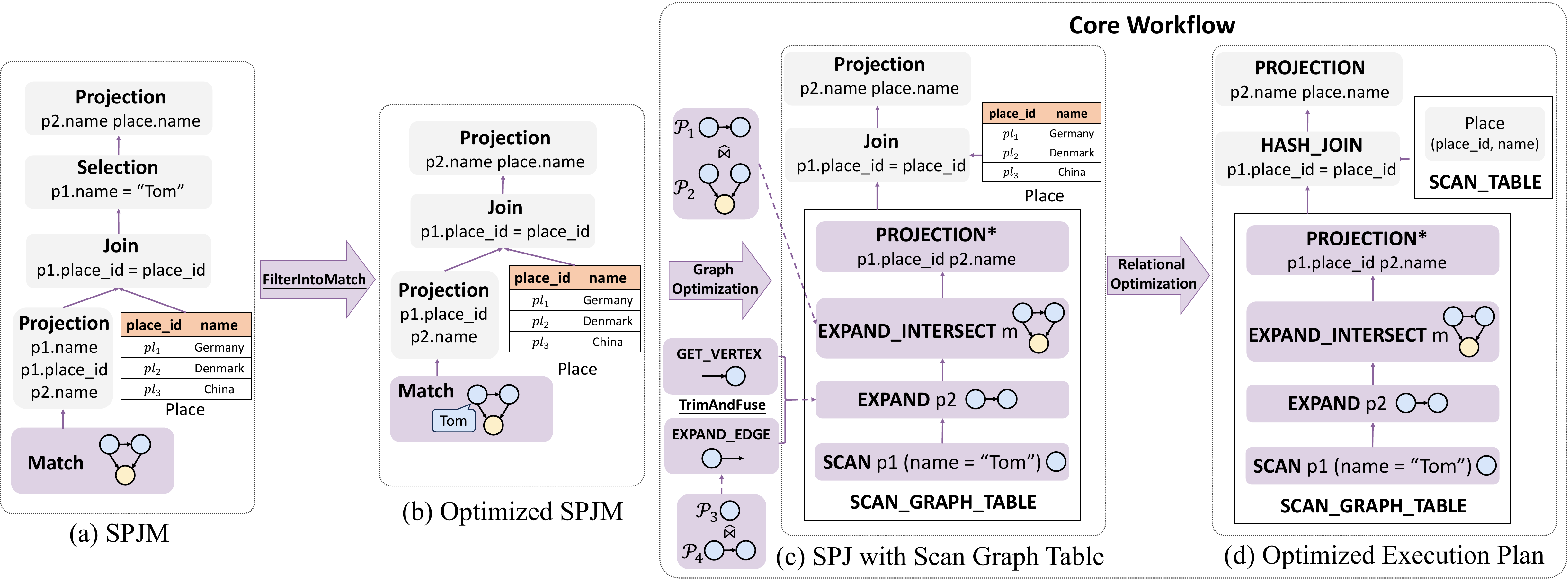}
    \caption{The converged optimization workflow}
    \label{fig:framework-workflow}
\end{figure*}

\subsection{Graph-Agnostic Approach}
\label{sec:relational-only}
The graph-agnostic approach is straightforward: it applies the graph-agnostic transformation for the matching operator in an \spjm query into a series of relational operations (\modify{\reflem{spjm-to-spj}}), effectively converting the \spjm query into an \spj query. The resulting \spj query can then be optimized by any existing relational optimizer, producing an execution plan. As an improvement, if a graph index (\refsec{graph-index}) is available, certain hash-join operators in the execution plan can be replaced by the predefined-join operator, as discussed in GRainDB~\cite{graindb}. The main advantage of this solution is its easy integration with any existing relational database. However, it suffers from two significant drawbacks discussed in \refrem{graph-agnostic-vs-graph-aware}.

\subsection{The Converged Approach}
\label{sec:converged}
As illustrated in \reffig{framework-workflow}, the core workflow of the \name framework consists of two components: \emph{graph optimization} and \emph{relational optimization}. The graph optimization is responsible for handling the graph component in an \spjm query, leveraging graph optimization techniques to determine the optimal decomposition tree of the matching operator. On the other hand, the relational optimization takes over to optimize the relational component in the query.
The order in which these two components are applied is not strictly defined. However, for the purpose of our discussion, we will first focus on the graph optimization and then proceed to the relational optimization.
In addition to the core workflow, we further explore heuristic rules that highlight the non-trivial interplay between the relational and graph components in an \spjm query.




\vspace*{-1mm}
\subsubsection{The Graph Optimization}
\label{sec:graph-optimizer}
We adopt the graph optimization techniques developed in \glogs~\cite{GLogS}. However, it is crucial to note that \glogs~ was originally designed for native graph data, whereas our framework deals with relational data, which necessitates a careful adaptation of \glogs's techniques to the relational setting.

\stitle{\glogue~Construction.} \glogs is built upon a data structure called \glogue, which is essentially a graph $\searchgraph_{\pattern}(V, E)$. In this graph, each vertex represents a pattern $\pattern'$ consisting of up to $k$ vertices (typically, $k=3$) that has non-empty matched instances in the original graph. There is an edge from $\pattern''$ to $\pattern'$, if there is a decomposition tree where $\pattern''$ is a child node of $\pattern'$.

 Each vertex $\pattern'$ in \glogue maintains $|\matching(\pattern')|$, denoting the cardinality of the pattern. To reduce computation costs, \glogs employs a sparsification technique to construct a subgraph $G'$. The pattern cardinality can then be estimated using $|\matching_{G'}(\pattern')|$ based on subgraph $G'$. In our work, we adapt this sparsification technique to construct \glogue. We sample a subset of vertex and edge relations in the \rgmapping process. Once the subset of relations is obtained, they can serve as the input tables to the techniques presented in~\cite{gart} for constructing the sparsified graph $G'$.


\stitle{Cost Calculation.} The optimization process is essentially searching for the execution plan that incurs the minimal cost. Let the cost of an execution plan $\Phi$ for computing $\matching(\pattern)$ be $\cost_\Phi(\pattern)$. 

Consider $\matching(\pattern') = \matching(\pattern'_l) \gjoin \matching(\pattern'_r)$ as an intermediate computation in an execution plan. We have:
\[
\cost_{\Phi}(\pattern') = \cost_{\Phi_l}(\pattern'_l) + \cost_{\Phi_r}(\pattern'_r) + \cost(\gjoin),
\]
where $\Phi_l$ and $\Phi_r$ are the execution plans for computing $\matching(\pattern'_l)$ and $\matching(\pattern'_r)$, respectively, and $\cost(\gjoin)$ is the cost of the join operation.

When a graph index is available, there are three physical implementations of $\gjoin$, depending on the type of $\pattern'_r$, and the calculation of $\cost(\gjoin)$ differs accordingly:
\begin{itemize}
\item If $\pattern'_r$ is a single-edge pattern, $\gjoin$ is implemented using the \expandedge~ operator followed by \getvertex. The cost is calculated based on the cardinality of $\matching(\pattern'_l)$ (can be looked up in the \glogue) and the average degree of the graph, namely $|\matching(\pattern'_l)| \times \overline{d}$.
\item If $\pattern'_r$ is a complete star pattern, $\gjoin$ is implemented using the \expandintersect~ operator. The cost is calculated based on the cardinality of $\matching(\pattern'_l)$ and the average intersection size of the neighbors of the vertices being intersected, which is maintained on the corresponding edge from $P'$ to $\pattern'_l$ in \glogue.
\item  If $\pattern'_r$ is any arbitrary pattern, $\gjoin$ is implemented as a \hashjoin. The cost is calculated as the product of the cardinalities of the two relations being joined, i.e., $\cost(\gjoin) = |\matching(\pattern'_l)| \times |\matching(\pattern'_r)|$.
\end{itemize}

In the absence of a graph index, \hashjoin~ is used for the entire plan of the matching operator for simplicity, and its cost is computed as the product of the cardinalities of the two relations being joined. Although other physical join implementations, such as nested loop join, may be more effective if the join condition is not selective, considering these alternatives is planned for future work.


\comment{
\begin{example}
    As shown in \reffig{framework-workflow}, since $\matching(\pattern_2)$ is a single-edge pattern, the join operator between $\matching(\pattern_1)$ and $\matching(\pattern_2)$ is implemented using \expandedge~followed by \getvertex.
    Besides, as $\matching(\pattern_4)$ is a complete star, the join operator between $\matching(\pattern_3)$ and $\matching(\pattern_4)$ is implemented using the \expandintersect~operator.
\end{example}
}

\stitle{Plan Computation.} Searching for the optimal execution plan in \name remains the same as in \glogs. The optimal plan is obtained by searching for the shortest path in the \glogue from the single-vertex pattern to the queried pattern.  
\reffig{framework-workflow}(c) demonstrates a physical plan for matching the given triangle pattern when a graph index is present. The plan reflects the example in~\refex{physical-implementation}, with one exception: the pair of \expandedge~ and \getvertex~ operators is fused into a single \expand~ operator, which will be discussed as a heuristic rule called \joinfuserule.

\subsubsection{The Relational Optimization}
Once the graph optimizer has computed the optimal execution plan for $\matching(\pattern)$, the next step is to integrate this plan with the remaining relational operators in the \spjm query. The relational optimization is responsible for optimizing these remaining operators, which are all relational operators.
\revise{Relational optimization has evolved into a well-established field, producing numerous significant results \cite{Chaudhuri98, Haffnerjoinorder}. Since existing relational optimization techniques can be seamlessly integrated into \name, we will focus on how graph optimization techniques can be applied to enhance relational queries.}


Specifically, to prevent the relational optimizer delve into the internal details of the graph pattern matching process, we introduce a new physical operator called \scangraphtable, as shown in \reffig{framework-workflow}(c), which encapsulates the $\gproject$ operator and the optimal execution plan for $\matching(\pattern)$.
The \scangraphtable~ operator acts as a bridge between the graph and relational components of the query. From the perspective of the relational optimizer, \scangraphtable~ behaves like a standard \scan~ operator, providing a relational interface to the matched results. 

\subsubsection{Heuristic Optimization Rules}
In real-life use cases, heuristic rules may involve non-trivial interactions between the relational and graph components of an \spjm query. We explore two representative rules, \filterrule~ and \joinfuserule, which can be applied at different stages of the optimization process to improve query performance.

\stitle{\filterrule.} To elaborate the rule, we extend the definition of a pattern $(\pattern, \constraints)$, introducing constraints within $\constraints$. For example, constraints can specify predicate $d$ such as $\id(v_1) = p_1$ for a vertex $v_1$, or $e_1.date > \text{"2024-03-31"}$ for an edge $e_1$. With the constraints defined, any matching result of $\pattern$ must have the corresponding vertices and edges adhering to the predicates.


While writing queries, users may not specify constraints on the pattern but rather use the selection operator after matching results have been projected into the relational relation, described as:
\[
\sigma_{d'_{v_a}} (\widehat{\pi}_{v.a \rightarrow \text{v\_a}, \ldots} \matching(\pattern))
\]
The predicate $d'_{v_a}$ defines a predicate in terms of an attribute of the pattern vertex that is projected by $\widehat{\pi}$ from the matched results. The motivation example in \refex{introduction:sqlpgq} illustrates such a case, where the selection predicate \kk{g.p1\_name = ``Tom''} is applied to the pattern vertex $v_{p_1}$.
There is wasteful computation if the selection is applied after the costly pattern matching. A more efficient approach is to push the selection predicate down into the matching operator.
The \filterrule is formally defined as:
\begin{equation*}
\sigma_{\constraints} (\widehat{\pi}_{v.a \rightarrow \text{v\_a}, \ldots} \matching(\pattern)) \\
\equiv \sigma_{\constraints'} (\widehat{\pi}_{v.a \rightarrow \text{v\_a}, \ldots} \matching((\pattern, \{d_v\}))),
\end{equation*}
where $\constraints' = \constraints \setminus \{d'_{v_a}\}$, and $\{d_v\}$ is the corresponding constraints that are appended to the pattern $\pattern$.

It is recommended to apply the \filterrule before graph optimization, as this allows the optimizer to leverage the pushed-down constraints to recalculate the cost, potentially generating more efficient execution plans. \reffig{framework-workflow}(b) showcases the effects of applying the \filterrule, where the selection predicate \kk{g.p1\_name = ``Tom''} is pushed down into the matching operator.

\stitle{\joinfuserule.}
The \joinfuserule~ is utilized to streamline a query plan by merging the \expandedge~ and \getvertex~ operators which are commonly coupled in the implementation of matching operations, into a single \expand~ operator that retrieves the neighboring vertices directly.
However, such a fusion is permissible solely when the output edges by \expandedge~ are deemed unnecessary, so this rule further incorporates a preceded field trim step.
Specifically, the field trimmer would examine whether any subsequent relational processes rely on these edges, such as utilizing them for property projections or for filtering based on their attributes.
If no such operations are found, the edges can be trimmed.
Furthermore, the field trimmer would also consider a special case that the edges might be projected in the \scangraphtable~ operator as part of the matching results, but are subsequently unused in relational operations. In such cases, the edges can be trimmed as well.
After the field trim step, if the output edges are trimmed, the \expandedge~ operator can be fused with the \getvertex~ operator to form a single \expand~ operator, which can directly retrieve the neighboring vertices efficiently by looking up the VE-index of the source vertex when the graph index is available.

%


Note \revise{that \filterrule is actually a global optimization rule because there are cases where pushing the predicate into the matching operator does not always yield better plans\cite{full-version}. However, since it is mostly effective, we greedily apply \filterrule in the current version. A more comprehensive evaluation of this rule will be conducted in future work. }
On the other hand, \joinfuserule is a local optimization rule specifically designed for graph optimization. The effectiveness of these two rules is validated in \refsec{experiment-opt}. Our \name framework is designed to be generic, allowing different optimization rules to be easily integrated. 

\vspace*{-2mm}
\subsection{System Implementation}
We engineered the frontend of \name in Java and built it upon Apache Calcite~\cite{calcite}
to utilize its robust relational query optimization infrastructure.
Firstly, we enhanced Calcite's SQL parser to recognize SQL/PGQ extensions, specifically to parse the \lstinline{GRAPH_TABLE} clause.
We created a new \lstinline{ScanGraphTableRelNode} that inherits from Calcite's core \lstinline{RelNode} class, translating the \lstinline{GRAPH_TABLE} clause into this newly defined operator within the logical plan.
Following the formation of the logical plan, the frontend invokes the converged optimizer to generate the optimal physical plan.
For the relational-graph interplay optimizations, we incorporate heuristic rules such as \filterrule and \joinfuserule into Calcite's rule-based HepPlanner, by specifying the activation conditions and consequent transformations of each rule.
For more nuanced optimization, we rely on the VolcanoPlanner, the cost-based planner in Calcite, to optimize the \lstinline{ScanGraphTableRelNode}.
We devised a top-down search algorithm that assesses the most efficient physical plan based on a cost model outlined in \refsec{graph-optimizer}, combined with high-order statistics from \glogue for more accurate cost estimation.
\revise{
    While low-order statistics primarily focus on the cardinalities of relational tables, high-order statistics also include the frequencies of sub-patterns (can be seen as the joined results of multiple tables of vertices and edges), which aids in more accurate cost estimation. It is important to note that \name remains functional with only low-order statistics, but the efficiency of the generated plan may decrease due to less accurate cost estimation.
}

For the remaining relational operators in the query, we leverage Calcite's built-in optimizer, which already includes comprehensive relational optimization techniques.
Lastly, the converged optimizer outputs an optimized and platform-independent plan formatted with Google Protocol Buffers (protobuf) \cite{protobuf}, ensuring the adaptability of \name's output to various backend database systems.

We developed the \name framework's backend in C++ using DuckDB as the relational execution engine to showcase its optimization capabilities.
We integrated graph index support in GRainDB~\cite{graindb}. 
With graph index, the \expand, \expandedge~ and \getvertex~ operators can be optimized by directly using the predefined join in GRainDB.
Note that we craft a new join on DuckDB called \emph{EI-Join} for the support of \expandintersect.
Without graph index, the \hashjoin~ operator is used throughout the entire plan.
To execute the optimized plans within DuckDB, we introduced a runtime module that translates the optimized physical plan into a sequence of DuckDB/GRainDB-compatible executable operators.
This runtime module essentially bridges the gap between the optimized plans produced by \name and DuckDB's execution engine, thereby validating \name's practicality and potential performance improvements for \spjm queries on an established relational database system.


\comment{
In the Converged Optimizer,
we incorporated both graph optimization techniques and relational optimization techniques, to optimize the \spjm queries in a converged manner.
We use the rule-based optimization planner named HepPlanner in Calcite, to plug in heuristic optimization rules like \filterrule and \joinfuserule. We set up each rule by specifying the conditions under which it activates and the actions it takes once those conditions are met.
For more nuanced optimization, we employ the VolcanoPlanner, a cost-based optimization planner provided by Calcite, to optimize the match operator in a cost-based manner.
We have crafted a top-down search algorithm that calculates the most efficient physical plan for the match operator, leveraging the cost model outlined in \refsec{graph-optimizer}. The algorithm is informed by advanced statistics obtained from the Metadata Provider, ensuring a more precise cost estimation.
For the relational part, Calcite's built-in optimizer, which already includes comprehensive relational optimization techniques, is leveraged to fine-tune the rest of the relational operators in the query.
The Converged Optimizer generates its output as a unified plan formatted with Google Protocol Buffers (protobuf)\cite{protobuf}. This serialization format is both platform-independent and highly interoperable, facilitating the easy conversion of the unified plan into executable query plans tailored for the destination database system.
}

\comment{
To optimize SPJM queries, we propose the converged graph relational optimization framework named relgo.
Specifically, optimizing SPJM queries with relgo is divided into three stages, i.e., preprocessing, optimizing, and converting.
In this section, we introduce the three stages of relgo in detail.

\subsection{Preprocessing Stage}
Given an SPJM query, relgo first parse the query and obtained the corresponding AST (Abstract Syntax Tree).
Then, the initial logical plan is obtained based on the AST.
Each node in the logical plan represents an operator, including the selection, projection, join, and scan operators.

Please note that the matching operator does not appear in logical plans, because it can be further decomposed into other operators such as join and selection operators.
Besides, according to \refdef{matching}, the output of the matching operator is a graph relation, and it is always followed by a projection operator $\widehat{\pi}$.
The output of $\widehat{\pi}$ is a relation and the matching operator as well as $\widehat{\pi}$ are considered as a whole as an implementation of the scan operator (named ScanMatchTable).

In the preprocessing stage, some universally effective optimizations can be applied to refine the plan in advance.
A typical optimization is \filterrule.

This rule is inspired by FilterPushdownRule in relational optimizer, which can push down the predicates to the scan operators to filter out invalid elements earlier.
Specifically, as the ScanMatchTable corresponding to $\matching(\pattern)$ is a physical implementation of the scan operator which acts like scanning a table obtained by matching $\pattern$, it is reasonable to integrate some filtering criteria into the ScanMatchTable operator.
Specifically, \filterrule finds predicates on the properties of elements in $\pattern$ and push them down into ScanMatchTable, so that invalid elements can be dropped earlier.
An example of applying \filterrule is given in Example \ref{example:push_down}.

Formally, the equation rule w.r.t.~\filterrule is as follows:
\begin{equation}
    \begin{split}
        & \pi_A(\sigma_{d}(R_1 \Join \cdots \Join R_m \Join \widetilde{R}) \\
        & \hspace*{2em} \equiv \pi_A(\sigma_{d_0}(R_1 \Join \cdots \Join R_m \Join \widetilde{R}_{d_1})), \\
        & \hspace*{4em} \text{where } \widetilde{R} = \widehat{\pi}_{attr*}(\mathcal{M}(\mathcal{P})) \\
        & \hspace*{4em} \text{and } \widetilde{R}_{d_1} = \widehat{\pi}_{attr*}(\mathcal{M}(\mathcal{P}_{d_1}))
    \end{split}
\end{equation}
where $d_1$ is a subset of $d$ with the constraints related to $\widetilde{R}$ and $d_0$ is obtained by removing $d_1$ from $d$.
Besides, $\mathcal{P}_{d_1}$ is obtained by adding constraints in $d_1$ to $\mathcal{P}$.

\subsection{Optimizing Stage}

Given a logical plan of a SPJM query, relgo optimizes the plan in the optimizing stage.
Specifically, the operations that appear in logical plans of SPJM queries often also appear in the logical plans of relational databases.
Therefore, typical relational optimizers, e.g., Calcite \cite{calcite}, can be employed to optimize this plan.

Besides, to optimize the implementation of ScanMatchTable, the matching plans are optimized with graph-aware methods.
Studies that optimize graph pattern matching can be utilized as the graph-aware methods, and when implementing \name, we leverage GLogS \cite{GLogS} to optimize matching plans.
Moreover, if graph indices are available, join operators in \expandvertex and \expandintersect can be implemented as predefined joins to further improve the efficiency.

\subsection{Converting Stage}

As different databases usually support different operators and their physical plans can be greatly varied, it is of critical importance for an optimization framework to be flexible.
Therefore, we implement a PlanConverter in the framework to ensure the flexibility.
Given the generated optimal physical plan, the PlanConverter transforms the plan to an internal representation (e.g., Substrait \cite{substrait}), and then the internal representation is transformed to the physical plan that can be executed by the target database.
Finally, the plan is executed and the query results are obtained.

The introduction of the framework is concluded with an example.

\begin{figure}
    \centering
    \includegraphics[width=.8\linewidth]{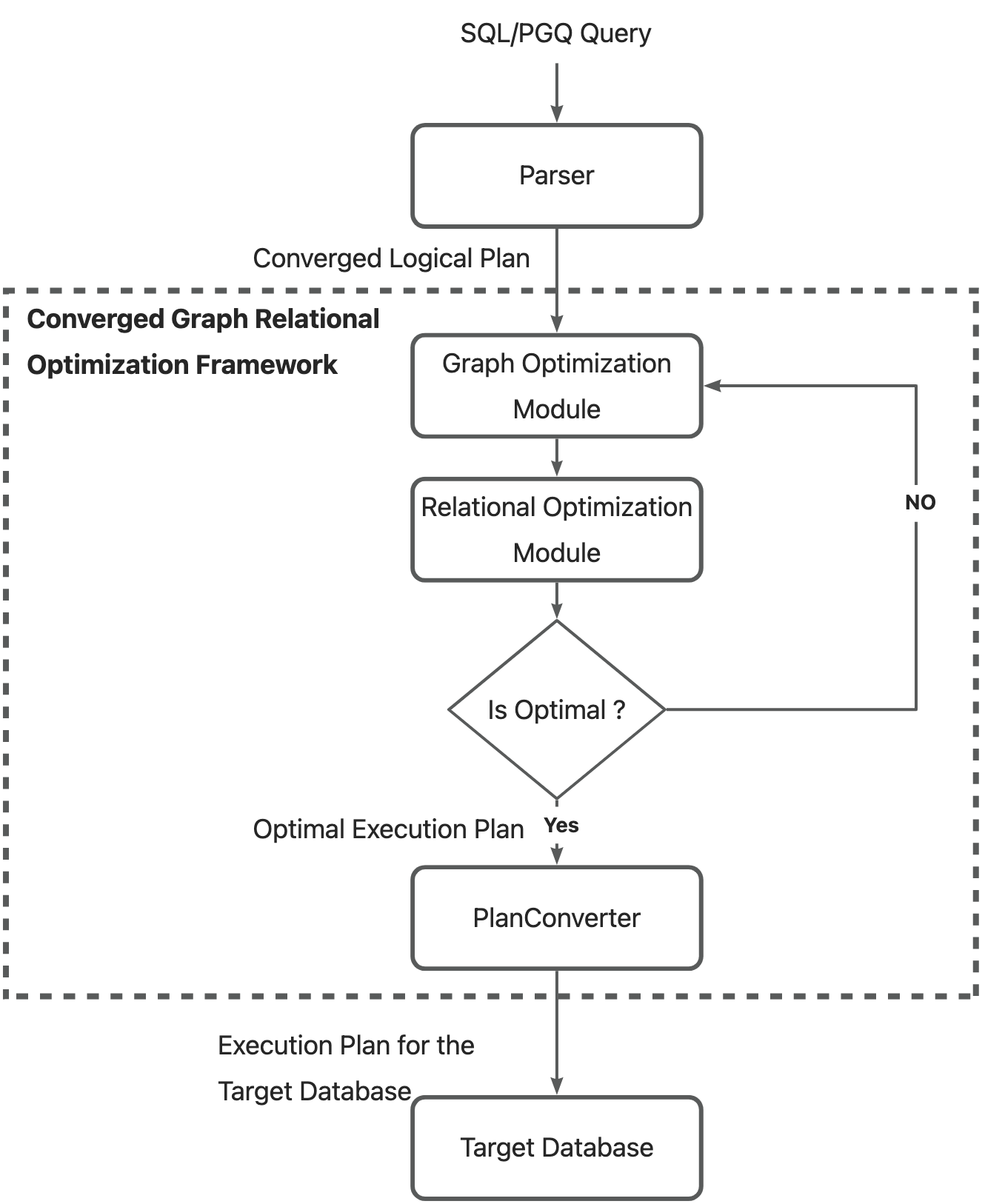}
    \caption{Workflow of the Converged Graph Relational Optimization Framework.}
    \label{fig:workflow}
\end{figure}

\begin{figure}
    \centering
    \includegraphics[width=.6\linewidth]{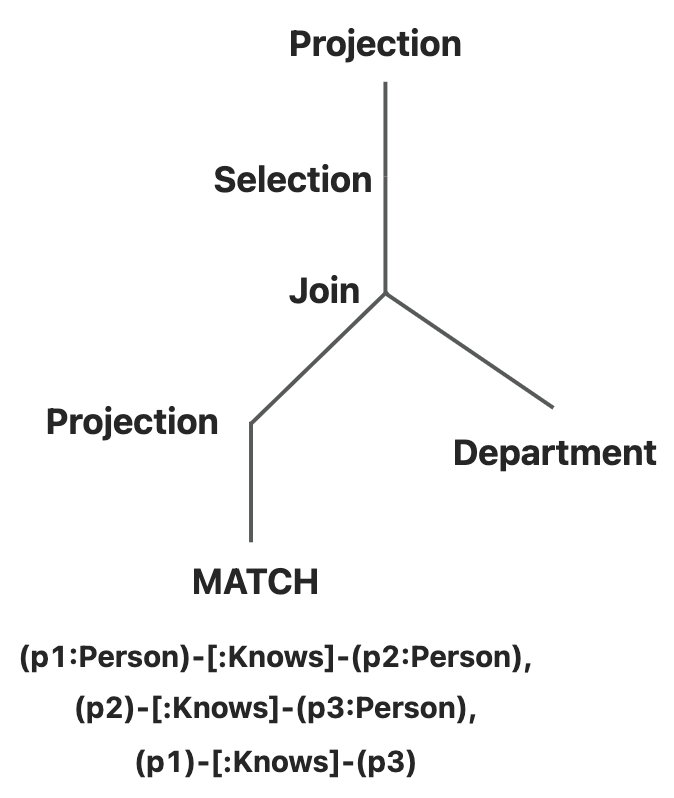}
    \caption{The operator tree of SPJM query in Example \ref{example:framework}.}
    \label{fig:example-operator-tree}
\end{figure}

\begin{figure*}
    \centering
    \begin{subfigure}[b]{0.4\linewidth}
        \centering
        \includegraphics[width=\linewidth]{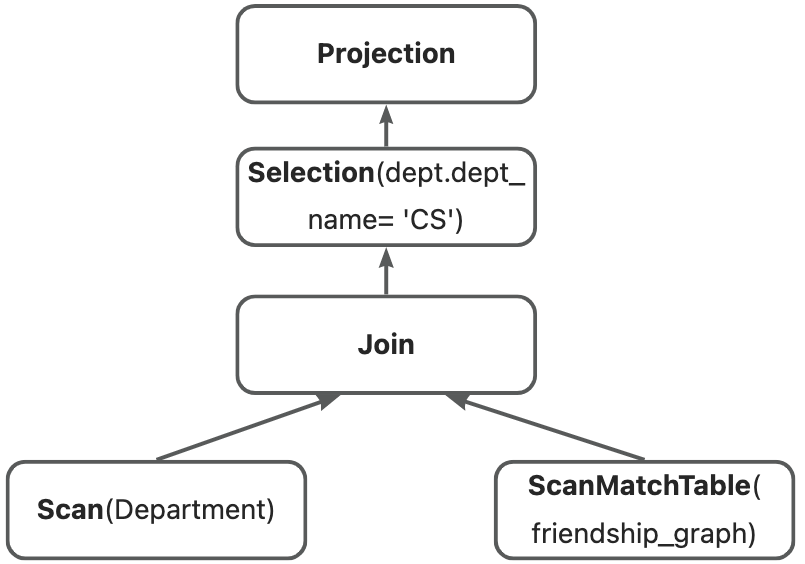}
        \caption{Outer query plan.}
        \label{fig:converged-logical-plan-relational}
    \end{subfigure}
    \begin{subfigure}[b]{0.4\linewidth}
        \centering
        \includegraphics[width=\linewidth]{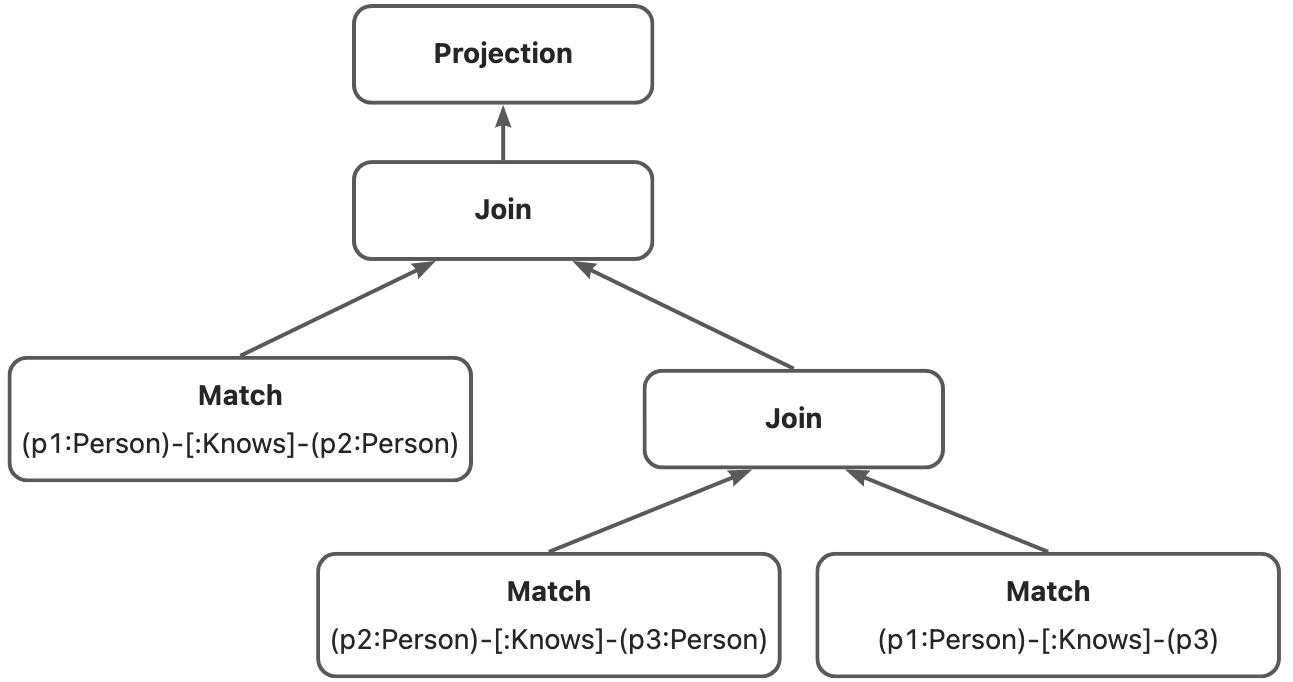}
        \caption{Match scanning plan.}
        \label{fig:converged-logical-plan-graph}
    \end{subfigure}
    \begin{subfigure}[b]{0.4\linewidth}
        \centering
        \includegraphics[width=\linewidth]{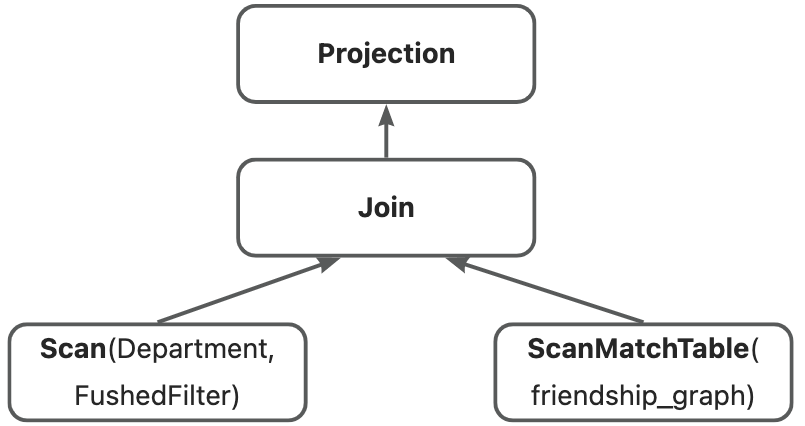}
        \caption{Outer query plan after Optimization.}
        \label{fig:relational-plan-optimized}
    \end{subfigure}
    \begin{subfigure}[b]{0.4\linewidth}
        \centering
        \includegraphics[width=\linewidth]{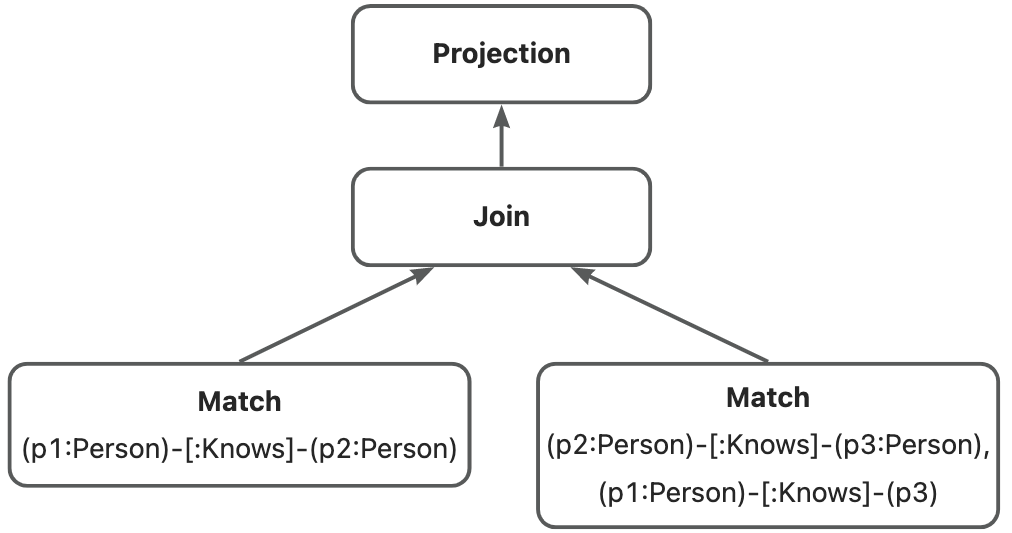}
        \caption{Match scanning plan after Optimization.}
        \label{fig:graph-plan-optimized}
    \end{subfigure}
    \begin{subfigure}[b]{0.5\linewidth}
        \centering
        \includegraphics[width=\linewidth]{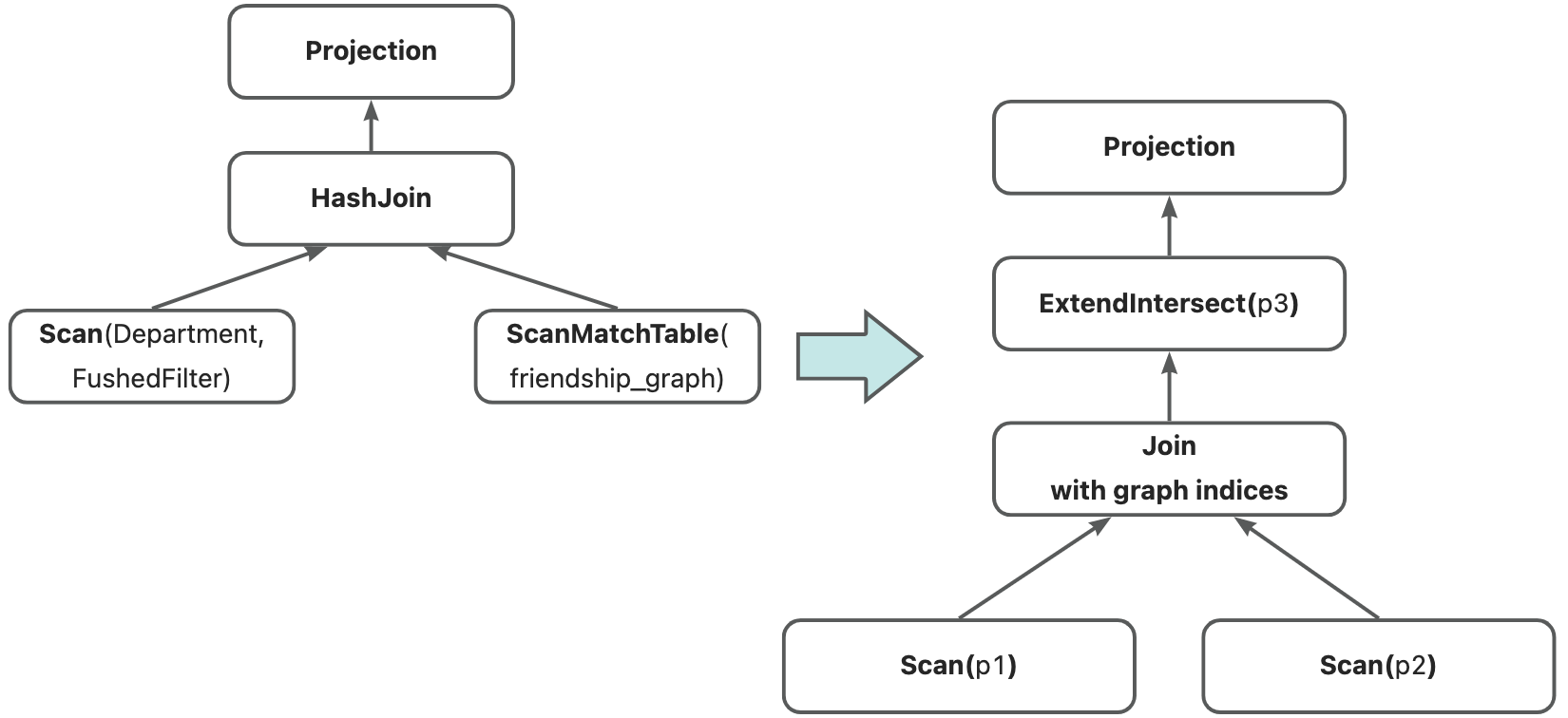}
        \caption{Obtained Optimal Physical Plan.}
        \label{fig:physical-plan-optimized}
    \end{subfigure}
    \caption{An example of query optimization.}
    \label{fig:query-grtree-example}
\end{figure*}

The above process of query processing is illustrated with the following example.

\begin{example}
    \label{example:framework}
    Given a relational database with tables as follows,
    \begin{equation*}
        \begin{split}
            & \textit{Person = (\underline{id}, name, dept\_id)} \\
            & \textit{Knows = (\underline{id1}, \underline{id2})} \\
            & \textit{Department = (\underline{dept\_id}, dept\_name)}, \\
        \end{split}
    \end{equation*}
    suppose we are going to find three persons satisfying:
    (1) These three persons know each other;
    (2) At least two of them are from the department of computer science.
    The SPJM query can be illustrated as shown in Fig.~\ref{fig:example-operator-tree}.

    In relational matching algebra, the SPJM query can be expressed as follows:
    Firstly, to obtain the triangles, the pattern $\mathcal{P}_{\triangle}$ is
    \begin{lstlisting}
        (p1:Person)-[:Knows]-(p2:Person),
        (p2)-[:Knows]-(p3:Person),
        (p1)-[:Knows]-(p3)
    \end{lstlisting}
    Then, to get the relational table that records the three person and their departments, the algebra expression is
    \begin{equation*}
        \begin{split}
            \widehat{R}_{graph} = & \pi_{p1.name\rightarrow pn1, p1.dept\_id \rightarrow dept1,p2.name\rightarrow pn2, p2.dept\_id \rightarrow dept2,} \\
            & _{p3.name\rightarrow pn3, p3.dept\_id \rightarrow dept3}(\mathcal{M}(GR, \mathcal{P}_{\triangle})),
        \end{split}
    \end{equation*}
    where $GR$ is a graph relation with only one tuple, and each attribute of the tuple is a vertex or an edge.
    The vertices correspond to rows in table Person and edges correspond to rows in table Knows.

    Finally, to obtain the triangles of persons with at least two persons from the department of computer science, the algebra expression is
    \begin{equation*}
        \begin{split}
        \pi_{pn1, pn2, pn3}
        (& \sigma_{dept.dept\_name = \text{`Computer Science'}}( \\
        & dept \Join_{dept1=dept.dept\_id \land dept2=dept.dept\_id} \widehat{R}_{graph})).
        \end{split}
    \end{equation*}

    Based on the above algebra expressions, the match scanning plan is shown in Fig.~\ref{fig:converged-logical-plan-graph} and the outer query plan is shown in Fig.~\ref{fig:converged-logical-plan-relational}.

    Then, optimization modules in the optimization layer are applied to optimize the plans.
    The optimized plans are shown in Fig.~\ref{fig:relational-plan-optimized} and Fig.~\ref{fig:graph-plan-optimized}.
    Moreover, the finally obtained optimal physical plan is shown in Fig.~\ref{fig:physical-plan-optimized}.
\end{example}
}

%% file: sec-exp-2.tex
\section{Evaluation}
\label{sec:evaluation}


\subsection{Experimental Settings}
\label{sec:experiment-settings}

\noindent\textbf{Benchmarks.} Our experiments leverage two widely used benchmarks to assess system performance, as follows:

\stitle{LDBC SNB.}
We use $LDBC10$, $LDBC30$, and $LDBC100$ with scale factors of 10, 30, and 100, generated by the official LDBC Data Generator. These datasets were chosen because they can be accommodated in the main memory of a single configured machine.
    We select 10 queries from the LDBC Interactive workload for evaluation, denoted as $\text{IC}_{1, \ldots, 9, 11, 12}$, with $10$, $13$, and $14$ excluded since they involve either pre-computation or shortest-path that are not supported.
    To accommodate queries containing variable-length paths~\cite{graindb}, we followed~\cite{graindb} to slightly modify them by separating each query into multiple individual queries with fixed-length paths. Each of these modified queries is denoted with a suffix ``-$l$'', where $l$ represents the length of the fixed-length path. In addition, we carefully designed two sets of queries for the comprehensiveness of evaluation, including (1) $QR_{1\ldots 4}$ to test the effectiveness of \filterrule and \joinfuserule in \name, and (2) $QC_{1\ldots 3}$, comprising three typical patterns with cycles including triangle, square, and 4-clique, to assess the efficiency of $\expandintersect$ introduced in \refsec{physical-operators}.

\stitle{JOB.} The Join Order Benchmark (JOB)~\cite{job_snb} on Internet Movie Database (IMDB) is adopted. We select the variants marked with ``a'' of all JOB queries, referred to as $\text{JOB}_{1\ldots 33}$, without loss of generality. These queries are primarily designed to test join order optimization, with each query containing an average of $8$ joins.

The \revise{largest dataset (i.e., $LDBC100$) contains 282 million tuples in vertex relations and 938 million tuples in edge relations.
More detailed statistics of the datasets are available in the full version\cite{full-version}.}
We manually implement the queries using SQL/PGQ, which are presented in the artifact~\cite{artifact}.
Furthermore, we perform the \rgmapping process in a manner that allows the construction of the same graph index on the LDBC and JOB datasets used in GRainDB's experiments~\cite{graindb}.
Specifically, the EV-index and VE-index on potential edge relations are constructed on foreign keys and tables that depict many-to-many relationships.

\noindent\textbf{Compared Systems. }
To ensure a fair comparison, all systems except K\`uzu use DuckDB v0.9.2 as the relational execution engine, differing only in their optimizers.
Since GRainDB was originally implemented on an older version of DuckDB, we have reimplemented it on DuckDB v0.9.2, which offers improved performance over the original version. 
K\`uzu utilizes its own execution engine (v0.4.2) as a baseline of a graph database management system (GDBMS).

\stitle{DuckDB~\cite{duckdb}}: This system optimizes queries using the graph-agnostic approach, leveraging DuckDB's built-in optimizer as described in \refsec{relational-only}. It serves as the naive baseline for extending a relational database system to support \spjm.

\stitle{GRainDB~\cite{graindb}}: This system uses same optimizer as DuckDB but employs the graph index (\refsec{graph-index}) for query execution. It acts as the baseline to demonstrate that solely using graph index is insufficient for optimizing \spjm.


\stitle{Umbra~\cite{umbra2020vldb,umbra2020cidr}}: \revise{
    This system features an advanced hybrid optimizer capable of generating wco join plans. We obtained the Umbra executable from the authors and configured its parameters according to their recommendations for computing the execution plan. The execution plan is then executed on DuckDB\footnote{Notably, all Umbra's plans for the benchmark queries exclude the multiway-join operator, allowing for direct transformation into DuckDB's runtime.}, utilizing the graph index when applicable, as done in GRainDB. This helps demonstrate that even with an advanced relational optimizer and the addition of a graph index, it can still fall short in optimizing \spjm.
    }

    \stitle{\name}: This system optimizes queries using the converged optimizer presented in \refsec{converged} and utilizes the graph index for query execution. It demonstrates the full range of techniques introduced in this paper. There are some variants
of \name~ for verifying the effectiveness of the proposed techniques, which will be introduced in the corresponding experiments.

\stitle{K\`uzu~\cite{jin2023cidr}}:
\revise{
    This system is a GDBMS that adopts the property graph data model. We use it as a baseline to compare the performance gap between \name on relational databases and native graph databases.
}

\noindent\textbf{Configurations. }
Our experiments were conducted on a server equipped with an Intel Xeon E5-2682 CPU running at 2.50GHz and \revise{256GB} of RAM, with parallelism restricted to a single thread.
For a comprehensive performance analysis, each query from the LDBC benchmark was run 50 times using the official parameters, while each query from the JOB benchmark was executed 10 times. We report the average time cost for each query to mitigate potential biases.
We imposed a timeout limit of 10 minutes for each query, and queries that fail to finish within the limit are marked as \ot.


\subsection{Micro Benchmarks on RelGo}
\label{sec:experiment-opt}
In this subsection, we conducted three micro benchmarks to evaluate the effectiveness of \name,
including assessing the efficiency of the optimizer, testing its advanced optimization strategies, and examining its effectiveness in optimizing join order.

\noindent\textbf{Optimization Efficiency Evaluation.}
First, we assessed the optimization efficiency by comparing \name with GRainDB\cite{graindb}.
We tested their optimization time
and also evaluated the execution time for their optimized plans as a measure of the plan quality.
We considered end-to-end time as optimization time plus execution time.
We randomly selected two subsets of the LDBC and JOB queries, and conducted the experiments on $LDBC30$ and IMDB datasets. 

The results in Fig.~\ref{fig:exp-optimization} reveal that \name significantly outperforms GRainDB in terms end-to-end time, achieving an average speedup of $7.5\times$ on $LDBC30$ and $3.8\times$ on IMDB.
However, note that \name incurs a slightly higher optimization cost compared to GRainDB. Although \name theoretically has a narrower search space, as analyzed in \refsec{compare-search-space}, GRainDB benefits from DuckDB's optimizer, which includes very aggressive pruning strategies.
Despite the slightly higher optimization cost, \name generates superior optimized plans, surpassing GRainDB by an average of $9.7\times$ on LDBC30 and $4.3\times$ on IMDB in execution time.


For fair comparison, in the subsequent experiments, we evaluate the efficiency of different systems using the end-to-end time.

\comment{
\textcolor{blue}{
First, we assessed the optimization efficiency by comparing \name with Apache Calcite\cite{calcite} and GRainDB\cite{graindb}.
Calcite, utilizing  a volcano-style query optimizer that is graph-agnostic, explores the search space to determine the most efficient plan.
In contrast, \name employs a graph-aware optimizer, which explores in a significantly narrowed search space as analyzed in \refsec{compare-search-space}.
To validate our theoretical findings, we conducted comparisons between \name and Calcite.
To broaden our analysis, we also compared \name against GRainDB, another graph-agnostic optimizer but generally employs a greedy optimization approach.
We compared the optimization time of \name, Calcite, and GRainDB, and also evaluated the execution time for their optimized query plans as a measure of the plan quality.
We select two subsets of the LDBC and JOB queries without loss of generality, and conduct the experiments on $LDBC30$ and IMDB datasets respectively.
The results are shown in Fig.~\ref{fig:exp-optimization}.
}

The experimental results reveal \name is significantly faster in optimization compared to Calcite.
Furthermore, the execution times of the query plans optimized by \name are either faster or similar to those produced by Calcite.
For example, in the LDBC queries ($IC[7]$ excluded) shown in \reffig{exp-opt-ldbc}, the optimization time cost of \name is about two orders of magnitude faster than that of Calcite, and the execution time of the plans optimized by \name is more than 52.0$\times$ faster, on average.
It is worth noting that for query $IC[7]$ and all queries in JOB, Calcite cannot even finish the optimization within the timeout limit (thus we omit the results of Calcite in \reffig{exp-opt-job}).
On the contrary, \name shows an efficient optimization process for these queries.
This stark contrast accentuates \name's capability to efficiently generate optimized plans, confirming the effectiveness of its graph-aware optimizer in reducing the search space and improving the optimization efficiency.

Then we move to the comparison between \name with GRainDB. Compared with GRainDB, \name demonstrates a slightly higher optimization cost. This outcome is expected to some extent, given GRainDB's design focuses on greedy, faster optimization methods.
Nevertheless, the quality of \name's optimized plans are usually superior. This can be observed in the execution time results, where \name outperforms GRainDB by $9.5\times$ and $2.8\times$ on average on the $LDBC30$ and IMDB datasets, respectively.
For the total time cost of optimization and execution, \name still achieves $6.5\times$ and $2.2\times$ speedup over GRainDB on average, on the $LDBC30$ and IMDB datasets, respectively.
These findings suggest that \name's optimization cost is justified, as it generates high-quality plans that are more efficient than those optimized by GRainDB in execution.

To ensure a fair comparison, in the subsequent experiments, we evaluated the efficiency of different systems by considering the total time for both optimization and execution.
}

\begin{figure}[t]
    \centering
    \begin{subfigure}[b]{.9\linewidth}
        \centering
        \includegraphics[width=\linewidth]{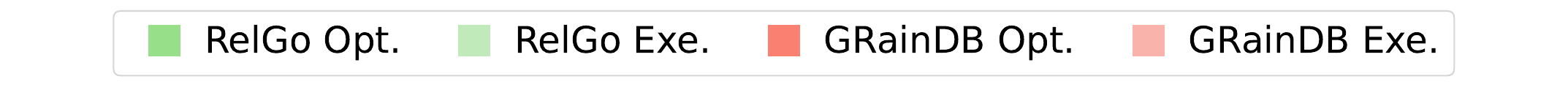}
        \label{fig:exp-opt-legends}
        \vspace*{-3.5ex}
    \end{subfigure}
    \begin{subfigure}[b]{0.45\linewidth}
        \centering
        \includegraphics[width=.8\linewidth]{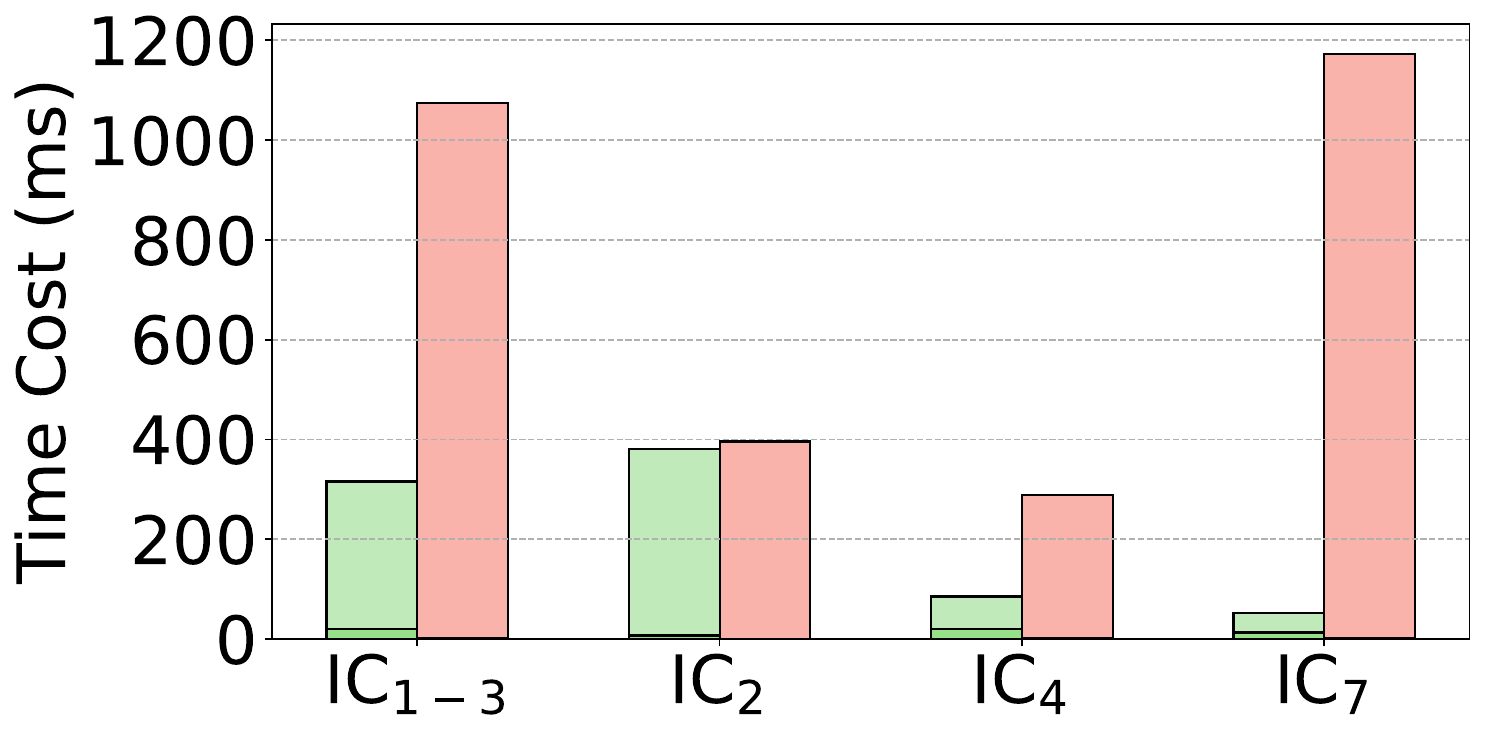}
        \caption{E2E Time on $LDBC30$.}
        \label{fig:exp-opt-ldbc}
    \end{subfigure}
    \begin{subfigure}[b]{0.45\linewidth}
        \centering
        \includegraphics[width=.8\linewidth]{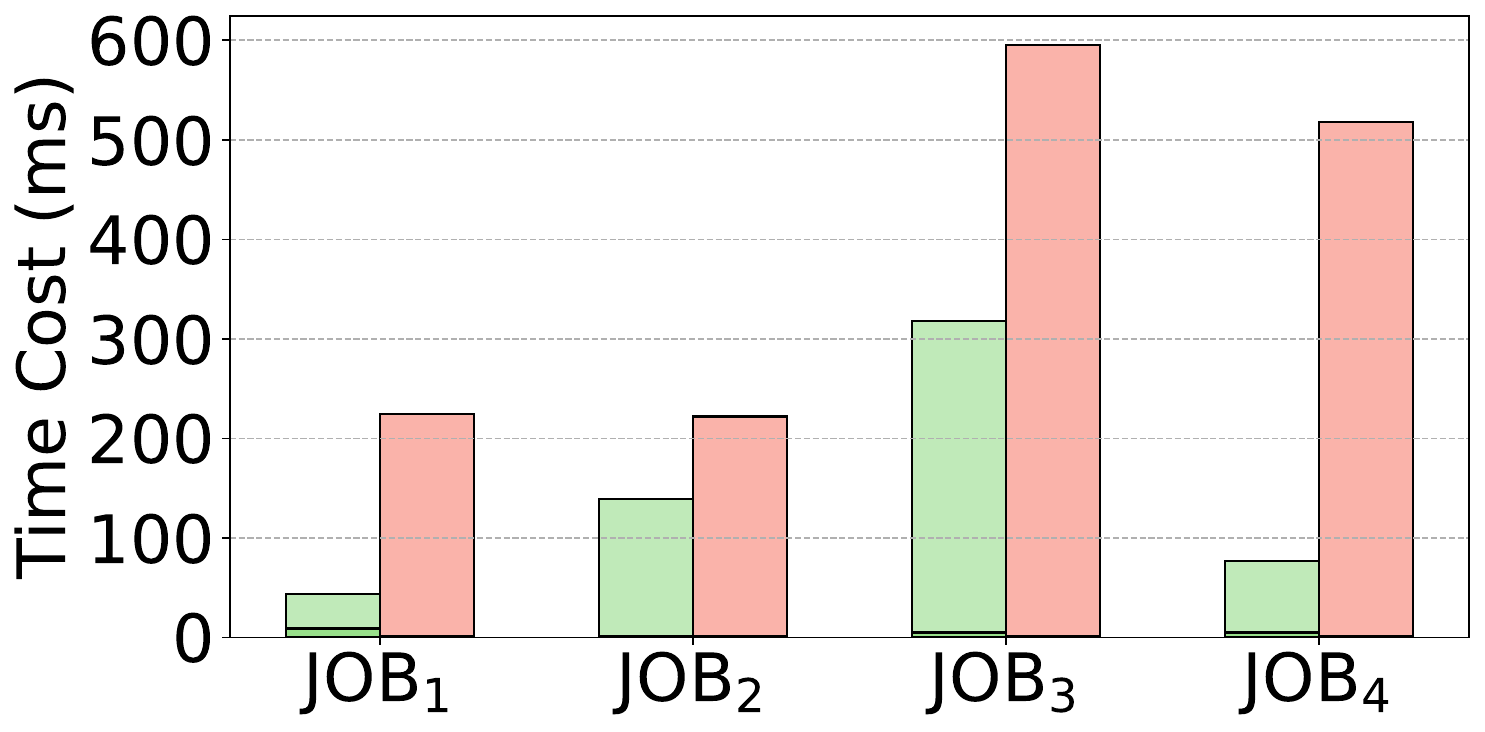}
        \caption{E2E Time on IMDB.}
        \label{fig:exp-opt-job}
    \end{subfigure}
    \caption{Experiments on optimization and execution cost}
    \vspace{-0.5em}
    \label{fig:exp-optimization}
\end{figure}



\begin{figure}[t]
    \centering
    \begin{subfigure}[b]{.45\linewidth}
        \centering
        \includegraphics[width=.8\linewidth]{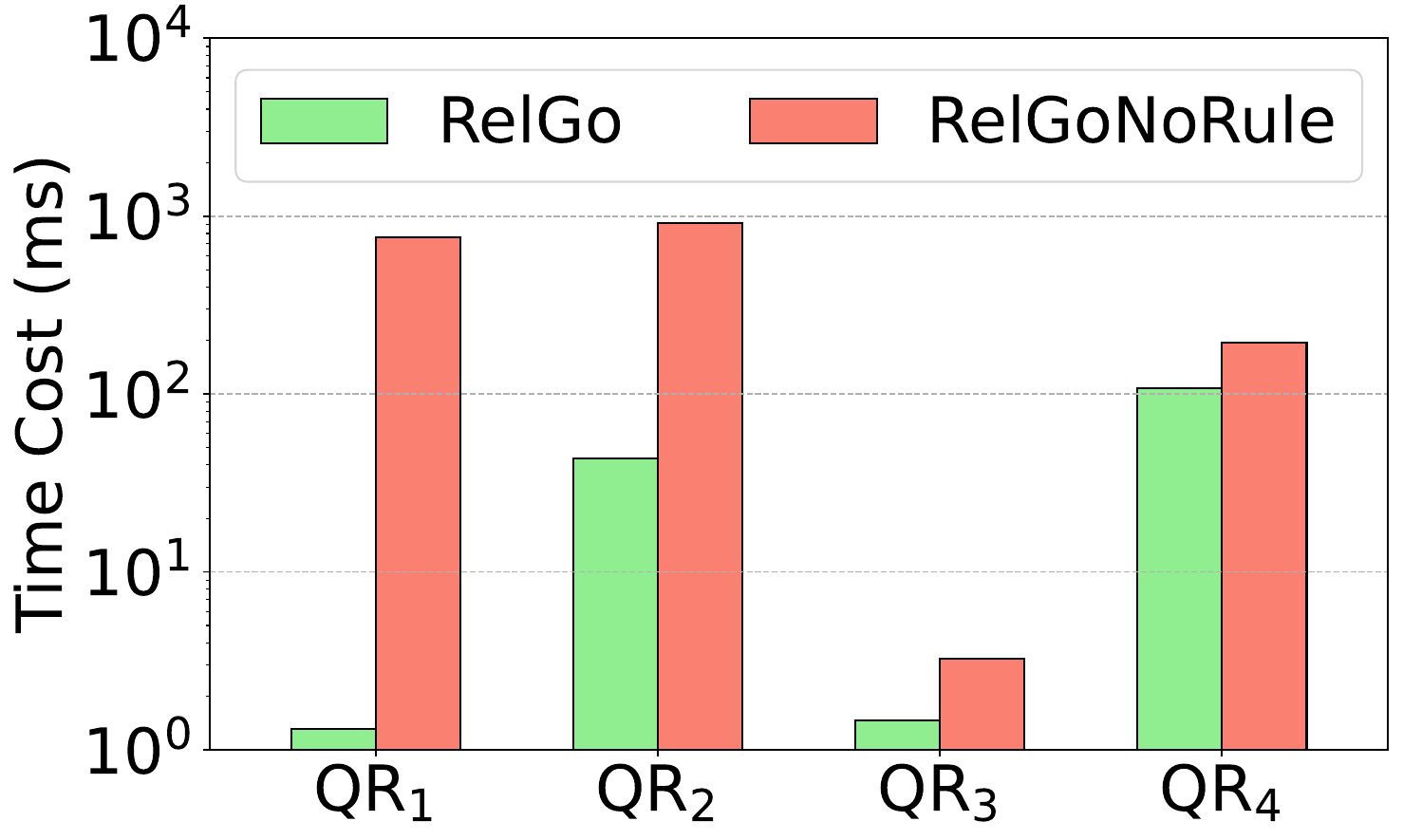}
        \caption{Time Cost on $LDBC10$.}
        \label{fig:exp-filter-sf10}
    \end{subfigure}
    \begin{subfigure}[b]{0.45\linewidth}
        \centering
        \includegraphics[width=.8\linewidth]{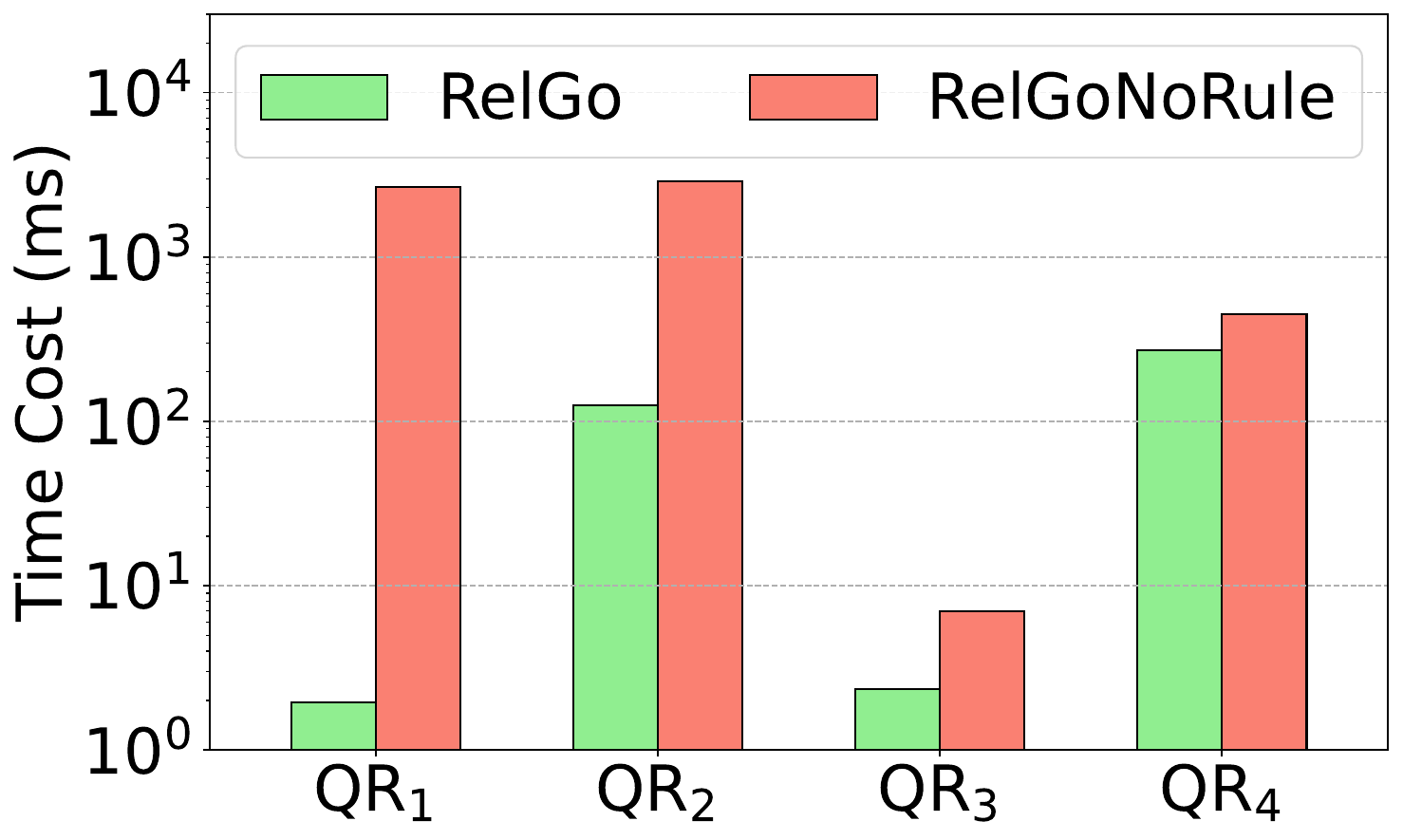}
        \caption{Time Cost on $LDBC30$.}
        \label{fig:exp-filter-sf30}
    \end{subfigure}
    \caption{Efficiency comparison of \name and \relgonofi}
    \label{fig:exp-filter}
\end{figure}

\begin{figure}[t]
    \centering
    \begin{subfigure}[b]{.45\linewidth}
        \centering
        \includegraphics[width=.8\linewidth]{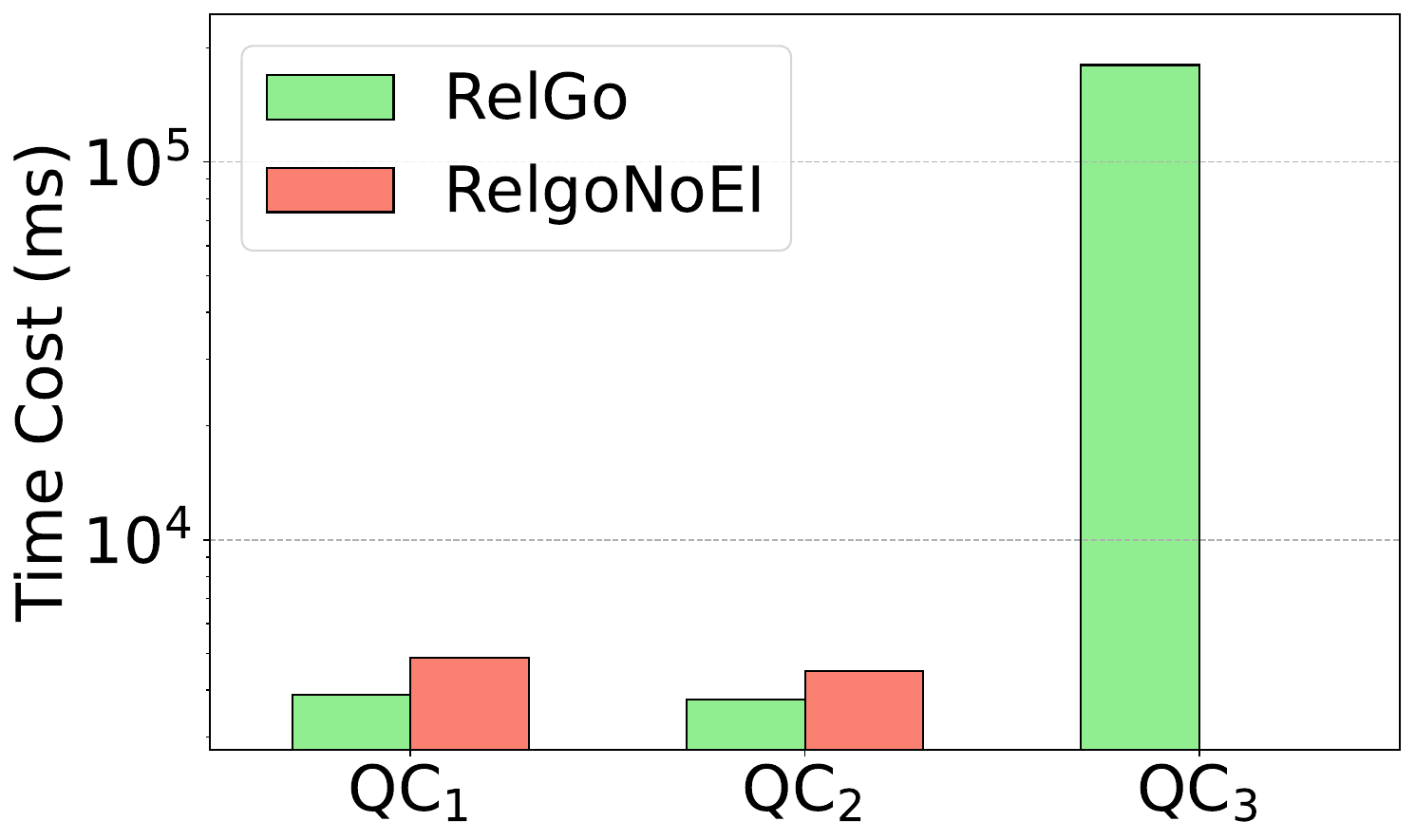}
        \caption{Time Cost on $LDBC10$.}
        \label{fig:exp-expand-intersect-sf10}
    \end{subfigure}
    \begin{subfigure}[b]{0.45\linewidth}
        \centering
        \includegraphics[width=.8\linewidth]{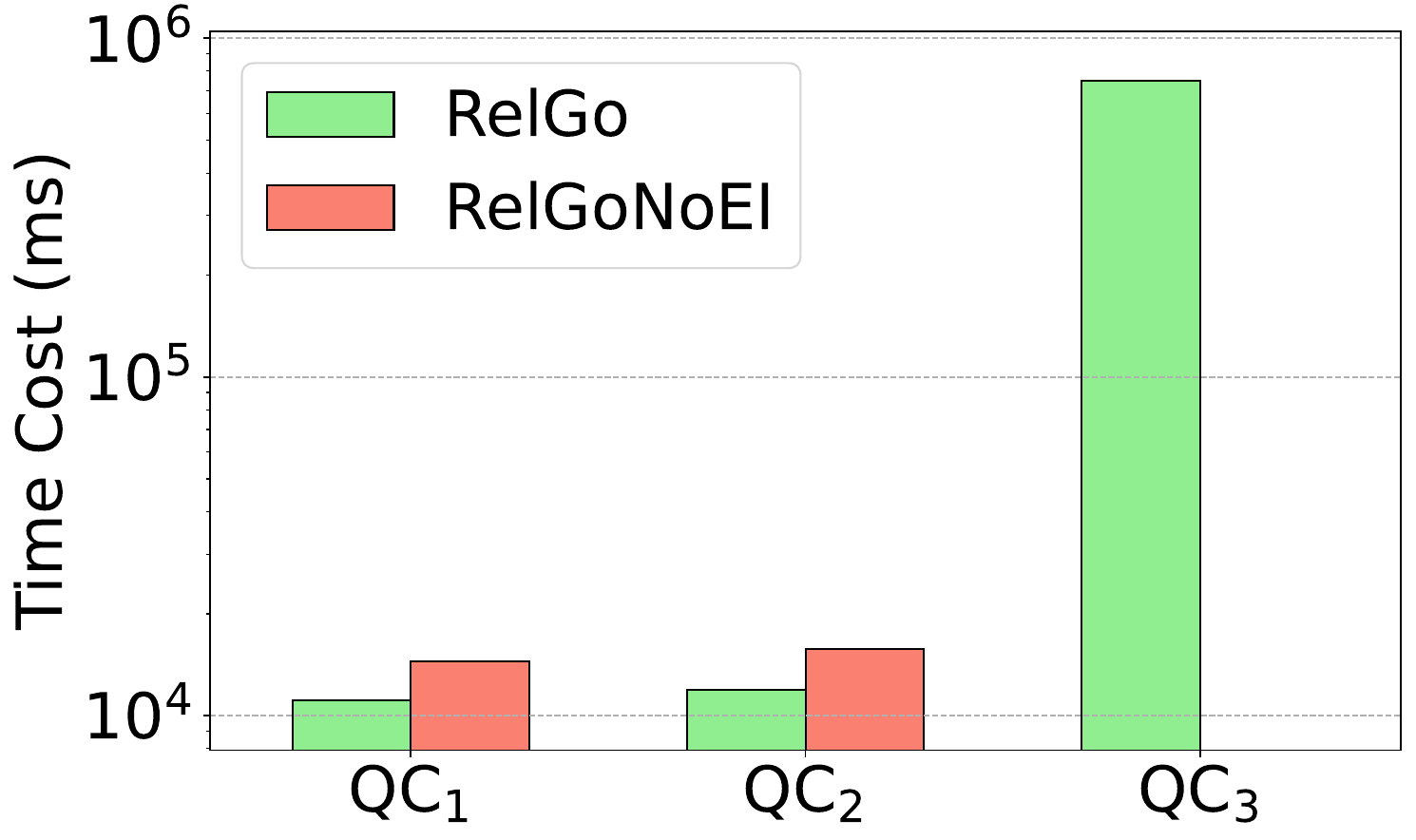}
        \caption{Time Cost on $LDBC30$.}
        \label{fig:exp-expand-intersect-sf30}
    \end{subfigure}
    \caption{Efficiency comparison of \name and \relgomj}
    \label{fig:exp-expand-intersect}
\end{figure}

\noindent\textbf{Advanced Optimization Strategies.}
In this experiment, we assessed the advanced optimization strategies in \name, including the heuristic \filterrule and \joinfuserule, and the optimized implementation of \expandintersect~ operator that aims to improve the efficiency of complete star join.

We began by testing heuristic rules \filterrule and \joinfuserule. 
We conducted experiments on $LDBC10$ and $LDBC30$, using $QR_1$ and $QR_2$ to test \filterrule, and $QR_3$ and $QR_4$ to test \joinfuserule. The results in \reffig{exp-filter} compared the performance of \name with and without applying these rules, denoted as \name and \relgonofi, respectively.
The results show that \filterrule significantly improves query performance, providing an average speedup of 299.4$\times$ on $LDBC10$ and 699.8$\times$ on $LDBC30$. With \joinfuserule, query execution is accelerated by an average of 2.0$\times$ on $LDBC10$ and 2.3$\times$ on $LDBC30$. These findings suggest that the heuristic rules, particularly \filterrule, are highly effective in enhancing query execution efficiency.


Next, we evaluated the effectiveness of the \expandintersect, which focuses on improving the efficiency of complete star join. Without this optimization strategy, the \expandintersect~ operator would be implemented as a traditional multiple join, and we denote this variant as \relgomj.
\revise{Queries $QC_{1 \ldots 3}$ that contain cycles are used to compare the performance of \name and \relgomj}.
The performance results in \reffig{exp-expand-intersect} suggest that, compared to \relgomj, \name achieves an average speedup of 1.22$\times$ on $LDBC10$ and 1.31$\times$ on $LDBC30$ (excluding $QC_3$). Notably, for $QC_3$, which is a complex 4-clique, the plans optimized by \relgomj confront an out-of-memory (OOM) error. 
The results indicate that \expandintersect~\\ with an optimized implementation not only enhances query performance but also significantly reduces the spatial overhead.

\begin{figure}[t]
    \centering
    \includegraphics[width=.75\linewidth]{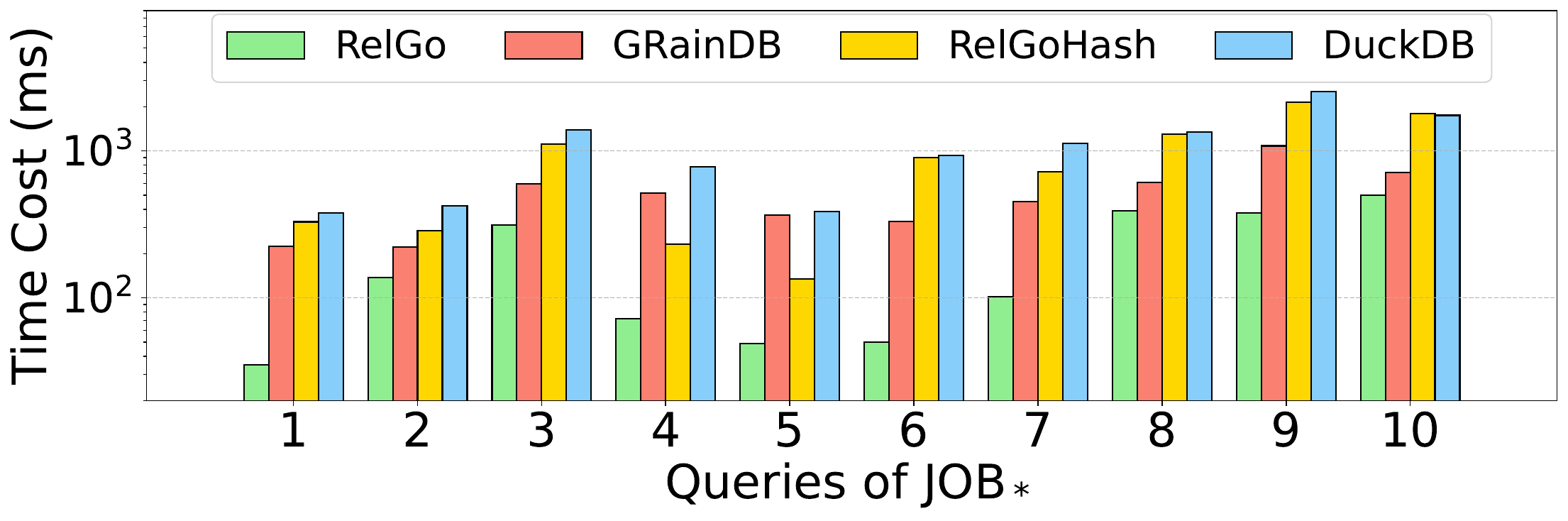}
    \vspace{-0.5em}
    \caption{Experiments on join order efficiency}
    \label{fig:exp-hash-plan}
\end{figure}

\noindent\textbf{Efficiency of Join Order.}
We compared \name with GRainDB and DuckDB, focusing on the efficiency of the join order. For this purpose, we introduced a variant of \name called \relgohash, which optimizes the plan in a converged manner like \name but deliberately bypasses the use of graph index. We selected 10 queries from the JOB benchmark and showed the performance results in \reffig{exp-hash-plan}.
The results demonstrate that \name outperforms GRainDB on all the queries, accelerating the execution time by factors ranging from $1.4\times$ to $7.5\times$, with an average speedup of $4.1\times$. Additionally, the plans optimized with \relgohash are at least as good as those optimized by DuckDB, achieving an average speedup of $1.6\times$. The effectiveness of \name and \relgohash stems from their use of advanced graph-aware optimization techniques in optimizing the matching operator, resulting in good join order and thus robust performance regardless of graph index.
It is worth noting that \name does not always generate the absolute best join orders, as it relies on the estimated cost of the plans. However, its optimized plans generally remain competitive in most cases, thanks to its integration of \glogue that use high-order statistics for cost estimation.

\begin{figure}[t]
    \centering
    \begin{subfigure}[b]{\linewidth}
        \centering
        \includegraphics[width=.8\linewidth]{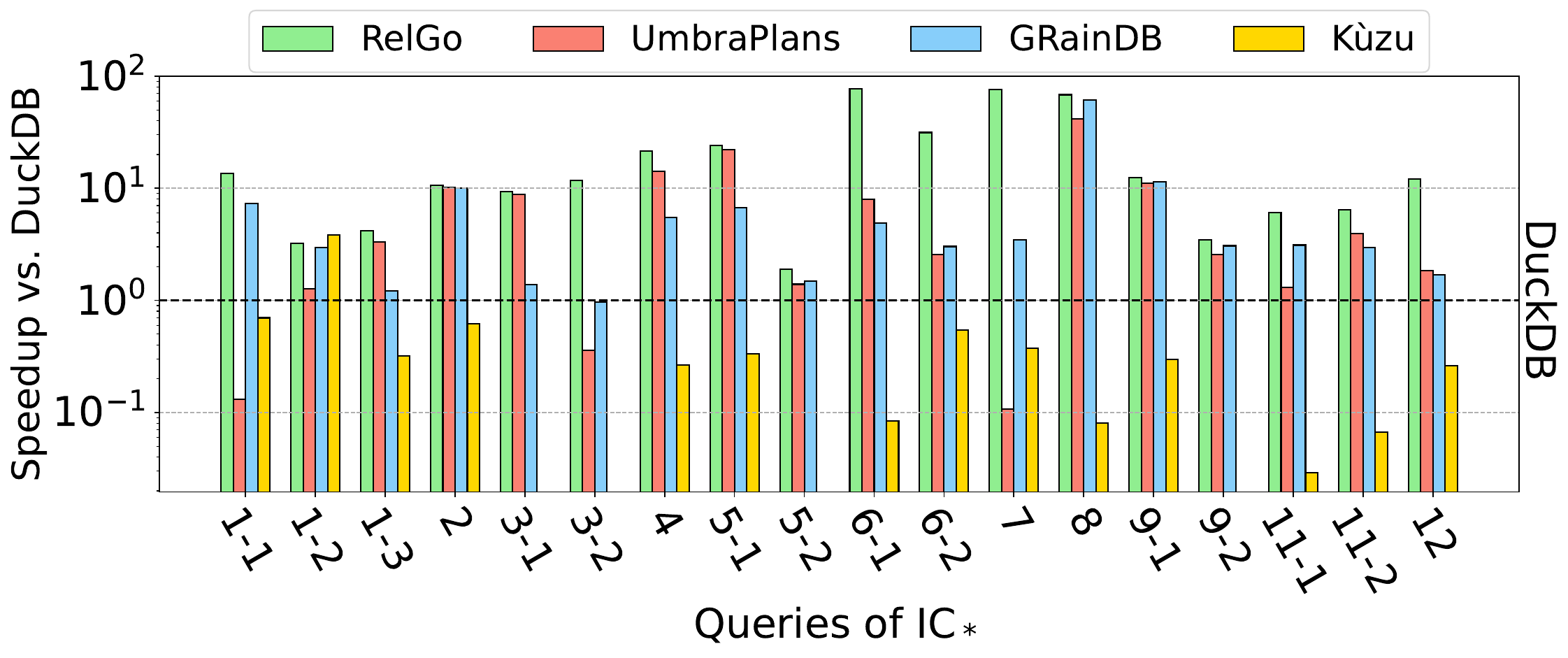}
        \vspace{-.5em}
        \caption{Speedup Compared to DuckDB on $LDBC100$.}
        \label{fig:exp-e2e-sf100}
    \end{subfigure}
    \begin{subfigure}[b]{\linewidth}
        \centering
        \includegraphics[width=.8\linewidth]{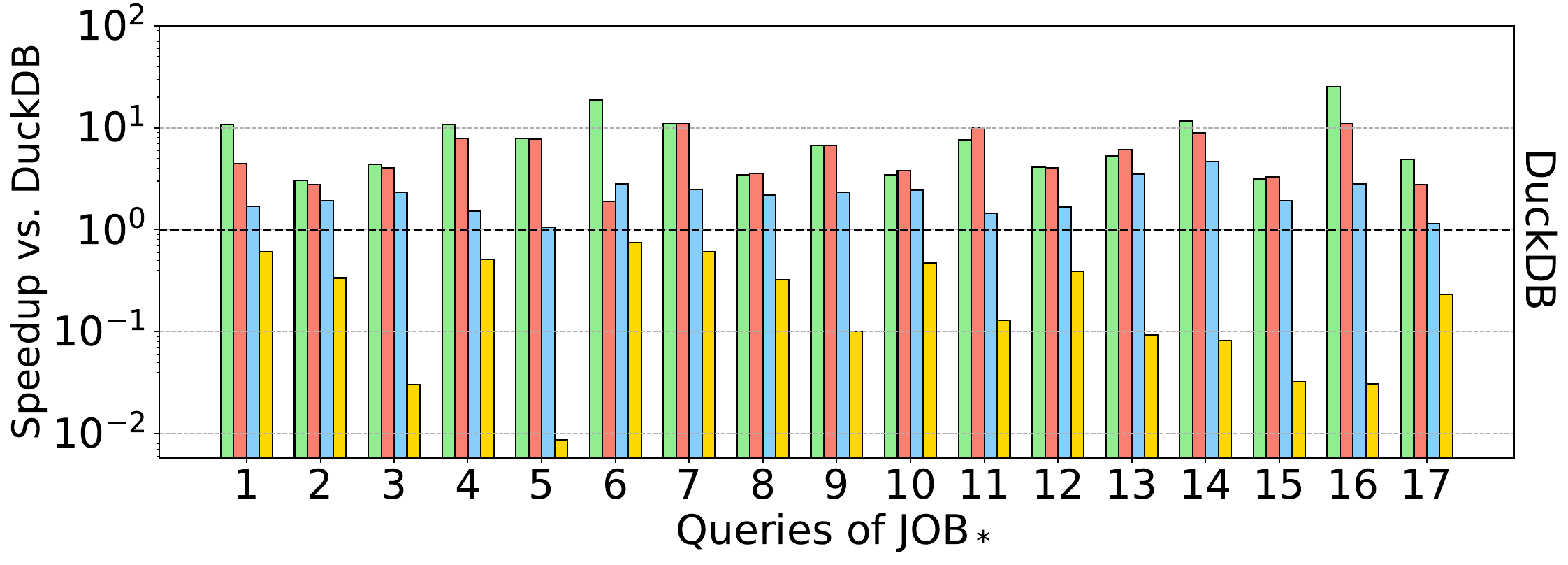}
    \end{subfigure}
    \begin{subfigure}[b]{\linewidth}
        \centering
        \includegraphics[width=.8\linewidth]{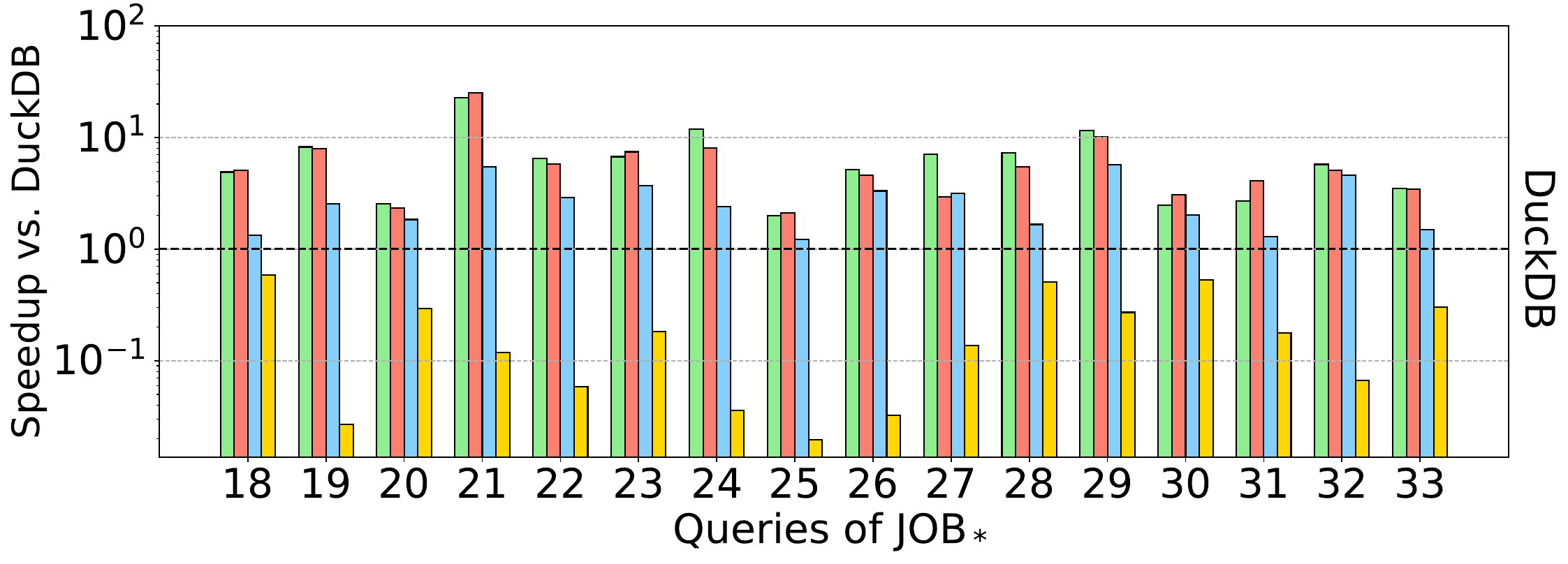}
        \vspace{-1ex}
        \caption{Speedup Compared to DuckDB on IMDB.}
        \label{fig:exp-e2e-job}
    \end{subfigure}
    \caption{Results of the comprehensive experiments. The speedup is computed as $\frac{\text{Time}(\text{DuckDB})}{\text{Time}(\text{Compared Method})}$. }
    \label{fig:exp-e2e}
    \vspace{-1em}
\end{figure}

\subsection{Comprehensive Experiments}
\label{sec:experiment-e2e}

We conducted comprehensive experiments on the LDBC and JOB benchmarks to comprehensively evaluate the performance of \name compared to \revise{DuckDB, GRainDB, Umbra, and K\`uzu}.
\revise{
The experimental results are shown in \reffig{exp-e2e}.
The results on LDBC10 and LDBC30 are omitted because they are comparable to those on LDBC100. Complete results are provided in the full version\cite{full-version}}.

\subsubsection{Comparison with DuckDB and GRainDB}

Firstly, we compared the performance of \name with DuckDB and GRainDB.
Specifically, for the LDBC benchmark, \revise{the execution time of the plans optimized by \name is about 21.9$\times$ and 5.4$\times$ faster on average than those generated by DuckDB and GRainDB on $LDBC100$}.
It is important to note that \name is especially effective for queries containing cycles, which can benefit more from graph optimizations. For example, in query $\text{IC}_{7}$, which contains a cycle, \name outperforms DuckDB and GRainDB by 76.3$\times$ and 22.0$\times$, respectively.
Conversely, the JOB benchmark, established for assessing join optimizations in relational databases, lacks any cyclic-pattern queries. Despite this, \name still achieves better performance compared to DuckDB and GRainDB, with an average speedup of 8.2$\times$ and 4.0$\times$, respectively.

The experimental results reflect our discussions in \refsec{graph-aware}. We summarize \name's superiority as follows.
First, \name is designed to be aware of the existence of graph index in query optimization and can leverage the index to effectively retrieve adjacent edges and vertices. In contrast, for GRainDB, relational optimizers can occasionally alter the order of \EVjoin operations, making graph index ineffective. DuckDB, on the other hand, does not consider graph index in query optimization and executes queries using conventional hash joins, which are often less efficient compared to graph-aware approaches.
Second, by incorporating a matching operator in \spjm queries to capture the graph query semantics, \name is able to leverage advanced graph optimization techniques to optimize matching operators. These techniques include using high-order statistics to estimate the cost of plans more accurately and employing wco join implementations to optimize cyclic patterns. In contrast, DuckDB and GRainDB cannot benefit from these graph-specific optimizations, which may lead to suboptimal plans and inefficient execution.
Third, \name considers optimization opportunities across both graph and relational query semantics, introducing effective heuristic rules such as \filterrule and \joinfuserule. These rules can significantly improve the efficiency of the generated plans.

\begin{figure*}[ht]
    \centering
    \includegraphics[width=\linewidth]{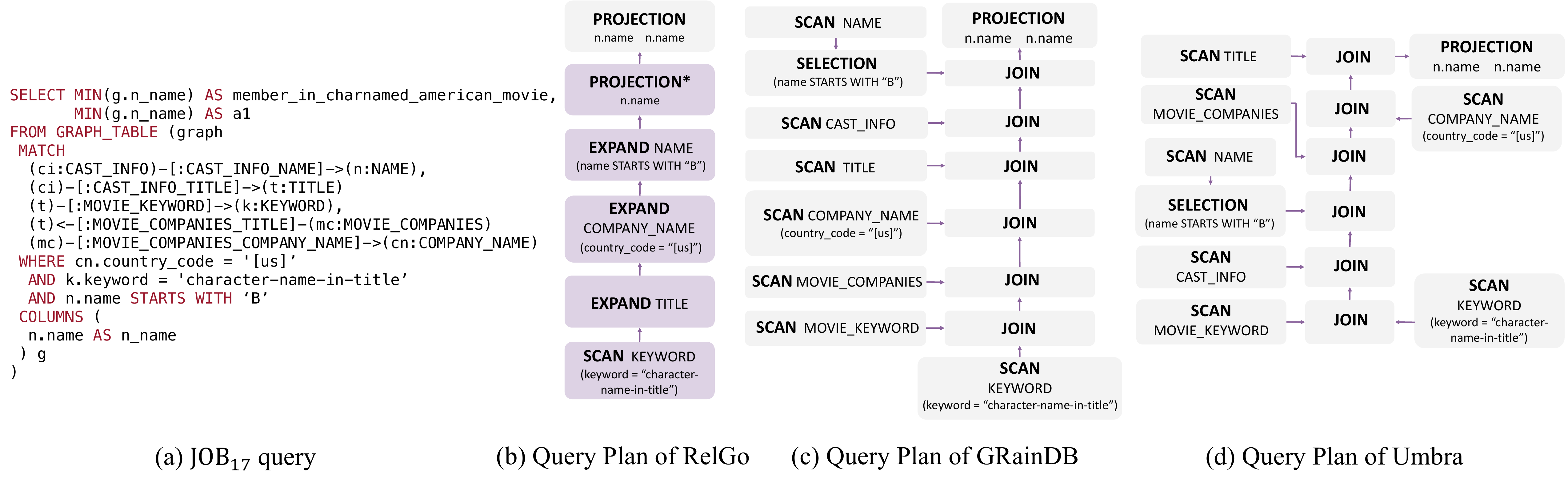}
    \caption{\revise{$\text{JOB}_{\text{17}}$'s plans given by RelGo, GRainDB and Umbra. JOINs are implemented as GRainDB's predefined joins if possible.}}
    \label{fig:job17a-query-plans}
\end{figure*}

\subsubsection{Comparsion with Umbra}
We \revise{then compared the performance of \name and Umbra.
In detail, the plans optimized by \name are about \revise{49.9$\times$} faster on average than those generated by Umbra on $LDBC100$.
On JOB benchmark, the plans generated by \name are on average 1.7$\times$ more efficient than those given by Umbra.} \revise{Several factors contribute to the results: (1) Umbra, due to its lack of a graph perspective, might generate query plans that encounter challenges in utilizing graph indexes effectively, similar to GRainDB; (2) Although Umbra's optimizer supports generating worst-case optimal plans that include multiway joins, none of Umbra's optimized plans for the tested queries in our experiments contained multiway joins. In contrast, \name excels at identifying opportunities to effectively utilize graph indices and adheres to worst-case optimality.
}

There are instances where Umbra outperforms \name in execution plans. For example, when querying $\text{JOB}_{30}$ on IMDB, the execution time of the plan generated by \name is approximately 1.2$\times$ slower than that of Umbra. A potential reason is that \name has not yet considered the distributions of attribute values. For example, when the predicate ``t.production\_year > 2000'' is present, knowing the distribution of the attribute ``production\_year'' can help better estimate the results after filtering by the predicate. Hence, Umbra can sometimes estimate cardinalities more accurately when such predicates exist. Addressing this will be an important future work. 

\comment{
The following reasons primarily contribute to this improvement:
(1) Umbra, due to its lack of a graph perspective, might generate query plans that could encounter issues in using the graph indexes similarly to GRainDB.
(2) Although Umbra's optimizer supports generating execution plans that include multiway joins, in our actual experiments, none of the optimized plans for the queries contained multiway joins. This indicates that Umbra also struggles to leverage the graph-specific optimizations.

\revise{Notably, there are instances where Umbra generates better execution plans compared to \name, e.g., when $\text{JOB}_{30}$ is queried on IMDB, the execution time of the plan generated by \name is about 1.2$\times$ that of Umbra. A possible reason is that the distributions of attribute values are not yet considered in \name and the Umbra can sometimes estimate cardinalities more accurately when predicates exist.
We consider this as our further work}.
}

\subsubsection{Comparison with K\`uzu}
\revise{Finally, we compared \name with the GDBMS, K\`uzu. The experimental results show that \name is approximately \revise{188.7$\times$} faster on average than K\`uzu on $LDBC100$ and 136.1$\times$ faster on the JOB benchmark.
Some results of K\`uzu are omitted (e.g., $\text{IC}_{3-1}$ on $LDBC100$) due to OOM errors.
As K\`uzu is also developed based on DuckDB, we speculated that K\`uzu may not sufficiently exploit graph-specific optimizations as \name does.}


\vspace{-2mm}
\subsection{Case Study}
To further illustrate why the plans generated by \name are superior to those produced by the baseline optimizers, we conducted a case study on $\text{JOB}_{17}$ as an example, shown in \reffig{job17a-query-plans}(a). The optimized query plans by \name, GRainDB, and Umbra for this query are presented in \reffig{job17a-query-plans}(b)-(d).
\reffig{exp-e2e-job} shows that \name's plan runs 4.3$\times$ and 1.8$\times$ faster than those optimized by GRainDB and Umbra, respectively.

A key difference between the plan of \name and those of GRainDB and Umbra is that \name can consistently follow the graph query semantics by continuously expanding from a starting vertex to its neighbors, leveraging the graph index. For example, \name's plan begins with scanning $\relation{\text{KEYWORD}}$, then expands to its neighbors $\relation{\text{TITLE}}$, followed by $\relation{\text{COMPANY\_NAME}}$, and finally $\relation{\text{NAME}}$. In this order, the graph indices (both EV-index and VE-index) introduced in \refsec{graph-index} are fully utilized to efficiently retrieve neighboring vertices. In contrast, GRainDB and Umbra, as relational optimizers, may not always adhere to this semantics. For instance, in GRainDB's plan, after joining $\relation{\text{KEYWORD}}$ with $\relation{\text{MOVIE\_KEYWORD}}$, the plan misses the opportunity to immediately join $\relation{\text{TITLE}}$, thus failing to use the EV-index constructed between $\relation{\text{MOVIE\_KEYWORD}}$ and $\relation{\text{TITLE}}$ right away. A similar situation occurs in Umbra's plan.




%% file: sec-relatedwork.tex
\vspace{-3mm}
\section{Related Work}
\label{sec:related-work}


\stitle{Query Optimization for Relational Databases.} 
Various studies of query optimization for relational databases were proposed to find the optimal join order \cite{IbarakiK84,optimize-nested-vldb-1986,Haffnerjoinorder,data-dependency-join}.
For example, 
Haffner et al.~\cite{Haffnerjoinorder} converted join order optimization into finding the shortest path on directed graphs and used the A* algorithm to solve it.
Kossmann et al.~\cite{data-dependency-join} summarized the methods to optimize queries with data dependencies, such as uniqueness constraints, foreign key constraints, and inclusion dependencies.
\revise{Recently, researchers attempt to incorporate wco joins into plans to better handle queries with cycles and reduce the size of intermediate results~\cite{Aberger2016Sigmod,free-join}.
CLFTJ~\cite{Kalinsky2017EDBT} introduces caching into trie join to reuse previously computed results.
Umbra~\cite{umbra2020vldb} proposes a new hash trie data structure and further reduces the cost of set intersection}.
All these techniques can be orthogonally adopted in \name's relational optimization.


\comment{
Some researchers \cite{selinger,postgres-row-estimation} estimate the number of cardinalities by computing the selectivity of $A \bowtie_{A.col_1 = B.col_2} B$ as
\begin{equation*}
    \frac{1}{max(DV(A.col_1), DV(B.col_2))},
\end{equation*}
where $DV(A.col_1)$ is the number of distinct values of $A.col_1$ in table $A$.
}

\stitle{Query Optimization for Graph Databases.} Graph pattern matching, a fundamental problem in graph query processing, has been extensively studied~\cite{angles2017foundations}. In sequential settings, Ullmann's backtracking algorithm~\cite{ullmann1976algorithm} has been optimized using various techniques, such as tree indexing~\cite{shang2008quicksi}, symmetry breaking~\cite{han13turbo}, and compression~\cite{bi2016efficient}. 
Join-based algorithms have been developed for distributed environments. These algorithms use cost estimation to optimize join order, with binary-join algorithms\cite{lai2015scalable, lai2019distributed} estimating costs using random graph models and worst-case-optimal join algorithms~\cite{ammar2018distributed} ensuring a worst-case upper bound on the cost. \revise{Hybrid approaches\cite{mhedhbi2019optimizing, huge, jin2023cidr} adaptively select between binary and wco joins based on the lower cost}. 
Recent studies have focused on improving cost estimation in graph pattern matching, including decomposing graphs into star-shaped subgraphs~\cite{cset} and comparing different cardinality estimation methods~\cite{gcare}. Some optimizers, like GLogS~\cite{GLogS}, search for the optimal plan by representing edges as binary joins or vertex-expansion subtasks. 
We follow the join-based methods such as~\cite{huge,GLogS} due to their compatibility with the relational context for which \name~ is designed. 

\stitle{Bridging Relational and Graph Models.} There is a growing interest in studying the interaction between relational and graph models.
DuckPGQ~\cite{DuckPGQ,DuckPGQ-VLDB} has demonstrated support for SQL/PGQ within the DuckDB~\cite{duckdb}, utilizing the straightforward, graph-agnostic approach to transform and process pattern matching.
\revise{Hence, DuckPGQ loses the opportunity to optimize the query from a graph query perspective}.
Index-based methods, such as GQ-Fast \cite{gqfast} and GRainDB \cite{graindb}, work towards construct graph-like index on relational databases to improve the performance of join execution. \name leveraged GRainDB's indexing technique for implementing physical graph operations. In contrast, methods like GRFusion \cite{GRFusion} and Gart~\cite{gart} work towards materializing graph from the relational tables,
so that graph queries can be executed directly on the materialized graph. Such methods incur additional storage costs and potential inconsistencies between relational and graph data.

%% file: sec-conclusion.tex
\section{Conclusions and Discussion}
\label{sec:conclusions}

In this paper, we introduce \name, a converged relational-graph optimization framework designed for SQL/PGQ queries. We formulate the \spjm query skeleton to better analyze and optimize the relational-graph hybrid queries introduced by SQL/PGQ. After discovering that a graph-agnostic approach can result in a larger search space and suboptimal query plans, we design \name to optimize the relational and graph components of \spjm queries using dedicated relational and graph optimization modules, respectively. Additionally, \name incorporates optimization rules, such as \filterrule, which optimize queries across the relational and graph components, further enhancing overall query efficiency.
We conduct extensive experiments comparing \name to graph-agnostic baselines, demonstrating its superior performance and confirming the effectiveness of our optimization techniques.


One \revise{interesting future direction is to extend \name to directly process existing \spj queries as inputs, enabling the automatic conversion from \spj to \spjm queries while being aware of the presence of graph indices. Boudaoud et al.~\cite{Boudaoud2022} may have discussed potential methods for such conversion. However, designing a global solution to determine which parts of an \spj query can be converted into a matching operator is challenging. This decision involves exhaustively exploring the search space, now including both join and pattern matching options. Given the high cost of optimizing joins alone, an exhaustive search could become prohibitively expensive. Therefore, it is necessary to carefully balance and select appropriate global and local optimization rules for given queries.}

